%% file: main.tex
\documentclass[12pt]{article}
\usepackage{graphicx}
\usepackage[dvipsnames]{xcolor}
\usepackage[section]{placeins}
\usepackage[top=1in, left=1in, right=1in, bottom=1in]{geometry}
\usepackage[parfill]{parskip}	
\usepackage{caption}
\usepackage{subcaption}
\usepackage[export]{adjustbox}
\usepackage{amssymb}
\usepackage{bbm}
\usepackage{bm}
\usepackage{amssymb}
\usepackage{amsmath}
\usepackage{amsfonts}

\usepackage{enumitem}
\usepackage{setspace}
\usepackage{booktabs}
\usepackage[compact]{titlesec}
\newcommand{\argmax}{\arg\!\max}

\usepackage{accents}
\newcommand{\vardbtilde}[1]{\tilde{\raisebox{0pt}[0.85\height]{$\tilde{#1}$}}}

\usepackage[nocitation]{apacite}
\usepackage{natbib}
\bibpunct{(}{)}{;}{a}{}{,}
\usepackage{hyperref}
\usepackage{titling}
\usepackage{pifont}
\usepackage[capposition=top]{floatrow}
\newcommand{\subtitle}[1]{%
  \posttitle{%
    \par\end{center}
    \begin{center}\large#1\end{center}
    \vskip0.5em}%
}
\usepackage{amsthm}
\usepackage{amsmath} 
\usepackage{multirow}
\usepackage{multicol}
\usepackage{mathrsfs}

\newenvironment{psmallmatrix}
  {\left(\begin{smallmatrix}}  
  {\end{smallmatrix}\right)}

\newtheorem{theorem}{Theorem}[section]
\newtheorem{lemma}{Lemma}
\newtheorem{proposition}[theorem]{Proposition}

\newtheorem{definition}{Definition}

\usepackage{mathrsfs}

\usepackage{color}

\usepackage{array}
\newcolumntype{C}[1]{>{\centering\arraybackslash}m{#1}}

\begin{document}
\pagestyle{plain}

\newtheoremstyle{mystyle}
{\topsep}
{\topsep}
{\it}
{}
{\bf}
{.}
{.5em}
{}
\theoremstyle{mystyle}
\newtheorem{assumptionex}{Assumption}
\newenvironment{assumption}
  {\pushQED{\qed}\renewcommand{\qedsymbol}{}\assumptionex}
  {\popQED\endassumptionex}
\newtheorem{assumptionexp}{Assumption}
\newenvironment{assumptionp}
  {\pushQED{\qed}\renewcommand{\qedsymbol}{}\assumptionexp}
  {\popQED\endassumptionexp}
\renewcommand{\theassumptionexp}{\arabic{assumptionexp}$'$}

\newtheorem{assumptionexpp}{Assumption}
\newenvironment{assumptionpp}
  {\pushQED{\qed}\renewcommand{\qedsymbol}{}\assumptionexpp}
  {\popQED\endassumptionexpp}
\renewcommand{\theassumptionexpp}{\arabic{assumptionexpp}$''$}

\newtheorem{assumptionexppp}{Assumption}
\newenvironment{assumptionppp}
  {\pushQED{\qed}\renewcommand{\qedsymbol}{}\assumptionexppp}
  {\popQED\endassumptionexppp}
\renewcommand{\theassumptionexppp}{\arabic{assumptionexppp}$'''$}

\renewcommand{\arraystretch}{1.3}

\newcommand\carl[1]{\cmnt{#1}{Carl}}
\newcommand\ambarish[1]{\cmnt{#1}{Ambarish}}
\newcommand\jose[1]{\cmnt{#1}{Jose}}

\newcommand{\argmin}{\mathop{\mathrm{argmin}}}
\makeatletter
\newcommand{\grande}{\bBigg@{2.25}}
\newcommand{\enorme}{\bBigg@{5}}

\newcommand{\blind}{0}

\newcommand{\tit}{
Balanced and Robust Randomized Treatment Assignments: The Finite Selection Model for the \\  Health Insurance Experiment and Beyond}

\if0\blind

{\title{\tit\thanks{We thank John Golden, Angela Lee, and Bijan Niknam for helpful research assistance and comments. We also thank participants at Euro-CIM 2023 for their valuable comments.
This work was supported through a grant from the Alfred P. Sloan Foundation (G-2020-13946).}}
\author{Ambarish Chattopadhyay\thanks{Stanford Data Science, Stanford University, 450 Jane Stanford Way Wallenberg, Stanford, CA 94305; email: \url{hsirabma@stanford.edu}.}, \and Carl N. Morris\thanks{Department of Statistics, Harvard University, 1 Oxford Street
Cambridge, MA 02138; email: \url{carl.morris@comcast.net}.},\and Jos\'{e} R. Zubizarreta\thanks{Departments of Health Care Policy, Biostatistics, and Statistics, Harvard University, 180 Longwood Avenue, Office 307-D, Boston, MA 02115; email: \url{zubizarreta@hcp.med.harvard.edu}.}
}

\date{} 

\maketitle
}\fi

\if1\blind
\title{\tit}
\date{} 
\maketitle
\fi

\begin{abstract}
\input{sec0}
\end{abstract}


\begin{center}
\end{center}
\clearpage
\doublespacing

\singlespacing
\pagebreak
\tableofcontents
\pagebreak
\doublespacing
\input{sec1}

\input{sec2}

\input{sec3}

\input{sec4}

\input{sec5}

\input{sec6}

\input{sec7}

\input{sec8}

\input{sec9}

\onehalfspacing
\bibliographystyle{asa}
\bibliography{mybibliography21}

\clearpage
\appendix
\section*{Supplementary materials} 

\input{sec_supp}

\end{document}

%% file: sec0.tex
\indent \hspace{.5cm} The Finite Selection Model (FSM) was developed by Carl Morris in the 1970s for the design of the RAND Health Insurance Experiment (HIE) (\citealt{morris1979finite}, \citealt{newhouse1993free}), one of the largest and most comprehensive social science experiments conducted in the U.S. 
The idea behind the FSM is that each treatment group takes its turns selecting units in a fair and random order to optimize a common assignment criterion. 
At each of its turns, a treatment group selects the available unit that maximally improves the combined quality of its resulting group of units in terms of the criterion. 
In the HIE and beyond, we revisit, formalize, and extend the FSM as a general tool for experimental design.

\hspace{.5cm} Leveraging the idea of D-optimality, we propose and analyze a new selection criterion in the FSM. The FSM using the D-optimal selection function has no tuning parameters, is affine invariant, and when appropriate, retrieves several classical designs such as randomized block and matched-pair designs. For multi-arm experiments, we propose algorithms to generate a fair and random selection order of treatments. We demonstrate FSM's performance in a case study based on the HIE and in ten randomized studies from the health and social sciences. On average, the FSM achieves 68\% better covariate balance than complete randomization and 56\% better covariate balance than rerandomization in a typical study. We recommend the FSM be considered in experimental design for its conceptual simplicity, efficiency, and robustness.

%% file: sec1.tex
\section{Introduction}
\label{sec_introduction}

\subsection{The RAND Health Insurance Experiment}

In the 1970's, the challenge of financing and delivering high-quality and affordable health care to all Americans was at the center of national policy debate. 
At the time, two central questions were ``How much more medical care would people use if it is provided free of charge?'' and ``What are the consequences of using more medical care on their health?''
To address these and other related questions, an interdisciplinary team of researchers led by Joseph P. Newhouse at RAND designed and conducted the Health Insurance Experiment (HIE), a large-scale, multi-year, randomized public policy experiment developed and completed between 1971 and 1982.
To this day, the HIE is one of the largest and most comprehensive social science experiments ever conducted in the U.S.
Even now, four decades after its completion, evidence from the HIE is still fundamental to the national discussion on health care cost sharing and health care reform.

In the HIE, a representative sample of 2,750 families comprising more than 7,700 individuals was chosen from six urban and rural sites across the United States.
At the beginning of the study, participants completed a baseline survey providing numerous demographic, medical, and socioeconomic measurements.
Families were then assigned to health insurance plans that varied substantially in their coinsurance rates and out-of-pocket expenditure maxima, for a total of 13 possible treatment groups. The goal of the study was to estimate the marginal averages of utilization and health outcomes in each of the six sites under each plan.

To provide the strongest possible evidence on health utilization and outcomes, the study had to be randomized. 
However, achieving balance for numerous continuous and categorical baseline covariates through randomization is challenging in experiments with so many treatment groups and different implementation sites.
In the HIE the groups had to be balanced and representative of the sites.
In the health and social sciences, there is an ever-increasing need for methods for random assignment of units into multiple treatment groups that are balanced, efficient, and robust.

\subsection{Toward balanced, efficient, and robust experimental designs}
 
Randomized experiments are considered to be the gold standard for causal inference, as randomization provides an unequivocal basis for inference and control. 
In randomized experiments, the act of randomization ensures balance on both observed and unobserved covariates \textit{on average}.
However, a given realization of the random assignment mechanism may produce substantial imbalances on one or more covariates. 
This imbalance problem can be exacerbated in settings like the HIE,
where treatments are multi-valued and many baseline covariates exist, leading to loss in efficiency of the effect estimates.

A variety of methods have been proposed in the literature to address this problem, such as blocking (\citealt{fisher1925statistical}, \citealt{fisher1935design}, \citealt{cochran1957experimental}), optimal pair-matching (\citealt{greevy2004optimal}), greedy pair-switching (\citealt{krieger2019nearly}),
and designs using mixed-integer programming (\citealt{bertsimas2015power}). 
In particular, rerandomization (\citealt{morgan2012rerandomization}) has gained popularity over the last few years and has become commonplace in experiments. 
However, rerandomization may not protect against and be robust to chance imbalances in functions of the covariates that are not explicitly addressed by the rerandomization criterion \citep{banerjee2017decision}, especially in experiments with multi-valued ($>$2) treatments. Moreover, defining the rerandomization criterion requires the selection of a tuning parameter governing the acceptable degree of imbalance, which may be difficult to choose and require iteration in practice. 

To address these and other related challenges, we revisit and extend the Finite Selection Model (FSM) for experimental design. The original version of the FSM was proposed and developed by Carl N. Morris in the design of the HIE (\citealt{morris1979finite}, \citealt{newhouse1993free}, \citealt{morris1993the}).
The idea behind the FSM is that each treatment group takes turns in a fair and random order to select units from a pool of available units such that, at each stage, each treatment group selects the unit that maximally improves the combined quality of its current group of units. 
The criterion for measuring quality is flexible. 
Among other contributions, in this paper we develop a new criterion based on D-optimality, which does not require tuning parameters.

To illustrate, Figure \ref{fig:simu_boxplot} exhibits the performance of complete randomization, rerandomization, and the FSM in a version of the HIE data with four treatment groups and 20 covariates.
For rerandomization, we compute the maximum Mahalanobis distance (across all pairs of treatment groups) based on the 20 covariates and their squares and pairwise products (i.e., all second-order transformations), and following \cite{lock2011rerandomization}, accept 0.1\% of the assignments with the smallest covariate distance (see Sections \ref{sec_hiedata} and \ref{sec_intuition} for details). 
The figure displays the distribution of absolute standardized mean differences (ASMD; \citealt{rosenbaum1985constructing})\footnote{The absolute standardized mean difference for a single covariate $X$ between treatment groups $g$ and $g'$ is $\text{ASMD}(X) = {|\bar{X}_g - \bar{X}_{g'} |}/{\sqrt{(s^2_g + s^2_{g'})/{2}}}$, where $\bar{X}_g$ and $s^2_g$ are the mean and variance of $X$ in treatment group $g$, respectively. 
Please see \cite{rosenbaum1985constructing}) for details.} in covariates and the second-order transformations across multiple realizations of the randomization mechanisms for the three designs. 
Lower values of ASMD indicate better balance on the covariates or their transformations.
Better balance can improve the validity and credibility of a study, and can also translate into increased efficiency and robustness.

\begin{figure}[!ht]
\centering
\includegraphics[scale =0.55]{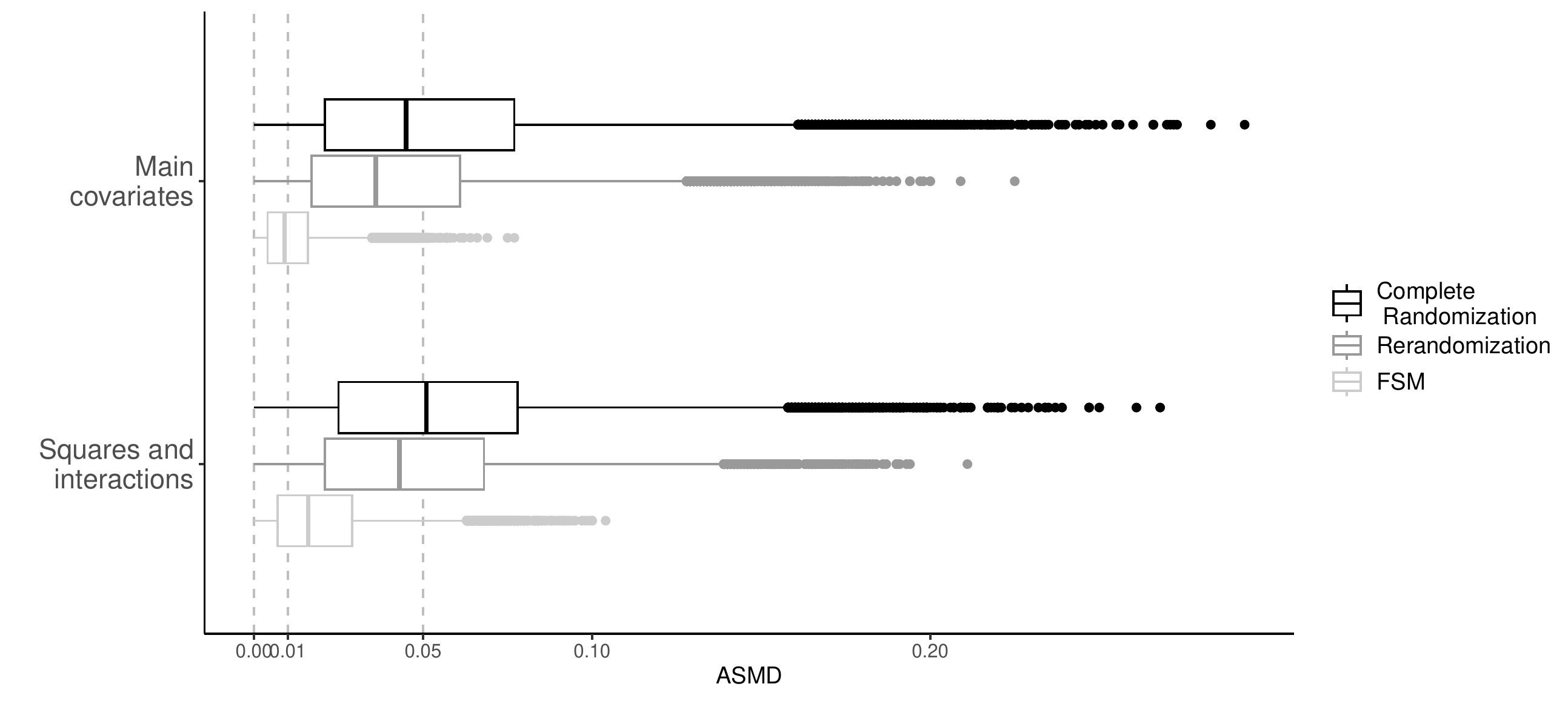}
\caption{Distributions of ASMD for complete randomization, rerandomization, and the FSM, for 20 baseline covariates in the HIE data. Without tuning parameters, the FSM handles multiple ($>$2) treatment groups and substantially improves covariate balance and, thereby, statistical efficiency.}
\label{fig:simu_boxplot}
\end{figure}

We observe that, as expected, rerandomization outperforms complete randomization in terms of imbalances on the main covariates and the second-order transformations. The FSM, however, markedly outperforms both methods for both types of covariates without requiring tuning parameters.
This analysis reveals that, while rerandomization performs well by common covariate balance standards (the majority of the ASMD is smaller than 0.1), there is room for improvement.
As we explain in Section \ref{sec_thehie}, in experiments like the HIE, the space of possible assignments is vast, and the FSM can meaningfully improve the assignment of units into treatment groups to achieve better balance and efficiency.

In a nutshell, the FSM does better because it progressively randomizes units into treatment groups in a controlled manner towards a criterion that is common to all groups and robust against general outcome models. As we show in theory and in practice in sections \ref{sec_theD}, \ref{sec_thehie}, and \ref{sec_additional} the FSM is a flexible tool for random assignment in various settings.


\subsection{Contribution and outline}

In this paper, we revisit, formalize, and extend the FSM for experimental design.
We show that the FSM can be used for balanced, efficient, and robust random treatment assignments, outperforming common assignment methods on these three dimensions.
In particular, we re-introduce the FSM under the potential outcomes framework (\citealt{neyman1923application}, \citealt{rubin1974estimating}).
We use the sequentially controlled Markovian random sampling (SCOMARS, \citealt{morris1983sequentially}) algorithm to determine the selection order of treatments for two-group experiments and extend it to multi-group experiments.
We propose a new selection criterion for treatments based on the idea of D-optimality and discuss its theoretical properties.
Under suitable conditions, we show that the FSM retrieves several classical experimental designs, such as randomized block and matched-pair designs. 
We explain model-based approaches to inference under the FSM and develop randomization-based alternatives. 
We analyze the FSM's performance empirically and compare it to common assignment methods. 
Finally, we discuss potential extensions of the FSM to more complex experimental design settings, such as stratified experiments and experiments with sequential arrival of units. 
In an accompanying paper \citep{chattopadhyay2021randomized}, we describe how these methods can be implemented in the new \texttt{FSM} package for \texttt{R}, which is publicly available on CRAN.

The paper proceeds as follows. In Section \ref{sec_hiedesign}, we describe the design of the RAND Health Insurance Experiment, focusing on the assignment of each family to one of 13 health insurance plans. 
In Section \ref{sec_foundations}, we present the setup, notation, and main components of the FSM. 
In Section \ref{sec_theD}, we propose a selection criterion based on D-optimality and analyze its properties. 
In Section \ref{sec_inference}, we discuss inference under the FSM. 
In Section \ref{sec_thehie}, we evaluate the performance of the FSM and compare it to standard methods such as complete randomization and rerandomization using the HIE data. In Section \ref{sec_additional}, we perform a similar comparison using the data from ten experimental studies from the health and social sciences. 
Finally, in Section \ref{sec_practical} we consider extensions of the FSM to other settings such as multi-group, stratified, and sequential experiments.
In Section \ref{sec_summary}, we conclude with a summary and remarks.
In the Online Supplementary Materials, we present all the proofs of the propositions and theorems, extended theoretical results, further empirical results based on a simulation study, and supplemental experimenal results on the HIE study and the ten case studies.


%% file: sec2.tex
\section{Design of the Health Insurance Experiment}
\label{sec_hiedesign}
In the HIE, families were assigned to different health insurance plans using the original version of the FSM. 
Initially, assignments were made in each of the six HIE sites to 12 or 13 fee-for-service plans with varying combinations of coinsurance (cost sharing) rates and income-related deductibles. Coinsurance plans consisted of $0\%$ (free care), $25\%$, $50\%$, or $95\%$ coinsurance rates, plus a plan with mixed coinsurance rates, and an individual deductible plan. Within the cost sharing plans, families were further assigned to different out-of-pocket maxima where the out-of-pocket expenditures were capped at 5\%, 10\%, or 15\% of family income, with an annual maximum of \$1,000 \citep{brook2006health}. 
To ensure that the resulting treatment groups were balanced relative to the population of each site, the FSM considered a discard group of study non-participants as an additional treatment group.

Listed in chronological order of study initiation, the following sites were tracked for several years: Dayton, OH; Seattle, WA; Fitchburg, MA; Franklin County, MA; Charleston, SC; and Georgetown County, SC.  The FSM was used, independently in each of the sites, to make random assignments to improve balance on up to 22 family-level baseline covariates across treatment groups. In each of the first two sites, the FSM was used multiple times for separate independent subsets of families to maintain baseline data schedules.
In addition to estimating the overall marginal effects of health insurance plan design on healthcare utilization and outcomes, the HIE team also sought to understand how the experimental results were affected by particular design choices, e.g., longer versus shorter enrollment duration, receiving versus not receiving participation incentives, 
higher versus lower interviewing frequency. To this end, four additional sub-experiments were conducted, and the FSM was used to randomize families to the sub-treatment groups.  

%% file: sec3.tex
\section{Foundations and overview of the FSM}
\label{sec_foundations}
\subsection{Setup and notation}
\label{sec_setup}

Consider a sample of $N$ units indexed by $i = 1, ..., N$.
Each of these units is to be assigned into one of $G$ treatment groups labeled by $g$, with $g = 1, ..., G$.
Write $n_g$ for the pre-specified size of group $g$.  
Denote $Z_i \in \{1, 2, ..., G\}$ as the assigned treatment group label of unit $i$ and $\bm{Z} = (Z_1,...,Z_N)^\top$ as the vector of treatment group labels. 
Following the potential outcomes framework for causal inference \citep{neyman1923application, rubin1974estimating}, each unit $i$ has a potential outcome under each treatment $g$, $Y_i(g)$, but only one of these outcomes is observed: $Y^{\text{obs}}_i = \sum_{g = 1}^{G} \mathbbm{1}(Z_i = g) Y_i(g)$. Denote $\bm{Y}(g) = (Y_1(g),...,Y_N(g))^\top$ as the vector of potential outcomes under treatment $g$.
Each unit has a vector of $K$ observed covariates, $\boldsymbol{X}_{i}$.
We write $(\underline{\bm{X}}_{\text{full}})_{N \times k}$ for the matrix of observed covariates, and $\bar{\bm{X}}_{\text{full}}$ and  $\underline{\bm{S}}_{\text{full}}$ for the mean vector and covariance matrix of these covariates in the full sample, respectively. Denote $(\underline{\tilde{\bm{X}}}_{\text{full}})_{N \times (k+1)}$ as the design matrix in the full sample.\footnote{The design matrix includes a column of all 1's (for the intercept) and $k$ columns of covariates.} We assume that $\underline{\tilde{\bm{X}}}_{\text{full}}$ has full column rank.
In Table \ref{tab_notation} of the Online Supplementary Materials we provide a list of the notation used in this paper.

Based on this notation, $Y_i(g') - Y_i(g'')$ is the causal effect of treatment $g'$ relative to treatment $g''$ for unit $i$.
We are interested in estimating the sample average treatment effect $\text{SATE}_{g',g''} = \frac{1}{N}\sum_{i=1}^{N} \{Y_i(g') - Y_i(g'') \}$ and the population average treatment effect $\text{PATE}_{g',g''} = \mathbb{E} \{Y_i(g') - Y_i(g'')\}$.
For this, we will randomly assign the units into treatment groups using the FSM.

\subsection{Components of the FSM}
\label{sec_components}

In the FSM, the $G$ treatment groups take turns selecting units in a random but controlled order while optimizing a common criterion. This is accomplished by the two components of the FSM, namely, the \textit{selection order matrix} and the \textit{selection function}.
\begin{enumerate}
 \item Selection order matrix (SOM): An SOM is a matrix 
that determines the order in which the treatment groups select the units. Typically, an SOM has two columns; the first specifies the stages of selection (from $1$ to $N$), and the second specifies the treatment group that
selects first at that stage.
    
 \item Selection function: A selection function is a function that determines which unit gets selected by the choosing treatment group at each stage. Typically, a selection function is based on an optimality criterion that is common to all treatment groups.
\end{enumerate}

 A good SOM guarantees that the selection of units is fair, so that no single treatment group selects all the units of a given type, and random, so that both observed and unobserved covariates are balanced in expectation and there is a basis for inference. A good selection function will produce efficient and robust inferences under a wide class of possible outcome functions.

To illustrate, Table \ref{tab1}(a) presents an example data set with 12 observations and one covariate, age. We consider assigning these 12 units into two groups of equal sizes using the FSM. Table \ref{tab1}(b) shows an example of an SOM in this setting. The SOM determines the order in which each treatment selects a unit at each stage. In the example, treatment group 2 selects first in stage 1, treatment group 1 selects in stage 2, and so on. Treatment groups select units based on the selection function.

\begin{singlespacing}
\begin{table}[H]
\caption{(a) Example data set; (b) selection order matrix and an assignment using the FSM.}
\begin{subtable}{.45\linewidth}
\centering
\caption{\footnotesize Data set}
 \scalebox{0.65}{
\begin{tabular}{cc}
  \toprule
 Index & Age \\ 
  \hline
1 & 24 \\ 
  2 & 30 \\ 
  3 & 34 \\ 
  4 & 36 \\ 
  5 & 40 \\ 
  6 & 41 \\ 
  7 & 45 \\ 
  8 & 46 \\ 
  9 & 50 \\ 
  10 & 54 \\ 
  11 & 56 \\ 
  12 & 60 \\ 
  \hline
Mean & 43 \\
   \bottomrule
\end{tabular}
}
\end{subtable}
   \begin{subtable}{.45\linewidth}
    \centering
   \caption{\footnotesize Selection order matrix and assignment}
   \scalebox{0.65}{
   \begin{tabular}{ccccc}
            \toprule 
    \multicolumn{2}{c}{Selection order matrix} & \multicolumn{3}{c}{Unit selected}\\
   \cmidrule(r){1-2} \cmidrule(r){3-5} 
Stage    & Treatment & Index & Age\\
    \toprule
1 & 2 & 1 & 24 \\ 
  2 & 1 & 12 & 60 \\ 
  3 & 1 & 2 & 30 \\ 
  4 & 2 & 11 & 56 \\ 
  5 & 1 & 3 & 34 \\ 
  6 & 2 & 10 & 54 \\ 
  7 & 1 & 9 & 50 \\ 
  8 & 2 & 4 & 36 \\ 
  9 & 1 & 5 & 40 \\ 
  10 & 2 & 8 & 46 \\ 
  11 & 2 & 6 & 41 \\ 
  12 & 1 & 7 & 45 \\ 
  \bottomrule
  \end{tabular}
  }
\end{subtable}
\label{tab1}
\end{table}
\end{singlespacing}

In general, it is crucial that the order of selection is random, but that no group chooses in a disproportionate manner. 
For two treatment groups of arbitrary sizes, this can be accomplished by means of the Sequentially Controlled Markovian Random Sampling (SCOMARS) algorithm \citep{morris1983sequentially}. 
In the FSM, SCOMARS specifies the probability of a treatment group selecting at stage $r$ ($r \in \{1,2,...,N\}$), conditional on the number of selections made by that group up to stage $r-1$. See the Online Supplementary Materials for a formal description of the algorithm. SCOMARS satisfies the sequentially controlled condition (\citealt{morris1983sequentially}), which requires the deviation of the observed number of selections made by a treatment group up to stage $r$ from its expectation to be strictly less than one. Intuitively, this condition ensures that throughout the selection process, no treatment group departs too much from its expected fair share of choices. Moreover, SCOMARS is Markovian because for each group, the probability of selection at stage $r$ depends solely on the number of selections made up to stage $r-1$.
For two groups of equal sizes (as in the example in Table \ref{tab1}), generating an SOM under SCOMARS boils down to successively generating $N/2$ independent random permutations of the treatment labels $(1, 2)$.
In Section \ref{sec_multi} and in the Online Supplementary Materials, we describe this and other extensions of SCOMARS to multi-group experiments. 
Unless otherwise specified, in the rest of the paper, we will use SCOMARS to generate the SOM for experiments with two treatment groups.

The selection function gives a value to each of the units available for selection at each stage.
This value depends on the characteristics of each available unit in addition to those already assigned to the treatment group that selects next.
In principle, any criterion can be used in the selection function.
For example, if the selection function is constant, then the treatment group selects a unit randomly from the available pool.
Alternatively, the selection function can compute the contribution of each unit to a measure of the accuracy of the estimator. 
In this spirit, we propose the \textit{D-optimal} selection function, which, at each stage, minimizes the generalized variance of the estimated regression coefficients in a linear potential outcome model (see Section \ref{sec_theD} for details). 

To build intuition, in Table \ref{tab1}(b) we discuss the special case of $k=1$ covariate. 
With the D-optimal selection function, the choosing group, in its first choice, selects the unit whose covariate value is farthest from the full-sample mean of the covariate; and in the subsequent choices, selects the unit whose covariate value is farthest from its current mean of the covariate. 
In the example in Table \ref{tab1}, treatment $2$ selects unit $1$ with age $24$, the farthest age from the full-sample mean $43$. 
In the next stage, treatment $1$ selects unit $12$ with age $60$, the farthest age from $43$.\footnote{Notice that for treatment 1's first selection, the mean of age remains 43 (i.e., the full-sample mean of age) and is not recalculated based on the 11 unselected units.} 
Next, treatment $1$ selects unit $2$ with age $30$, the farthest age from its current mean age $60$. 
The process continues until all the 12 units are selected.

In general, with multivariate data, the FSM first selects the units that are farthest from the full-sample mean of the covariates and successively approaches this target, ultimately selecting the units that are closest to it.
In the FSM, the SOM produces balance out of an optimality criterion that is common to all the treatment groups.
This is crucial so that all the choosers know the same, and as they choose, they produce groups that are balanced and equally robust against the unknown outcome model.

Another important feature of the FSM is that, in addition to several treatment groups, it can accommodate a discard group of unassigned units.
This is important, for example, in settings where the number of available units for assignment is greater than the total number of units that can feasibly be assigned (e.g., because of budgetary constraints).
This feature of the FSM was used in the HIE to secure the representativeness of the treatment groups relative to the target populations.


%% file: sec4.tex
\section{The D-optimal selection function}
\label{sec_theD}



Here, we formalize the D-optimal selection function and provide an equivalent, closed-form characterization that explains how this criterion governs the selection of units at each stage. 
Without loss of generality, assume that treatment 1 selects at stage $r$, $r \in \{1,2,...,N\}$. 
Let $\tilde{n}_{r-1}$, $\bar{\bm{X}}_{r-1}$, $\underline{\bm{S}}_{r-1}$, and $\underline{\tilde{\bm{X}}}_{r-1}$ be the number, mean vector, covariance matrix, and the design matrix of the units selected after the $(r-1)$th stage by treatment 1, respectively. 

To define the selection function, we consider a linear potential outcome model of $Y_i(1)$ on $\bm{X}_i$, i.e., $Y_i(1) = \bm{\beta}^\top (1, \bm{X}_i^\top)^\top + \eta_i$, where $\eta_i$ is an error term satisfying $\mathbb{E}\{\eta_i|\bm{X}_i\} = 0$.\footnote{More generally, one can consider a linear model of $Y_i(1)$ on a vector of basis functions $\bm{B}(\bm{X}_i)$ of the covariates.} Denote $\mathcal{R}_{r-1}$ as the set of unselected units after stage $r-1$. For unit $i \in \mathcal{R}_{r-1}$, let $\underline{\tilde{\bm{X}}}_{r,i}$ be the resulting design matrix in treatment group 1 if unit $i$ is selected. When $\underline{\tilde{\bm{X}}}^\top_{r-1} \underline{\tilde{\bm{X}}}_{r-1}$ is invertible, the D-optimal selection function selects unit $i' \in \mathcal{R}_{r-1}$, where $i' \in \argmax\limits_{i \in \mathcal{R}_{r-1}}\det(\underline{\tilde{\bm{X}}}^\top_{r,i}\underline{\tilde{\bm{X}}}_{r,i})$. In other words, at the $r$th stage, the D-optimal selection function chooses the unit in $\mathcal{R}_{r-1}$ that optimally decreases the generalized variance of the estimated regression coefficients of the fitted linear model in treatment 1. Ties in the values of the generalized variances are resolved randomly. When $\underline{\tilde{\bm{X}}}^\top_{r-1} \underline{\tilde{\bm{X}}}_{r-1}$ is not invertible, we define the D-optimal selection function by using a form of Ridge augmentation (see Lemma \ref{lemma:dopt} in the Online Supplementary Materials). 
The following theorem provides an equivalent characterization of the D-optimal selection function that elucidates the selection made by the choosing treatment group at each stage. 
\begin{theorem}\normalfont
Assume treatment 1 chooses at stage $r$. 
Then the D-optimal selection function chooses unit $i'$ such that 
\begin{equation*}
    i' \in \argmax\limits_{i \in \mathcal{R}_{r-1}} (\bm{X}_i - \bar{\bm{X}}^*_{r-1})^\top (\underline{\bm{S}}^*_{r-1})^{-1} (\bm{X}_i - \bar{\bm{X}}^*_{r-1}),
\end{equation*} 
\vspace{-1.5cm}
\begin{singlespacing}
where 
{\small
\begin{equation*}
  \bar{\bm{X}}^*_{r-1} =
    \begin{cases}
       \bar{\bm{X}}_{\text{full}} & \text{if $\tilde{n}_{r-1} = 0$}\\
        \frac{\bar{\bm{X}}_{r-1}+\epsilon\bar{\bm{X}}_{\text{full}}}{1+\epsilon} & \text{if $\tilde{n}_{r-1} \geq 1$ and $\underline{\tilde{\bm{X}}}^\top_{r-1} \underline{\tilde{\bm{X}}}_{r-1}$ is not invertible}\\
       \bar{\bm{X}}_{r-1} & \text{if $\tilde{n}_{r-1} \geq 1$ and $\underline{\tilde{\bm{X}}}^\top_{r-1} \underline{\tilde{\bm{X}}}_{r-1}$ is invertible}
    \end{cases} 
\end{equation*}
}
and {\small
\begin{equation*}
    \underline{\bm{S}}^*_{r-1} =
    \begin{cases}
      \underline{\bm{S}}_{\text{full}} & \text{if $\tilde{n}_{r-1} = 0$}\\
        (\frac{1}{\tilde{n}_{r-1}}\underline{\bm{X}}_{r-1}^\top \underline{\bm{X}}_{r-1} + \frac{\epsilon}{N} \underline{\bm{X}}_{\text{full}}^\top \underline{\bm{X}}_{\text{full}}) - (1+\epsilon)\bar{\bm{X}}^*_{r-1}\bar{\bm{X}}^{*\top}_{r-1} & \text{if $\tilde{n}_{r-1} \geq 1$ and $\underline{\tilde{\bm{X}}}^\top_{r-1} \underline{\tilde{\bm{X}}}_{r-1}$ is not invertible}\\
       \underline{\bm{S}}_{r-1} & \text{if $\tilde{n}_{r-1} \geq 1$ and $\underline{\tilde{\bm{X}}}^\top_{r-1} \underline{\tilde{\bm{X}}}_{r-1}$ is invertible.}
    \end{cases}
\end{equation*}
}
\end{singlespacing}

\label{thm:mahal}
\end{theorem}

Theorem \ref{thm:mahal} shows that at every stage, the D-optimal selection function selects the unit among the remaining pool of available units whose covariate vector maximizes a type of Mahalanobis distance. In its first choice, treatment 1 maximizes the Mahalanobis distance from the covariate distribution in the full sample (in particular, from $\bar{\bm{X}}_{\text{full}}$), thereby choosing the most outlying unit available in the full sample. For the subsequent stages where $\underline{\tilde{\bm{X}}}^\top_{r-1} \underline{\tilde{\bm{X}}}_{r-1}$ is not invertible, treatment 1 maximizes the Mahalanobis distance from a mixture covariate distribution between treatment group 1 and the full sample, where $\epsilon$ determines the mixing rate. Finally, the latter selections by treatment 1 maximize the Mahalanobis distance from the covariate distribution in treatment group 1. Therefore, with every selection, treatment 1 maximizes the overall separation of the covariates from its current mean, which increases the efficiency of the estimated regression coefficients. 


By definition, the D-optimal selection function improves the accuracy of the fitted linear model in each treatment group by sequentially minimizing the generalized variance of the estimated regression coefficients. With the D-optimal selection function, we can also establish several additional desirable properties of the FSM. In particular, leveraging the connection between D-optimality and Mahalanobis distance, we can show that FSM with the D-optimal selection function is affine invariant, i.e., the selections of units by the treatment groups remain unchanged even if the covariates are transformed linearly. See Section \ref{sec_properties_dopt} in the Online Supplementary Materials for a proof. An implication of this property is that the FSM is invariant with respect to changes in the location and scale of the covariates.

The FSM with the D-optimal selection function is appealing also because it can encompass several classical designs, such as randomized blocked and matched-pair designs. Theorem \ref{thm:retrieve} formalizes this result.
In the traditional randomized block design (RBD), the units are grouped into blocks of size $G$ according to a categorical, blocking variable, and each treatment is randomly applied to exactly one unit within each block (see, e.g., \citealt{cox2000theory}, Section 3.4). Here we consider a more general version of an RBD where the blocks are of size $c \times G$ (where $c$ is a fixed positive integer) and each treatment is applied to $c$ units within each block. This is a special case of a stratified randomized experiment with strata of equal size and equal allocation among treatments per stratum. In a matched-pair design with $G=2$ treatments, similar units are grouped into pairs, and each treatment is randomly applied to one unit within each pair. This is also a special case of a stratified randomized experiment with equal allocation per strata, where the size of each stratum equals two. 

\begin{theorem}\normalfont
\begin{enumerate}[label=(\alph*)]

\item Consider $N = cBG$ units belonging to $B$ blocks of equal size that are to be randomly assigned into $G$ treatment groups of equal size, where $c$ is a fixed positive integer.
Then, if the linear model in the FSM consists of an intercept and indicators of any $B-1$ levels of the blocking variable, the FSM with the D-optimal selection function produces the same assignment as an RBD.
\item Consider $N/2$ identical pairs of units in terms of baseline covariates $\bm{X}_i$ that are to be assigned into $G = 2$ treatment groups of equal size.
Assume $\bm{X}_i$ is drawn from a continuous distribution.
Then, if the linear model in the FSM consists of the intercept and the covariates $\bm{X}_i$, then the FSM almost surely produces the same assignment mechanism as a matched-pair design.
\end{enumerate}
\label{thm:retrieve}
\end{theorem}

In the first case, Theorem \ref{thm:retrieve}(a) states that, by including the levels of a blocking variable as regressors, the FSM with the D-optimal selection function automatically blocks on that variable. Thus, the FSM retrieves an RBD without explicitly performing separate randomizations within each block. 
In the second case, Theorem \ref{thm:retrieve}(b) states that, by including the covariates as regressors, the FSM with the D-optimal selection function produces the same assignment as a matched-pair experiment, without explicitly performing separate randomizations in each pair. This phenomenon is particularly useful when the sample consists of near-identical twins but that are difficult to identify a priori due to multiple covariates.


%% file: sec5.tex
\section{Inference under the FSM}
\label{sec_inference}

Using the FSM we can make model- and randomization-based inferences. Both modes of inference are feasible for any selection function and any randomized SOM. 
In model-based inference, the sample is typically assumed to be drawn randomly from some superpopulation, and inference for the PATE is done by modeling the observed outcome distribution conditional on the treatment indicators and the covariates. 
For instance, let the potential outcome model under treatment $g$ be $Y_i(g) = \bm{\beta}^\top_{g}\bm{B}(\bm{X}_i) + \epsilon_{ig}$, where  $\bm{B}(\bm{X}_i) = ( B_1(\bm{X}_i),...,B_b(\bm{X}_i) )^\top$ is a vector of $b$ basis functions of the covariates, and $\epsilon_{ig}$, $i \in \{1,2,...,N\}$ are mutually independent errors, independent of the covariates. 
Under this model, $\text{PATE}_{g',g''}$ can be unbiasedly estimated by $\widehat{\text{PATE}}_{g',g''} = \hat{\bm{\beta}}^\top_{g'}\overline{\bm{B}(\bm{X})} - \hat{\bm{\beta}}^\top_{g''}\overline{\bm{B}(\bm{X})}$, where $\overline{\bm{B}(\bm{X})} = \frac{1}{N}\sum_{i=1}^{N}\bm{B}(\bm{X}_i)$ and $\hat{\bm{\beta}}_g$ is the OLS estimator of $\bm{\beta}_g$  obtained by fitting a linear regression of $Y^{\text{obs}}_i$ on $\bm{B}(\bm{X}_i)$ in treatment group $g = g', g''$. 
We call this the regression imputation estimator of $\text{PATE}_{g',g''}$. 
The standard error of this estimator and the corresponding confidence interval for $\text{PATE}_{g',g''}$ can be obtained using standard OLS theory. 
We note that, in model-based inference, the standard errors and confidence intervals do not take into account the randomness stemming from the assignment mechanism. 


In randomization-based inference, the potential outcomes and the covariates are typically considered fixed and the assignment mechanism is the only source of randomness (see Chapter 2 of \citealt{rosenbaum2002observational} and chapters 5--7 of \citealt{imbens2015causal} for overviews). 
Inference for causal effects can be done via exact randomization tests for sharp null hypotheses on unit-level causal effects (\citealt{fisher1935design}), or via estimation under Neyman's repeated sampling approach (\citealt{neyman1923application}). 
Under the FSM, randomization tests for sharp null hypotheses can be performed by approximating the distribution of the test statistic through repeated realizations of the FSM. 
To illustrate, consider testing the sharp null hypothesis of zero unit-level causal effects, i.e., $H_0: Y_i(2) - Y_i(1) = 0$ for all $i$, at level $\alpha$ using the FSM. While any choice of test statistic preserves the validity of the test, a common choice is the absolute difference-in-means statistic $|\frac{1}{n_2}\sum_{i:Z_i = 2}Y^{\text{obs}}_i - \frac{1}{n_1}\sum_{i:Z_i = 1}Y^{\text{obs}}_i| = |\frac{1}{n_2}\sum_{i:Z_i = 2}Y_i(2) - \frac{1}{n_1}\sum_{i:Z_i = 1}Y_i(1)| = : T\{\bm{Z},\bm{Y}(1),\bm{Y}(2)\}$.
Large values of $T\{\bm{Z},\bm{Y}(1),\bm{Y}(2)\}$ are considered evidence against $H_0$. 
Under $H_0$, $Y_i(2) = Y_i(1) = Y^{\text{obs}}_i$ and the vectors of potential outcomes $\bm{Y}(1)$ and $\bm{Y}(2)$ are known and fixed. 
The $p$-value of the test is given by  $p = P_{H_0}[T\{\bm{Z},\bm{Y}(1),\bm{Y}(2)\}\geq t_{\text{obs}}]$, where $t_{\text{obs}}$ is the value of the test statistic for the observed realization of $\bm{Z}$ under the FSM. 
We can compute this $p$-value by Monte Carlo approximation, i.e., we generate independent vectors of assignments $\bm{Z}^{(m)} = (Z^{(m)}_1,...,Z^{(m)}_N)^\top$, $m \in \{1,2,...,M\}$ using the FSM and approximate the $p$-value as $\hat{p} = \frac{1}{M}\sum_{m = 1}^{M}\mathbbm{1}\big[T\{\bm{Z}^{(m)},\bm{Y}(1),\bm{Y}(2)\}\geq t_{\text{obs}}\big]$. We reject $H_0$ at level $\alpha$ if $\hat{p}\leq \alpha$. 

Similar tests can be applied for more general sharp hypotheses of treatment effects (e.g., dilated and tobit effects; \citealt{rosenbaum2002observational}, \citeyear{rosenbaum2010design2}). 
We can invert these tests to obtain a confidence interval for the hypothesized effect (\citealt{rosenbaum2002observational}, Section 2.6.1). Moreover, we can get a point estimate of the effect by solving a Hodges-Lehmann estimating equation corresponding to these tests (\citealt{rosenbaum2002observational}, Section 2.7.2).
Finally, under Neyman's approach, we can estimate the sample average treatment effect $\text{SATE}_{g',g''}$ by the difference-in-means statistic. 
In particular, for groups of equal size, this difference-in-means statistic is unbiased for $\text{SATE}_{g',g''}$ under the FSM (see Proposition \ref{fsm_prop:unbiased} for a proof).


%% file: sec6.tex
\section{The Health Insurance Experiment}
\label{sec_thehie}

\subsection{Data}
\label{sec_hiedata}

We evaluate the performance of the FSM relative to other common treatment assignment approaches using the baseline data of the HIE. To this end, we consider a version of the HIE data presented in \cite{aron2013rand}. This dataset comprises the six cost-sharing plans described in Section \ref{sec_hiedesign}. To make the group sizes more homogeneous, we combine the groups with $25\%$, $50\%$, and mixed coinsurance plans. Thus, in our analysis, we have $G = 4$ treatment groups corresponding to $g = 1$, ``free care'' ($n_1 = 564$); $g = 2$, ``$25\%, 50\%$, or mixed  coinsurance'' ($n_2 = 456$); $g = 3$, ``$95\%$ coinsurance'' ($n_3 = 372$); and $g = 4$, ``individual deductible'' ($n_4 = 495$). In total, there are $N = n_1 + ... + n_4 = 1,887$ families. We assign all $N$ families to the four treatment groups (i.e., without a discard group of non-participants).   
In this version of the HIE data, we pool the data across five of the six sites, and we randomly assign all the families to the four treatment groups. Due to loss of data, the Dayton site is excluded from this analysis. 

We consider $k = 20$ family-level baseline covariates, where $X_1, ..., X_5$ are scaled non-binary covariates, $X_6, ..., X_{14}$ are binary covariates, and $X_{15}, ..., X_{20}$ are binary covariates indicating missing data (see Table \ref{tab:hie1} for a description of each baseline covariate).
Using this data, we compare complete randomization (CRD), rerandomization (RR), and the FSM in terms of balance and efficiency. For the FSM, we generate the SOM by first using SCOMARS on the combined groups $\{1,2\}$ and $\{3,4\}$, and then using SCOMARS again to split each combined group into its component groups. For the FSM, we also use the D-optimal selection function based on a linear potential outcome model on the main covariates. The assignments under the FSM are generated using the open source R package \texttt{FSM} available on CRAN. 
For rerandomization, we consider two balance criteria, one based on the Wilks' lambda statistic (RR Wilks; \citealt{lock2011rerandomization}, Section 5.2) and the other based on the maximum pairwise Mahalanobis distance between any two treatment groups (RR Mahalanobis; \citealt{morgan2012rerandomization}). The balance criteria for both RR Wilks and RR Mahalanobis are based on all the main covariates and the squares and pairwise products of the scaled (non-binary) covariates. Finally, for both rerandomization methods, we use an acceptance rate of 0.001 \citep{lock2011rerandomization}.  
We draw 400 independent assignments for each approach. The results under RR Wilks and RR Mahalanobis are roughly the same (see Section \ref{appsec_hie} in the Online Supplementary Materials), and hence, for conciseness, here we only discuss the results for RR Mahalanobis. 
The runtime of each of these assignments was approximately 78 seconds with RR Mahalanobis and 28 seconds with the FSM on a Windows 64-bit laptop computer with an Intel(R) Core i7 processor. 

\subsection{Balance}

Figures \ref{fig:test1}(a) and \ref{fig:test1}(b) display the distributions of ASMD across randomizations for the main covariates and their second-order transformations (squares and pairwise products). 
RR balances the main covariates and the second-order terms better than CRD. However, in both cases, the FSM improves considerably over CRD and RR. In fact, with the FSM, the average imbalance is less than half (0.02) of those under CRD and RR.
Also, with both CRD and RR, it is common to see imbalances greater than 0.1 ASMD, whereas such extreme imbalances are non-existent with the FSM. 
\begin{figure}[!ht]
\begin{subfigure}{.33\textwidth}
  \includegraphics[scale =0.33]{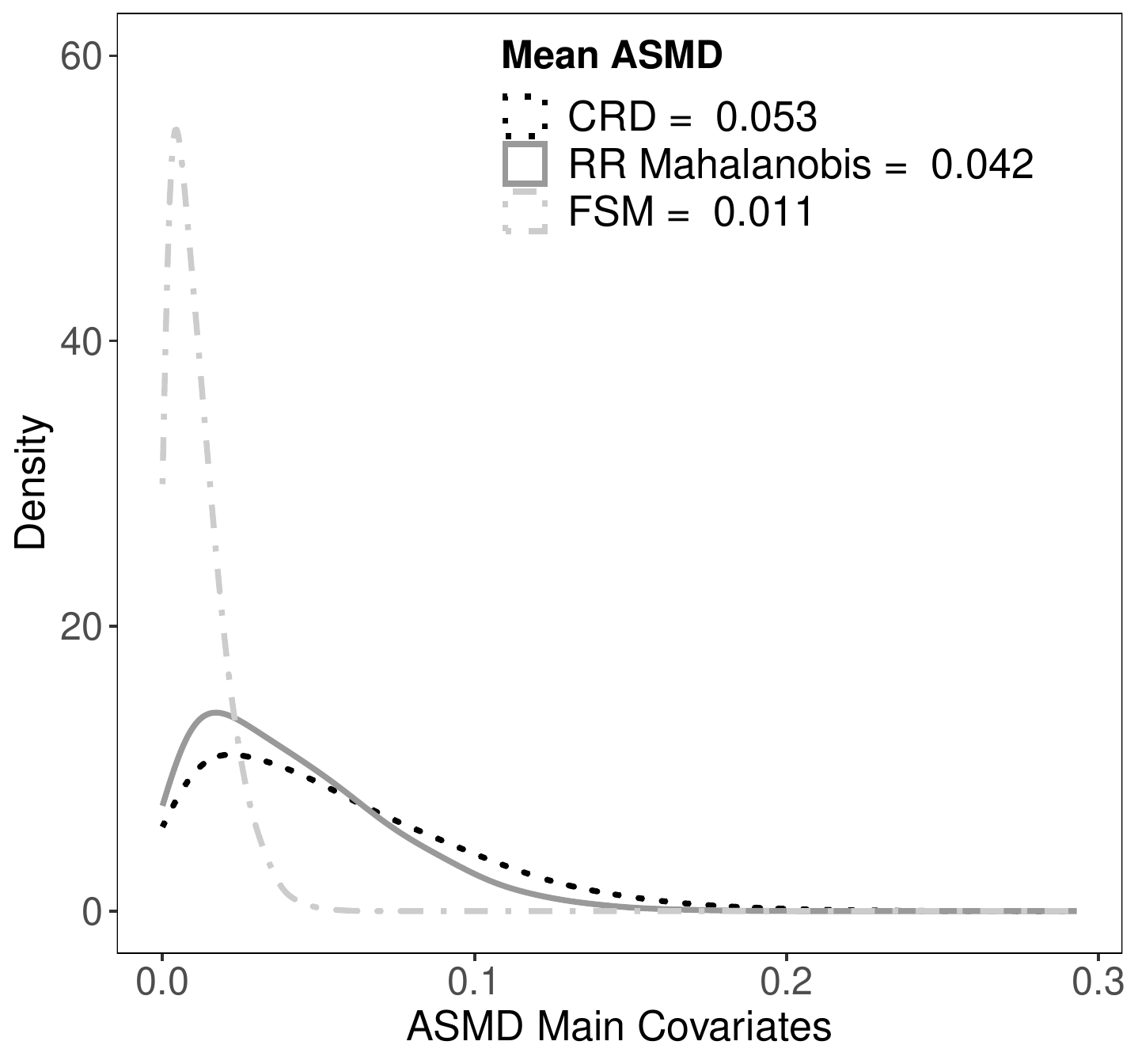}
  \caption{\footnotesize Main covariates}
\end{subfigure}%
\begin{subfigure}{.33\textwidth}
  \includegraphics[scale =0.33]{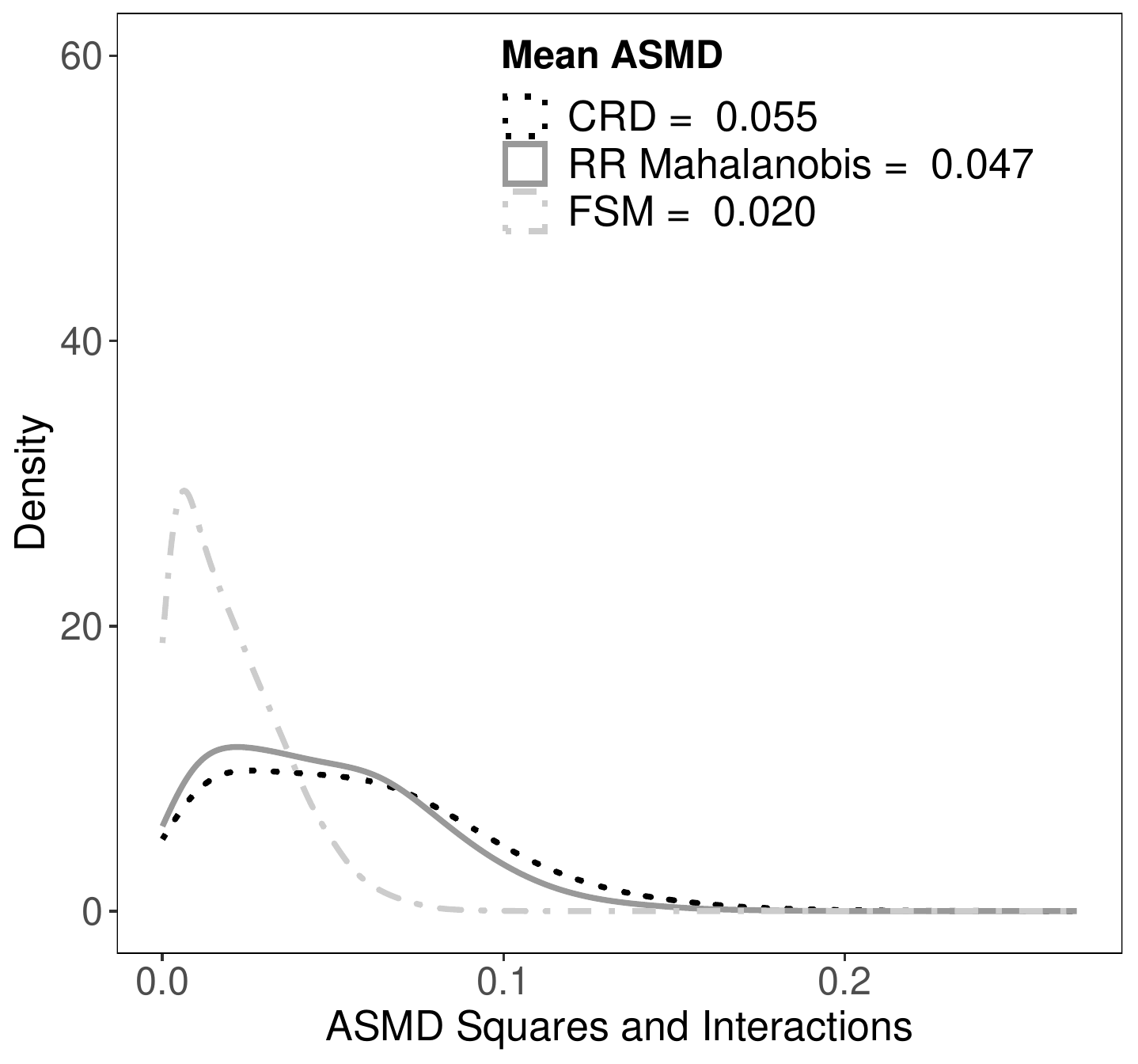}
  \caption{\footnotesize Second order terms}
\end{subfigure}
\begin{subfigure}{.33\textwidth}
  \includegraphics[scale =0.33]{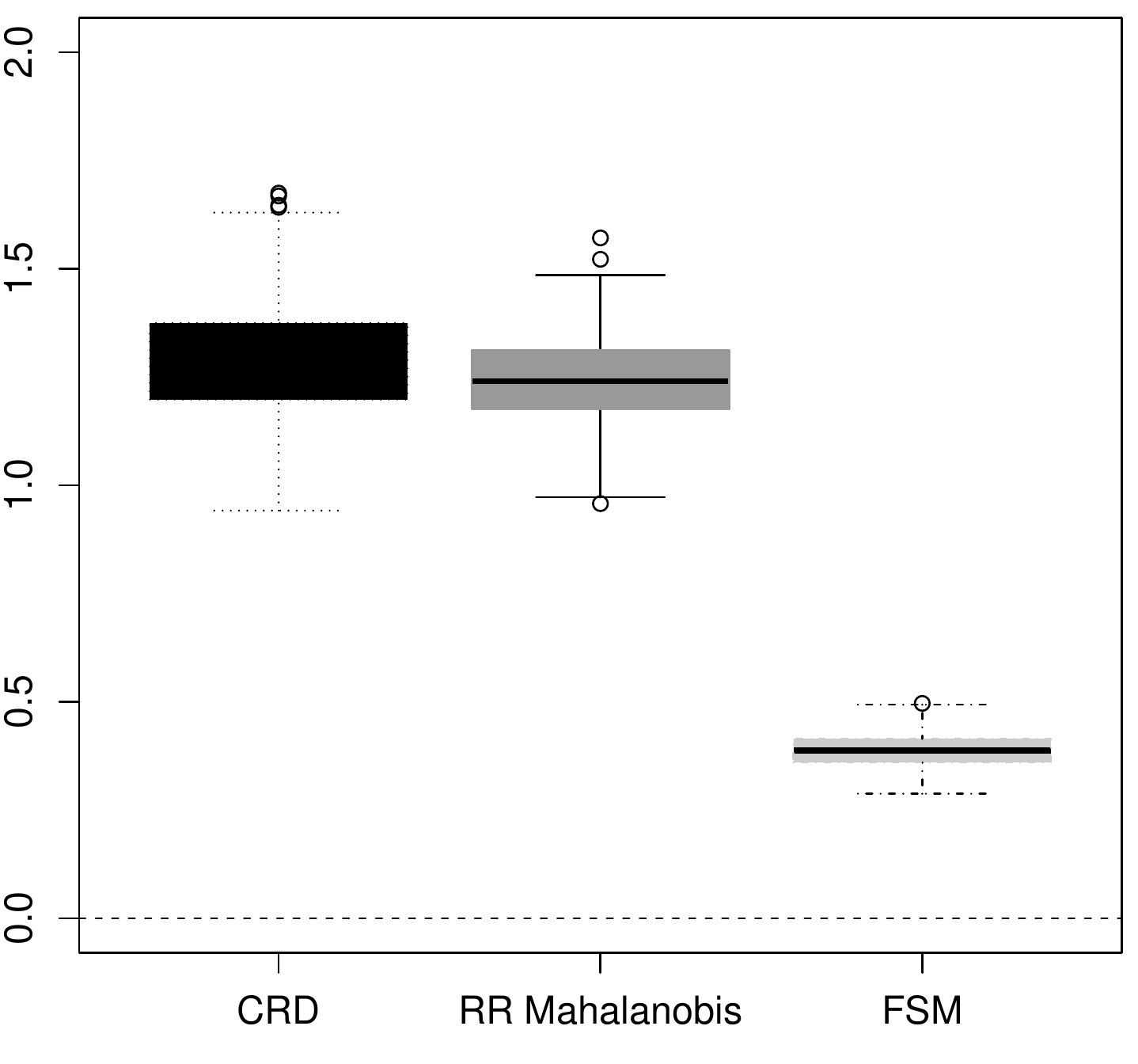}
  \caption{\footnotesize Frobenius norm}
\end{subfigure}
\caption{Distributions of absolute standardized mean differences (ASMD) of the main covariates (panel (a)) and their squares and pairwise products (panel (b)) across randomizations. For each plot, the legend presents the average ASMD across simulations for each method. Panel (c) shows the distributions of discrepancies between the correlation matrices of the covariates in treatment groups 1 and 2, as measured by the Frobenius norm, $||\underline{\bm{R}}_1 - \underline{\bm{R}}_2||_F$.
In terms of the main covariates, second-order transformations, and correlation matrices, the FSM substantially outperforms CRD and RR.}
\label{fig:test1}
\end{figure}

A related question is how well the methods balance all second-order features of the joint distribution of the covariates.
Figures \ref{fig:test1}(c) and \ref{figfrob} provide an answer to this question in the boxplots of the discrepancies between correlation matrices across randomizations. As a measure of discrepancy, we consider the Frobenius norm of the difference between correlation matrices in two groups, i.e., $||\underline{\bm{R}}_{g} - \underline{\bm{R}}_{g'} ||_F$, where $\underline{\bm{R}}_g$ is the sample correlation matrix in group $g$ and $||\cdot||_F$ is the Frobenius norm.\footnote{The Frobenius norm of a matrix is the square root of the sum of squares of all its elements.} Smaller values of $||\underline{\bm{R}}_g - \underline{\bm{R}}_{g'}||_F$ indicate better balance on the correlation matrix of the covariates between the groups $g$ and $g'$. 
As in the aforementioned second-order transformations, we see a similar performance between complete randomization and rerandomization, which is considerably improved by the FSM with a median about three times smaller.


\subsection{Efficiency}
\label{sec_hie_efficiency}

In this section, we evaluate the estimation accuracy of the methods under model- and randomization-based approaches to inference.
The main differences between the model- and randomization-based standard errors is that in the model-based approach, the variance calculation does not explicitly take into account the variability arising through the randomization distribution, whereas in the randomization-based approach it does. For illustration, here we consider estimating the average treatment effect of treatment 3 relative to treatment 2, i.e., $\text{SATE}_{3,2}$ and $\text{PATE}_{3,2}$. The results for the average treatment effects with other pairs of treatment groups are similar.

Under the model-based approach, we consider two potential outcome models, one that is linear on the main covariates (Model A1), and another that is linear on the main covariates and the second-order transformations of the scaled covariates (Model A2).
The results are summarized in Table 4.
While the performance of the three methods is similar under Model A1, under Model A2 there are substantial differences, with the FSM outperforming both complete randomization and rerandomization. 
In fact, under Model A2, there is a 14-15\% reduction in the average standard error, and a 53-64\% reduction in the maximum standard error, with the FSM.


\begin{singlespacing}
\begin{table}[H]
   \caption{Average and maximum model-based standard errors relative to the FSM across randomizations. 
   Under Model A1 (linear model on the covariates), the FSM is slightly more efficient than RR and CRD. Under Model A2 (linear model on the covariates and their second-order transformations), the FSM is considerably more efficient than CRD and RR.}
   \begin{subtable}{.5\linewidth}
   \centering
   \caption{\footnotesize Model A1}
   \scalebox{0.72}{
            \begin{tabular}{p{2.5cm}ccc}
    \toprule 
    \multirow{2}{5cm}{} & \multicolumn{3}{c}{Designs}\\
   \cline{2-4}
    & CRD & RR Mahalanobis & FSM\\
    \toprule
Average SE & 1.02 & 1.01 & 1.00 \\
Maximum SE & 1.04 & 1.02 & 1.00 \\
\bottomrule
  \end{tabular}
}
    \end{subtable}%
    \begin{subtable}{.5\linewidth}
    \centering
   \caption{\footnotesize Model A2}
   \scalebox{0.72}{
        \begin{tabular}{p{2.5cm}ccc}
    \toprule 
    \multirow{2}{5cm}{} & \multicolumn{3}{c}{Designs}\\
   \cline{2-4}
    & CRD & RR Mahalanobis & FSM\\
    \toprule
Average SE & 1.15 & 1.14 & 1.00 \\
Maximum SE & 1.64 & 1.53 & 1.00 \\
\bottomrule 
  \end{tabular}
  }
  \end{subtable}
  \label{tab:hie_var_model}
    \end{table}
\end{singlespacing}

Under the randomization-based approach, we consider the generative models $Y(3) = 10 + 2X_1 + 3X_2 + 0.5X_3 + 0.3X_4 + \eta$ (Model B1) and $Y(3) = 10 + 2X_1 + 2X_2X_3 - X_4X_5 + \eta$ (Model B2) where $Y(3) = Y(2)$ and $\eta \sim \mathcal{N}(0,1.5^2)$. Here, both the generative models satisfy the sharp-null hypothesis of zero treatment effect for every unit and hence, $\text{SATE}_{3,2} = 0$. Under each design, $\text{SATE}_{3,2}$ is estimated using the standard difference-in-means estimator and the corresponding randomization-based SE is obtained by generating 400 randomizations and computing the standard deviation of the estimator across these 400 randomizations.  
The results are summarized in Table \ref{tab:hie_var_rand}. See Appendix \ref{appsec_hie} for similar comparisons under a set of different generative models of the potential outcome.
In terms of efficiency, we see again a clear advantage of the FSM.
Under both Model B1 and Model B2, the average standard errors of complete randomization and rerandomization are more than twice of those under the FSM.

\begin{singlespacing}
\begin{table}[H]
   \caption{Randomization-based standard errors relative to the FSM. The standard error for the FSM is 0.11 under Model B1 (linear model on the covariates) and 0.64 under Model B2 (linear model on the covariates and their second-order transformations). Under both models, the FSM is considerably more efficient than both CRD and RR.}
   \begin{subtable}{.5\linewidth}
   \centering
   \caption{\footnotesize Model B1}
   \scalebox{0.75}{
            \begin{tabular}{p{1.5cm}ccc}
    \toprule 
    \multirow{2}{5cm}{} & \multicolumn{3}{c}{Designs}\\
   \cline{2-4}
    & CRD & RR Mahalanobis & FSM\\
    \toprule
SE & 2.47 & 2.08 & 1 \\
\bottomrule
  \end{tabular}
}
    \end{subtable}%
    \begin{subtable}{.5\linewidth}
    \centering
   \caption{\footnotesize Model B2}
   \scalebox{0.75}{
        \begin{tabular}{p{1.5cm}ccc}
    \toprule 
    \multirow{2}{5cm}{} & \multicolumn{3}{c}{Designs}\\
   \cline{2-4}
    & CRD & RR Mahalanobis & FSM\\
    \toprule
SE & 2.63 & 2.25 & 1 \\
\bottomrule 
  \end{tabular}
  }
  \end{subtable}
\label{tab:hie_var_rand}
    \end{table}
\end{singlespacing}

\subsection{Intuition and further explorations}
\label{sec_intuition}

Our analysis illustrates some important differences between the FSM, CRD, and RR.
With respect to RR, these differences pertain to the specification, role, and implementation of the assignment criterion.
First, regarding the specification of the criterion, while RR uses the Mahalanobis distance, the FSM uses the D-optimality criterion, which, coupled with a suitable SOM, leads to robust assignments under a more general class of potential outcome models.

Second, regarding the role of this criterion, while RR essentially constrains the allowable treatment assignments, the FSM seeks to optimize them toward the criterion.
In essence, while RR solves a feasibility problem by resampling, the FSM aims to solve a maximization problem by step-wise assignment. 
Furthermore, the feasibility problem solved by RR depends on the balance threshold, which can be difficult to select in practice.
While a very high threshold can accept assignments with poor covariate balance, a very low one can be computationally onerous.

Third, regarding the implementation of the criterion, while RR assigns all units in one step and then discards imbalanced assignments, the FSM assigns units in multiple steps (one at a time) in a random but optimal fashion determined by the selection order and the selection criterion. 
This difference is crucial because in experiments like the HIE with several treatment groups and many covariates, the space of possible treatment assignments is vast. 
As shown in our analyses, optimally selecting among these assignments in a step-wise manner can make a substantial improvement in terms of balance, efficiency, computational time, and, ultimately, in the use of scarce resources available for experimentation.
\footnote{Figures \ref{fig:simu_boxplot} and \ref{fig:test1} show that, although RR does well under common balance standards (the mean differences are systematically lower than the typical threshold of 0.1 ASMD), there is room to select better (more balanced) random treatment assignments, which is achieved by the FSM.}


To better see this, we asked how we would need to modify RR to achieve comparable performance to the FSM? 
Using the HIE data, we approximated the randomization distribution of the imbalance criterion of RR (i.e., the maximum Mahalanobis distance $M$ across all pairs of treatment groups) by generating random assignments for 100 hours.
See Table \ref{tab:hie_mahaldist} for a summary of the results. 
The table displays summary statistics of the distribution of $M$ under CRD, RR, and the FSM.
As shown in Table \ref{tab:hie_mahaldist}, the highest (worst-case) value of $M$ under the FSM is smaller than the smallest (best-case) value of $M$ under CRD and RR. 
Importantly, even if we set the RR acceptance rate to 0.0000001 (i.e., 1 over 10 million), we still have imbalances higher than the worst-case imbalance of the FSM. 
In sum, even with an acceptance rate as low as 0.0000001, RR did not perform as well as the FSM, despite taking 100 hours on average to generate a single assignment, as opposed to the 30 seconds of running time of the FSM.

\begin{singlespacing}
    \begin{table}[H]
        \centering
        \scalebox{0.8}{
        \begin{tabular}{ccccccc}
        \toprule
     Design &  Minimum & 1st Quartile  & Median  & Mean & 3rd Quartile  & Maximum \\
     \hline
CRD & 18.5 & 39.5 & 43.9  & 44.4  & 48.7 & 96.1 \\


RR (0.001) & 18.5  & 25.4 & 26.2 & 25.9 &  26.7 & 27.1\\

FSM & 2.8   & 4.7  & 5.3  & 5.4 & 6.0 & 10.6 \\
\bottomrule
        \end{tabular}
        }
        \caption{Distribution of the maximum pairwise Mahalanobis distance across groups ($M$). 
        For CRD, we obtain this distribution by generating over 10 million random assignments for 100 hours. For RR (0.001), we obtain this distribution using 0.1\% of all these assignments with the smallest values of $M$. For the FSM, we obtain this distribution using the  400 random assignments from Section \ref{sec_hiedata}.}
        \label{tab:hie_mahaldist}
    \end{table}
\end{singlespacing}

%% file: sec7.tex
\section{Ten further studies in the health and social sciences}
\label{sec_additional}

In addition to the previous study, we evaluate the performance of the FSM in ten randomized studies from the health and social sciences. These ten studies are labelled (1) Crepon, which evaluates the impact of a microcredit program in rural Morocco on assets, profits, and consumption \citep{crepon2015estimating}; (2) Angrist, which evaluates the impact of cash incentives on certification rates among low-achievers in Israel \citep{angrist1999using}; (3) Finkelstein, which evaluates the impact of the Camden Coalition of Healthcare Providers' Hotspotting program on hospital readmission rates among patients with high use of healthcare services \citep{finkelstein2020health}; (4) Durocher, which evaluates the impact of intravenous infusion versus intramusculur oxytocin on postpartum blood loss and hemmorhage rates \citep{durocher2019does}; (5) Lalonde, which evaluates the impact of Nationally Supported Work program on earnings \citep{lalonde1986evaluating}; (6) Karlan, which evaluates the impact of loans with an indemnity component on demand for credit and investment decisions of farmers \citep{karlan2014agricultural}; (7) Dupas, which evaluates the impact of different cost provisions for allocating dilute-chlorine water treatment solution on chlorine residuals in households' stored water \citep{dupas2016targeting}; (8) Blattman, which evaluates the impact of industrial job offers and entrepreneurial programs on health, income and other measures \citep{blattman2018impacts}; (9) Ambler, which evaluates the impact of offering Salvadoran migrant maching funds for educational remittances on educational investments and other outcomes \cite{ambler2015channeling}; (10) Wantchekon, which evaluates the impact of townhall meeting based on programmatic, nonclientelist platforms on clientelism, voter turnout, and vote shares \citep{fujiwara2013can}. Table \ref{tab:ten_studies} provides details on the design parameters considered in these studies.

\begin{singlespacing}
\begin{table}[!ht]
\centering
\scalebox{0.7}{
\begin{tabular}{ccccp{1.2cm}ccccp{1.2cm}ccccc}
\toprule 
\multirow{2}{1cm}{Study} & \multicolumn{4}{c}{Design parameters} & \multicolumn{5}{c}{Main covariates} & \multicolumn{5}{c}{Second-order transformations}\\
 \cline{2-5} \cline{6-10} \cline{11-15} 
 & $N$ & $G$ & $(n_1,...,n_G)$ & $k$ & CRD & RR & FSM & $\frac{\text{CRD}}{\text{FSM}}$ & $\frac{\text{RR}}{\text{FSM}}$ & CRD & RR & FSM & $\frac{\text{CRD}}{\text{FSM}}$ & $\frac{\text{RR}}{\text{FSM}}$\\
\toprule
Crepon  & 4465 & 2 & (2266, 2199) & 33 & 0.024 & 0.018 & 0.015 & 1.6 & 1.2 & 0.024 & 0.023 & 0.018 & 1.3 & 1.3\\
Angrist & 3821 & 2 & (1910,1911) & 20 & 0.025 & 0.014 & 0.002 & 12.5 & 7.0 & 0.026 & 0.023 & 0.003 & 8.7 & 7.7\\ 
Finkelstein & 782 & 2 & (389,393) & 10 & 0.062 & 0.020 & 0.010 & 6.2 & 2.0 & 0.059 & 0.048 & 0.013 & 4.5 & 3.7\\ 
Durocher & 480 & 2 & (239,241) & 12 & 0.072 & 0.031 & 0.017 & 4.2 & 1.8 & 0.073 & 0.068 & 0.022 & 3.3 & 3.1\\ 
Lalonde & 445 & 2 & (222,223) & 10 & 0.083 & 0.044 & 0.014 & 5.9 & 3.1 & 0.077 & 0.070 & 0.019 & 4.1 & 3.7\\ 
Karlan & 169 & 2 & (84, 85) & 16 & 0.124 & 0.059 & 0.053 & 2.3 & 1.1 & 0.123 & 0.119 & 0.060 & 2.1 & 2.0\\ 
Dupas & 1118 & 3 & (351, 382, 385) & 11 & 0.059 & 0.018 & 0.010 & 5.9 & 1.8 & 0.058 & 0.044 & 0.017 & 3.4 & 2.6\\
Blattman & 947 & 3 & (358,304,285) & 34 & 0.064 & 0.048 & 0.026 & 2.5 & 1.8 & 0.065 & 0.064 & 0.036 & 1.8 & 1.8\\
Ambler & 991 & 4 & (360, 211, 203, 217) & 16 & 0.073 & 0.053 & 0.015 & 4.9 & 3.5 & 0.073 & 0.071 & 0.017 & 4.3 & 4.2\\ 
Wantchekon & 24 & 2 & (12, 12) & 10 & 0.334 & 0.170 & 0.245 & 1.4 & 0.7 & 0.333 & 0.289 & 0.237 & 1.4 & 1.2\\ 
\hline
Average & & & & & 0.092 & 0.048 & 0.041 & & & 0.091 & 0.082 & 0.044 & & \\ 
Average* & & & & & 0.065 & 0.034 & 0.018 & & & 0.064 & 0.059 & 0.023 & &\\ 
   \bottomrule
\end{tabular}
}
\caption{Design parameters and balance results for ten case studies in the health and social sciences. The second average denoted with an asterisk ($^*$) excludes the Wantchekon study because $\underline{\tilde{\bm{X}}}^\top_{r} \underline{\tilde{\bm{X}}}_{r}$ matrix is non-invertible for the first $r = 22$ selections.}
\label{tab:ten_studies}
\end{table}
\end{singlespacing}

For each study, we generate 100 assignments of complete randomization (CRD), Rerandomization with Mahalanobis distance (based on the main covariates) and  0.001 acceptance rate (RR), and the FSM (based on the main covariates). The mean ASMD of the main covariates and their squares and interactions under each method are presented in Table \ref{tab:ten_studies}. See figures \ref{fig:five_studies} and \ref{fig:five_studies2} in the Online Supplementary Materials for plots of the  distributions of these imbalances, alongside the Frobenius norms of $\underline{\bm{R}}_1 - \underline{\bm{R}}_2$.\footnote{Groups 1 and 2 are chosen haphazardly as a typical pair of groups. The results for the other pairs of groups are similar.}

Table \ref{tab:ten_studies} shows that for each study, CRD achieves a similar mean balance on the main covariates and their squares and interactions. RR improves balance over CRD considerably for the main covariates, but only mildly for the squares and interactions. 
By contrast, for almost all the studies, the FSM substantially improves balance over CRD and RR in terms of both the main covariates and their transformations.
The only exception is the Wantchekon study, where the group sizes barely exceed the number of covariates $k = 10$. 
The FSM is not designed for settings like this, where the number of covariates is greater than or close to the minimum treatment group size.
In such settings, the matrix $\underline{\tilde{\bm{X}}}^\top_{r} \underline{\tilde{\bm{X}}}_{r}$ is non-invertible for almost every selection stage of the FSM, and therefore, the D-optimal selection function in the FSM relies on ridge augmentation to feasibly select units (see Section \ref{sec_theD}), producing suboptimal selections.

Across the ten studies, the ASMD on the main covariates are 55\% ($=\frac{0.092 - 0.041}{0.092}$) and 15\% ($=\frac{0.048 - 0.041}{0.048}$) lower on average with the FSM than CRD and RR, respectively. 
If we exclude Wantchekon, then these percent reductions in ASMD are amplified to 72\% and 47\%.
Similarly, across the ten studies, the ASMD on the squares and interactions of the covariates with the FSM are about 50\% smaller than both CRD and RR, and without Wantchekon, they are at least 60\% smaller. 
In fact, FSM has better balance on both the main covariates and their second-order transformations over CRD and RR uniformly across the first nine studies (as shown by the $\frac{\text{CRD}}{\text{FSM}}$ and $\frac{\text{RR}}{\text{FSM}}$ columns). 
For each study, the relative improvement in balance under the FSM over RR is larger for the second-order transformations than for the main covariates. 
In particular, for half of the ten studies, the mean ASMD of the second-order transformations under RR are at least three times larger than those under the FSM, implying substantial improvement in balance on these transformations under the FSM. 

Overall, averaging the ASMD of the main covariates and their second-order transformations across the first nine studies, we see that the FSM achieves 68\% better covariate balance than complete randomization and 56\% better covariate balance than rerandomization in a typical study. Across these studies, the FSM's performance relative to CRD and RR is consistent with those in the HIE study in Section \ref{sec_thehie}.\footnote{Notably, the relative performances of the methods in the HIE study are comparable to those of the Ambler study, which involves roughly half the sample size of the HIE study and similar values of the other design parameters. For instance, the average ASMD of the main covariates under CRD, RR, and the FSM in the Ambler study are roughly $\sqrt{2}$ times those in the HIE study, where $\sqrt{2}$ is the factor that corrects for the difference in sample size.}
A similar pattern to the HIE study is also noted in the plots of $||\underline{\bm{R}}_1 - \underline{\bm{R}}_2||_{F}$ in figures \ref{fig:five_studies} and \ref{fig:five_studies2} in the Online Supplementary Materials, where for most studies, the worst (least balanced) assignment among all the draws of the FSM has a better balance on the correlation matrices than the best (most balanced) assignment among all the draws of CRD and RR. As discussed in Section \ref{sec_hie_efficiency}, under both model- and randomization-based approaches to inference, better balance directly translates to more efficient estimates of treatment effects.  

Therefore, similar to the HIE case study, these results show that across a range of randomized experiments, the FSM is a flexible and robust approach to randomization. 


%% file: sec8.tex
\section{Practical considerations and extensions}
\label{sec_practical}

\subsection{Multi-group experiments}
\label{sec_multi}
As discussed, the FSM can readily handle experiments with multiple treatment groups. In so doing, the key methodological consideration is the choice of an SOM. 
As in two-group experiments, we would like to generate an SOM that is randomized and sequentially controlled, so that at every stage of the random selection process, the number of selections made by each treatment group up to that stage is close to its fair share. 
Constructing a sequentially controlled SOM for multi-group experiments with arbitrary group sizes is an open problem. However, such constructions are possible for several practically relevant configurations of the group sizes, namely (a) groups of equal size, (b) groups having one of two distinct sizes, and (c) groups of more than two distinct sizes such that when combined by groups of equal size they have the same total size. 
In the Online Supplementary Materials, we provide algorithms to construct an SOM for all three configurations and prove that the resulting SOM is sequentially controlled. 
In practice, for more general group size configurations, one strategy to generate an SOM is to first identify one of these three configurations that is structurally similar to the configuration at hand, and then use the corresponding SOM-generating algorithm. The resulting algorithm may not always be sequentially controlled, but is still likely to produce a well-controlled randomized selection order. 

\subsection{Stratified experiments}
\label{sec_stratified}

In stratified experiments, units are grouped into two or more strata, and within each stratum, units are randomly assigned to treatment.
Here we propose a family of extensions of the FSM to such settings.
Typically, in stratified experiments the treatment group sizes within each stratum are pre-specified by the investigator.
The main challenge arises when the treatment group sizes differ across strata. 
To address this challenge, we construct an augmented SOM with information of the treatment group that selects at each stage and the stratum that it selects from.
This construction guarantees that each treatment group is assigned the pre-specified number of units in each stratum.
In the Online Supplementary Materials, we discuss two approaches to construct such an SOM. 
At a high level, one approach generates a separate SOM for each stratum, while the other approach uses SCOMARS to determine the order of stratum labels for each treatment.

\subsection{Sequential experiments}
\label{sec_sequential}

Sequential experiments are experiments where units progressively become available for random assignment, possibly in batches of varying sizes. 
Here we describe extensions of the FSM to such settings. 
The simplest approach is to run an independent FSM for each new batch of available units.
However, in general, this approach fails to account for accrued covariate imbalances between the treatment groups. 
To address this issue, we propose an alternative approach that considers the new batch as a continuation of the previous one. 
More specifically, for each unit in the new batch we evaluate the value of the D-optimal selection function using all the units already assigned to the selecting treatment group.
In this way, this approach tends to remove accrued covariate imbalances. 
See the Online Supplementary Materials for technical details.
In sequential experiments, the balance and efficiency of the FSM assignments tend to increase with the batch sizes.


%% file: sec9.tex
\section{Summary and remarks}
\label{sec_summary}
We revisited, formalized, and extended the FSM for experimental design. We proposed a new selection function based on D-optimality that requires no tuning parameters. We showed that, equipped with this selection function, the FSM has a number of appealing properties. First, the FSM is affine invariant and hence, it self-standardizes covariates with possibly different units of measurements. 
Second, the FSM produces randomized block designs without explicitly randomizing in each block. 
Third, the FSM also produces matched-pair designs without explicitly constructing the matched pairs beforehand and randomizing within each pair. 
We described how both model-based and randomization-based inference on treatment effects can be conducted using the FSM. For a range of practically relevant configurations of group sizes in multi-group experiments, we proposed new algorithms to generate a fair and random selection order of treatments under the FSM. We also discussed potential extensions of the FSM to stratified and sequential experiments.
In a case study on the RAND Health Insurance Experiment, and ten additional randomized studies from the health and social sciences, we showed that the FSM is a robust approach to randomization, exhibiting better performance than complete randomization and rerandomization in terms of balance and efficiency. 

While there are settings where complete randomization may perform better than the FSM in terms of efficiency, such settings are less common and involve jagged, i.e., highly non-smooth, potential outcome models.
In settings where these models are reasonably smooth, the FSM is expected to perform well. 
Overall, through our extensive explorations with real and simulated experimental data, the FSM has consistently stood out as a robust design that can handle multiple treatment groups and a fairly large number of categorical and continuous covariates without requiring tuning parameters and nor coarsening covariates.
We recommend giving strong considerations to the FSM in experimental design for its conceptual simplicity, practicality, balance, and robustness.





%% file: sec_supp.tex
\addcontentsline{toc}{section}{Supplementary Materials}
\renewcommand{\thesubsection}{\Alph{subsection}}
\setcounter{table}{0}
\renewcommand{\thetable}{A\arabic{table}}

\renewcommand{\theequation}{A\arabic{equation}}
\setcounter{figure}{0}
\renewcommand\thefigure{A\arabic{figure}}

\setcounter{theorem}{0}
\renewcommand\thetheorem{A\arabic{theorem}}

\setcounter{lemma}{0}
\renewcommand\thelemma{A\arabic{lemma}}

\subsection{Notation and estimands}

\begin{singlespacing}
\begin{table}[htbp]
\caption{Notation}
\label{tab_notation}
\begin{center}
\begin{tabular}{r c p{10cm} }
\hline
$N$ & $\triangleq$ & Full sample size\\
$i$ & $\triangleq$ & Index of unit, $i = 1, ..., N$\\
$G$ & $\triangleq$  & Number of treatments\\
$g$ & $\triangleq$ & Index of treatment group, $g = 1, 2, ..., G$\\
$n_g$ & $\triangleq$ & Size of treatment group $g$\\
$k$ & $\triangleq$ & Number of baseline covariates\\
$\boldsymbol{X}_{i}$ & $\triangleq$ & Observed vector of baseline covariates of unit $i$\\
$\underline{\bm{X}}_{\text{full}}$ & $\triangleq$ & $N \times k$ matrix of covariates in the full sample\\
$\underline{\tilde{\bm{X}}}_{\text{full}}$ & $\triangleq$ & $N \times k+1$ design matrix in the full sample\\
$\bar{\bm{X}}_{\text{full}}$ & $\triangleq$ & $k\times 1$ vector of means of the baseline covariates in the full sample\\
$\underline{\bm{S}}_{\text{full}}$ & $\triangleq$ & $k\times k$ covariance matrix of the baseline covariates in the full sample\\
$Y_i(g)$ & $\triangleq$ & Potential outcome of unit $i$ under treatment $g$	\\
$\bm{Y}(g)$ & $\triangleq$ & Vector of potential outcomes under treatment $g$, $(Y_1(g), ..., Y_N(g))^\top$ \\
$Z_i$ & $\triangleq$ & Treatment assignment indicator of unit $i$, $Z_i \in \{1,2,...,G\}$ \\ 
$\bm{Z}$ & $\triangleq$ & Vector of treatment assignment indicators, $(Z_1,...,Z_N)^\top$ \\
$Y^{\text{obs}}_i$ & $\triangleq$ & Observed outcome of unit $i$, $Y^{\text{obs}}_i = \sum_{g = 1}^{G} \mathbbm{1}(Z_i = g) Y_i(g)$\\
\hline
\end{tabular}
\end{center}
\end{table}
\end{singlespacing}

\begin{singlespacing}
\begin{table}[H]
\caption{Estimands}
\label{tab_estimands}
\begin{center}
\begin{tabular}{r c p{10cm} }
\hline
$Y_i(g') - Y_i(g'')$ & $\triangleq$ & Unit level causal effect of treatment $g'$ relative to treatment $g''$ for unit $i$; $g',g''
\in \{1,2,...,G\}$ 	\\
$\text{SATE}_{g',g''}$ & $\triangleq$ & $\frac{1}{N}\sum_{i=1}^{N} \{Y_i(g') - Y_i(g'') \}$, the Sample Average Treatment Effect of treatment $g'$ relative to treatment $g''$ \\
$\text{PATE}_{g',g''}$ & $\triangleq$ & $\mathbb{E} \{Y_i(g') - Y_i(g'')\}$, the Population Average Treatment Effect of treatment $g'$ relative to treatment $g''$ \\\hline
\end{tabular}
\end{center}
\end{table}
\end{singlespacing}
\subsection{Proofs of theoretical results}

\begin{lemma} \normalfont
Let treatment 1 be the choosing group at the $r$th stage. Also, let $\underline{\tilde{\bm{X}}}_{r-1}$ be the $\tilde{n}_{r-1} \times (k+1)$ design matrix in treatment group 1 after the $(r-1)$th stage, where $\tilde{n}_{r-1} \geq 1$ and $\text{rank}(\underline{\tilde{\bm{X}}}_{r-1}) = k+1$. The D-optimal selection function chooses unit $i'$ with covariate vector $\bm{X}_{i'} \in \mathbb{R}^k$, where 
\begin{equation}
      i' \in 
        \argmax\limits_{i \in \mathcal{R}_{r-1}}(1, \bm{X}^\top_i) (\underline{\tilde{\bm{X}}}^\top_{r-1} \underline{\tilde{\bm{X}}}_{r-1})^{-1} \begin{psmallmatrix}
       1\\
       \bm{X}_i
       \end{psmallmatrix} 
    \end{equation}

\label{lemma:dopt}
\end{lemma}

\begin{proof}
We follow the notations outlined in Section \ref{sec_theD}. At the $r$th stage, D-optimal selection function selects unit $i' \in \mathcal{R}_{r-1}$, where $i' \in \argmax\limits_{i \in \mathcal{R}_{r-1}}\det(\underline{\tilde{\bm{X}}}^\top_{r,i}\underline{\tilde{\bm{X}}}_{r,i})$. Now, for $i \in \mathcal{R}_{r-1}$, 
\begin{align}
    \det(\underline{\tilde{\bm{X}}}^\top_{r,i}\underline{\tilde{\bm{X}}}_{r,i}) &= \det \Big\{
\underline{\tilde{\bm{X}}}^\top_{r-1} \underline{\tilde{\bm{X}}}_{r-1} + \begin{psmallmatrix}
1\\
\bm{X}_i
\end{psmallmatrix}(1, \bm{X}^\top_i)\Big\}\\
& = \det(\underline{\tilde{\bm{X}}}^\top_{r-1}\underline{\tilde{\bm{X}}}_{r-1}) \det \Big\{\underline{\bm{I}} + (\underline{\tilde{\bm{X}}}^\top_{r-1}\underline{\tilde{\bm{X}}}_{r-1})^{-\frac{1}{2}} \begin{psmallmatrix}
1\\
\bm{X}_i
\end{psmallmatrix}  (1, \bm{X}^\top_i) (\underline{\tilde{\bm{X}}}^\top_{r-1}\underline{\tilde{\bm{X}}}_{r-1})^{-\frac{1}{2}} \Big\}\\
& = \det (\underline{\tilde{\bm{X}}}^\top_{r-1}\underline{\tilde{\bm{X}}}_{r-1}) \Big\{1 +  (1, \bm{X}^\top_i) (\underline{\tilde{\bm{X}}}^\top_{r-1}\underline{\tilde{\bm{X}}}_{r-1})^{-1} \begin{psmallmatrix}
1\\
\bm{X}_i
\end{psmallmatrix}  \Big\}, \label{eq_B_1}
\end{align}
where the final equality holds since for two matrices $\underline{\bm{A}}_{m \times n}$ and $\underline{\bm{B}}_{n \times m}$, $\det(\underline{\bm{I}}_m + \underline{\bm{A}}\underline{\bm{B}}) = \det(\underline{\bm{I}}_n + \underline{\bm{B}}\underline{\bm{A}})$. Equation \ref{eq_B_1} implies that the selected unit $i'$ maximizes $(1, \bm{X}^\top_i) (\underline{\tilde{\bm{X}}}^\top_{r-1}\underline{\tilde{\bm{X}}}_{r-1})^{-1} \begin{psmallmatrix}
1\\
\bm{X}_i
\end{psmallmatrix}$. This completes the proof.

\end{proof}

\subsubsection*{Proof of Theorem \ref{thm:mahal}}
\begin{proof}
We use the notations in Section \ref{sec_setup} and Table \ref{tab_notation}.
We first consider the case where $\tilde{n}_{r-1} = 0$. The selected unit $i'$ satisfies, 
\begin{equation}
i' \in 
        \argmax\limits_{i \in \mathcal{R}_{r-1}}(1, \bm{X}^\top_i) (\underline{\tilde{\bm{X}}}^\top_{\text{full}} \underline{\tilde{\bm{X}}}_{\text{full}})^{-1} \begin{psmallmatrix}
       1\\
       \bm{X}_i
       \end{psmallmatrix}.    
\end{equation}
Now, denoting $\bm{e}_1 = (1,0,...,0)$ as the $k \times 1$ first standard unit vector, we have
\begin{align}
    (1, \bm{X}^\top_i) (\underline{\tilde{\bm{X}}}^\top_{\text{full}} \underline{\tilde{\bm{X}}}_{\text{full}})^{-1} \begin{psmallmatrix}
       1\\
       \bm{X}_i
       \end{psmallmatrix} & = (1, \bm{X}^\top_i) (\underline{\tilde{\bm{X}}}^\top_{\text{full}} \underline{\tilde{\bm{X}}}_{\text{full}})^{-1} \begin{psmallmatrix}
       1\\
       \bar{\bm{X}}_{\text{full}}
       \end{psmallmatrix} + (1, \bm{X}^\top_i) (\underline{\tilde{\bm{X}}}^\top_{\text{full}} \underline{\tilde{\bm{X}}}_{\text{full}})^{-1} \begin{psmallmatrix}
       0\\
       \bm{X}_i - \bar{\bm{X}}_{\text{full}}
       \end{psmallmatrix} \nonumber\\
       & = (1, \bm{X}^\top_i) (\underline{\tilde{\bm{X}}}^\top_{\text{full}} \underline{\tilde{\bm{X}}}_{\text{full}})^{-1} \underline{\tilde{\bm{X}}}^\top_{\text{full}} \underline{\tilde{\bm{X}}}_{\text{full}} \frac{\bm{e}_1}{N}  + (1, \bm{X}^\top_i) (\underline{\tilde{\bm{X}}}^\top_{\text{full}} \underline{\tilde{\bm{X}}}_{\text{full}})^{-1} \begin{psmallmatrix}
       0\\
       \bm{X}_i - \bar{\bm{X}}_{\text{full}} 
       \end{psmallmatrix} \nonumber \\
          & = \frac{1}{N} + \{0, (\bm{X}_i - \bar{\bm{X}}_{\text{full}})^\top \} (\underline{\tilde{\bm{X}}}^\top_{\text{full}} \underline{\tilde{\bm{X}}}_{\text{full}})^{-1} \begin{psmallmatrix}
       0\\
       \bm{X}_i - \bar{\bm{X}}_{\text{full}}
        \end{psmallmatrix} \nonumber\\
        & = \frac{1}{N} + \frac{1}{N}(\bm{X}_i - \bar{\bm{X}}_{\text{full}})^\top (\underline{\bm{S}}_{\text{full}})^{-1} (\bm{X}_i - \bar{\bm{X}}_{\text{full}}).
\end{align}
Here the last equality holds since, by the formula for the inverse of a partitioned matrix, $(\underline{\tilde{\bm{X}}}^\top_{\text{full}} \underline{\tilde{\bm{X}}}_{\text{full}})^{-1} = \begin{psmallmatrix}
\underline{\bm{B}}_{11} & \underline{\bm{B}}_{12}\\
\underline{\bm{B}}_{21} & \underline{\bm{B}}_{22}\\
\end{psmallmatrix}$, where $\underline{\bm{B}}^{-1}_{22} = \underline{\bm{X}}^\top_{\text{full}}\underline{\bm{X}}_{\text{full}} - N \bar{\bm{X}}_{\text{full}}\bar{\bm{X}}^\top_{\text{full}} = N \underline{\bm{S}}_{\text{full}}.$ This completes the proof of the $\tilde{n}_{r-1} = 0$ case. The proof for the case where $\tilde{n}_{r-1} \geq 1$ and $\underline{\tilde{\bm{X}}}^\top_{r-1} \underline{\tilde{\bm{X}}}_{r-1}$ is invertible follows similar steps and hence is omitted.

We now consider the case where $\tilde{n}_{r-1} \geq 1$ and $\underline{\tilde{\bm{X}}}^\top_{r-1} \underline{\tilde{\bm{X}}}_{r-1}$ is not invertible. We denote $\bar{\bm{X}}^*_{r-1} = \frac{\bar{\bm{X}}_{r-1}+\epsilon\bar{\bm{X}}_{\text{full}}}{1+\epsilon}$ and $\underline{\bm{S}}^*_{r-1}  = (\frac{1}{\tilde{n}_{r-1}}\underline{\bm{X}}_{r-1}^\top \underline{\bm{X}}_{r-1} + \frac{\epsilon}{N} \underline{\bm{X}}_{\text{full}}^\top \underline{\bm{X}}_{\text{full}}) - (1+\epsilon)\bar{\bm{X}}^*_{r-1}\bar{\bm{X}}^{*\top}_{r-1}$.
The selected unit $i'$ satisfies,
\begin{equation}
    i' \in \argmax\limits_{i \in \mathcal{R}_{r-1}} (1, \bm{X}^\top_i) \Big(\frac{1}{\tilde{n}_{r-1}}\underline{\tilde{\bm{X}}}^\top_{r-1} \underline{\tilde{\bm{X}}}_{r-1} + \frac{\epsilon}{N} \underline{\tilde{\bm{X}}}^\top_{\text{full}} \underline{\tilde{\bm{X}}}_{\text{full}}  \Big)^{-1} \begin{psmallmatrix}
       1\\
       \bm{X}_i
       \end{psmallmatrix}
\end{equation}
Denoting $\underline{\tilde{\bm{G}}} = \begin{psmallmatrix}
\sqrt{\frac{1}{\tilde{n}_{r-1}}}\underline{\bm{X}}_{r-1}\\
\sqrt{\frac{\epsilon}{N}}\underline{\bm{X}}_{\text{full}}
\end{psmallmatrix}$, we have
\begin{align}
 &(1, \bm{X}^\top_i) \Big(\frac{1}{\tilde{n}_{r-1}}\underline{\tilde{\bm{X}}}^\top_{r-1} \underline{\tilde{\bm{X}}}_{r-1} + \frac{\epsilon}{N} \underline{\tilde{\bm{X}}}^\top_{\text{full}} \underline{\tilde{\bm{X}}}_{\text{full}}  \Big)^{-1} \begin{psmallmatrix}
       1\\
       \bm{X}_i
       \end{psmallmatrix} \nonumber\\
       &= (1, \bm{X}^\top_i) (\underline{\tilde{\bm{G}}}^\top \underline{\tilde{\bm{G}}})^{-1} \begin{psmallmatrix}
       1\\
       \bm{X}_i
       \end{psmallmatrix}  \nonumber\\
&=  (1, \bm{X}^\top_i) (\underline{\tilde{\bm{G}}}^\top \underline{\tilde{\bm{G}}})^{-1} \begin{psmallmatrix}
       0\\
       \bm{X}_i - \bar{\bm{X}}^*_{r-1}
       \end{psmallmatrix} + (1, \bm{X}^\top_i) (\underline{\tilde{\bm{G}}}^\top \underline{\tilde{\bm{G}}})^{-1} \begin{psmallmatrix}
       1\\
       \bar{\bm{X}}^*_{r-1}
       \end{psmallmatrix}     \nonumber \\
& = (0, (\bm{X}_i-\bar{\bm{X}}^*_{r-1})^\top) (\underline{\tilde{\bm{G}}}^\top \underline{\tilde{\bm{G}}})^{-1} \begin{psmallmatrix}
       0\\
       \bm{X}_i - \bar{\bm{X}}^*_{r-1}
       \end{psmallmatrix} + \frac{1}{1+\epsilon} \nonumber \\
& = (\bm{X}_i-\bar{\bm{X}}^*_{r-1})^\top (\underline{\bm{S}}^*_{r-1})^{-1} (\bm{X}_i-\bar{\bm{X}}^*_{r-1}) + \frac{1}{1+\epsilon}.
\end{align}
Here, the third equality holds since $\begin{psmallmatrix}
1 \\
\bar{\bm{X}}^*_{r-1}
\end{psmallmatrix} = \frac{1}{1+\epsilon}\underline{\tilde{\bm{G}}}^\top \underline{\tilde{\bm{G}}} \bm{e}_1$ and the fourth equality holds since $(\underline{\tilde{\bm{G}}}^\top \underline{\tilde{\bm{G}}})^{-1} = \begin{psmallmatrix}
\underline{\bm{B}}_{11} & \underline{\bm{B}}_{12}\\
\underline{\bm{B}}_{21} & \underline{\bm{B}}_{22}\\
\end{psmallmatrix}$, where $\underline{\bm{B}}^{-1}_{22} = (\frac{1}{\tilde{n}_{r-1}}\underline{\bm{X}}_{r-1}^\top \underline{\bm{X}}_{r-1} + \frac{\epsilon}{N} \underline{\bm{X}}_{\text{full}}^\top \underline{\bm{X}}_{\text{full}}) - (1+\epsilon) \bar{\bm{X}}^*_{r-1}\bar{\bm{X}}^{*\top}_{r-1} =  \underline{\bm{S}}^*_{r-1}.$ This completes the proof.
\end{proof}

\subsubsection*{Proof of Theorem \ref{thm:retrieve}}
\begin{proof}
(a) We first consider the setting of a standard block design where $N = BG$ (i.e., $c=1$). The blocks are labelled $1,2,...,B$. Here, the SOM is constructed by stacking $B$ independent random permutations of the `chunk' $(1,2,...,G)$. We will show that the choices made by the treatment groups in the FSM follow the assignment mechanism of an RBD.

Consider the first randomized chunk of the SOM, i.e., a random permutation of $(1,2,...,G)$. At the first stage of this randomized chunk, the choosing treatment group aims to maximize $(1, \bm{X}^\top_i) (\underline{\tilde{\bm{X}}}^\top_{\text{full}} \underline{\tilde{\bm{X}}}_{\text{full}})^{-1} \begin{psmallmatrix}
       1\\
       \bm{X}_i
       \end{psmallmatrix}$. Note that we can write $\underline{\tilde{\bm{X}}}_{\text{full}}$ as $\underline{\tilde{\bm{X}}}_{\text{full}} = \begin{psmallmatrix}
       \underline{\bm{D}}\\
       \vdots \\
       \underline{\bm{D}}
       \end{psmallmatrix}$, where $\underline{\bm{D}}_{B \times B} = \begin{psmallmatrix}
       1 & 1 & 0 & ... & 0 & 0 \\
       1 & 0 & 1 & ... & 0 & 0 \\
       \vdots & \vdots & \vdots & ... & \vdots & \vdots \\ 
       1 & 0 & 0 & ... & 0 & 1 \\
       1 & 0 & 0 & ... & 0 & 0 \\
       \end{psmallmatrix}$. Now, consider a transformation of the rows of the design matrix given by $\vardbtilde{\bm{X}}_i = (\underline{\bm{D}}^\top)^{-1} \begin{psmallmatrix}
       1\\
       \bm{X}_i
       \end{psmallmatrix} $. The transformed design matrix is $\underline{\vardbtilde{\bm{X}}}_{\text{full}} =  \underline{\tilde{\bm{X}}}_{\text{full}} \underline{\bm{D}}^{-1} = \begin{psmallmatrix}
       \underline{\bm{I}}_B\\
       \vdots \\
       \underline{\bm{I}}_B
       \end{psmallmatrix}$. We note that the $\vardbtilde{\bm{X}}_i$s nothing but standard unit vectors. Now, 
\begin{align}
(1, \bm{X}^\top_i) (\underline{\tilde{\bm{X}}}^\top_{\text{full}} \underline{\tilde{\bm{X}}}_{\text{full}})^{-1} \begin{psmallmatrix}
       1\\
       \bm{X}_i
       \end{psmallmatrix} & = \vardbtilde{\bm{X}}^\top_i (\underline{\vardbtilde{\bm{X}}}^\top_{\text{full}} \underline{\vardbtilde{\bm{X}}}_{\text{full}})^{-1}  \vardbtilde{\bm{X}}_i.
\end{align}       
Therefore, the selection function remains the same under the above transformation. Now, $\vardbtilde{\bm{X}}^\top_i (\underline{\vardbtilde{\bm{X}}}^\top_{\text{full}} \underline{\vardbtilde{\bm{X}}}_{\text{full}})^{-1}  \vardbtilde{\bm{X}}_i = \frac{1}{G} \vardbtilde{\bm{X}}^\top_i \vardbtilde{\bm{X}}_i = \frac{1}{G}$ for all $i$, which essentially implies that the choosing group has no preference among the units for selection and hence chooses any one of the $N$ units randomly. Similarly, at the subsequent stages of this randomized chunk, the corresponding choosing groups select one of the remaining units randomly. 

Next, we consider the second randomized chunk of the SOM. Without loss of generality, suppose treatment 1 gets to choose first in this chunk. 
Also, without loss of generality, suppose that in its first choice, treatment 1 had selected a unit from block 1.
We claim that in this selection, treatment 1 will choose one of the remaining units randomly from any block other than block 1, which respects the assignment mechanism of an RBD. 

To prove the claim, we first consider the objective function at this stage.
Treatment 1 aims to maximize $(1, \bm{X}^\top_i) \Big(\frac{1}{\tilde{n}_{r-1}}\underline{\tilde{\bm{X}}}^\top_{r-1} \underline{\tilde{\bm{X}}}_{r-1} + \frac{\epsilon}{N} \underline{\tilde{\bm{X}}}^\top_{\text{full}} \underline{\tilde{\bm{X}}}_{\text{full}}  \Big)^{-1} \begin{psmallmatrix}
       1\\
       \bm{X}_i
       \end{psmallmatrix}$. Here, we denote the current stage by $r$. Using the same transformation as in the case of the first chunk, we can write the objective function as $\vardbtilde{\bm{X}}^\top_i \Big(\frac{1}{\tilde{n}_{r-1}}\underline{\vardbtilde{\bm{X}}}^\top_{r-1} \underline{\vardbtilde{\bm{X}}}_{r-1} + \frac{\epsilon}{N} \underline{\vardbtilde{\bm{X}}}^\top_{\text{full}} \underline{\vardbtilde{\bm{X}}}_{\text{full}}  \Big)^{-1} \vardbtilde{\bm{X}}_i$, where $\underline{\vardbtilde{\bm{X}}}_{r-1}  = \underline{\tilde{\bm{X}}}_{r-1} \underline{\bm{D}}^{-1}$. Since $\underline{\vardbtilde{\bm{X}}}^\top_{\text{full}} \underline{\vardbtilde{\bm{X}}}_{\text{full}}  = G \underline{\bm{I}}_B$, it is equivalent to maximize
       \begin{align}
          \vardbtilde{\bm{X}}^\top_i \Big(\underline{\bm{I}}_b +  \frac{B}{ \tilde{n}_{r-1}\epsilon G}\underline{\vardbtilde{\bm{X}}}^\top_{r-1} \underline{\vardbtilde{\bm{X}}}_{r-1}\Big )^{-1} \vardbtilde{\bm{X}}_i &=           \vardbtilde{\bm{X}}^\top_i \Big(\underline{\bm{I}}_b +  \delta \underline{\vardbtilde{\bm{X}}}^\top_{r-1} \underline{\vardbtilde{\bm{X}}}_{r-1}\Big )^{-1} \vardbtilde{\bm{X}}_i \label{eqA:3.3_0} \\
          & = \vardbtilde{\bm{X}}^\top_i \Big\{ \underline{\bm{I}}_b -  \delta \underline{\vardbtilde{\bm{X}}}^\top_{r-1}(\underline{\bm{I}}_{\tilde{n}_{r-1}} +  \delta \underline{\vardbtilde{\bm{X}}}_{r-1} \underline{\vardbtilde{\bm{X}}}^\top_{r-1})^{-1}\underline{\vardbtilde{\bm{X}}}_{r-1} \Big \} \vardbtilde{\bm{X}}_i. \label{eqA:3.3_1}
       \end{align}
Here, $\delta = \frac{B}{\tilde{n}_{r-1} \epsilon G}$. The final equality holds by the Woodbury matrix identity. Now, in this case, $\underline{\vardbtilde{\bm{X}}}_{r-1} = (1,0,...,0)$ (since treatment 1 has only selected one unit from block 1 up to this stage). So, the objective function in Equation \ref{eqA:3.3_1} equals $1 - \frac{\delta}{1+\delta} \vardbtilde{\bm{X}}^\top_i (\underline{\vardbtilde{\bm{X}}}^\top_{r-1} \underline{\vardbtilde{\bm{X}}}_{r-1})  \vardbtilde{\bm{X}}^\top_i$. Since $\delta>0$, it is equivalent to minimize $\vardbtilde{\bm{X}}^\top_i (\underline{\vardbtilde{\bm{X}}}^\top_{r-1} \underline{\vardbtilde{\bm{X}}}_{r-1})  \vardbtilde{\bm{X}}^\top_i = \vardbtilde{\bm{X}}^\top_i \begin{psmallmatrix}
1 & \bm{0}^\top_{1\times (B-1)}\\
\bm{0}_{(B-1) \times 1} & \underline{\bm{0}}_{(B-1)\times (B-1)} 
\end{psmallmatrix}\vardbtilde{\bm{X}}_i$, which takes the value $0$ for a unit in any block other than block 1 and $1$ for a unit in block $1$. This proves the claim for treatment 1. Moreover, by similar reasoning, the claim holds for all the other treatment groups in this randomized chunk.

Next, we consider a general randomized chunk of the SOM. Once again, without loss of generality, suppose treatment 1 gets to choose first in this chunk. Also, for simplicity of exposition and without loss of generality, suppose treatment 1 has already selected from blocks $1,2,...,b$, implying that $\tilde{n}_{r-1} = b$ and $\underline{\vardbtilde{\bm{X}}}_{r-1} = \begin{psmallmatrix}
\underline{\bm{I}}_b & \underline{\bm{0}}_{b \times (B-b)}
\end{psmallmatrix}$. This form of $\underline{\vardbtilde{\bm{X}}}_{r-1}$, along with Equation \ref{eqA:3.3_1} implies that it is equivalent to minimize $\vardbtilde{\bm{X}}^\top_i (\underline{\vardbtilde{\bm{X}}}^\top_{r-1} \underline{\vardbtilde{\bm{X}}}_{r-1})  \vardbtilde{\bm{X}}^\top_i = \vardbtilde{\bm{X}}^\top_i \begin{psmallmatrix}
\underline{\bm{I}}_{b} & \underline{\bm{0}}_{b\times (B-b)}\\
\underline{\bm{0}}^\top_{(B-b) \times b} & \underline{\bm{0}}_{(B-b)\times (B-b)} 
\end{psmallmatrix}\vardbtilde{\bm{X}}_i$, which is minimized for any unit $i$ belonging to the blocks $b+1,...,B$. This shows that at this stage, treatment 1 randomly chooses a unit from a block other than the blocks it has already chosen from. By similar reasoning, at subsequent stages of this randomized chunk, the choosing group follows the same selection strategy for their own group. This completes the proof of the theorem for the setting of a standard block design. 

We now prove the theorem for the general block design setting with $N = cBG$, $c>1$. The proof strategy is exactly the same as the $c=1$ setting. Here the SOM is generated by randomly permuting the chunk $(1,2,...,G)$ $B\times c$ times. Once the selections are completed for the the first $B$ chunks, the resulting assignment resembles that of a standard RBD (by the previous proof), where each treatment group randomly chooses exactly one unit from each block. For the $(B+1)$th chunk, suppose, without loss of generality, that treatment 1 gets to choose first. At this stage (denoted by stage $r$), treatment 1 tries to maximize,
\begin{align}
  (1, \bm{X}^\top_i) (\underline{\tilde{\bm{X}}}^\top_{r-1} \underline{\tilde{\bm{X}}}_{r-1})^{-1} \begin{psmallmatrix}
       1\\
       \bm{X}_i
       \end{psmallmatrix} & = \vardbtilde{\bm{X}}^\top_i (\underline{\vardbtilde{\bm{X}}}^\top_{r-1} \underline{\vardbtilde{\bm{X}}}_{r-1})^{-1}  \vardbtilde{\bm{X}}_i \nonumber \\ 
       & = \vardbtilde{\bm{X}}^\top_i  \vardbtilde{\bm{X}}_i = 1, 
\end{align}
where the penultimate equality holds since $\underline{\vardbtilde{\bm{X}}}_{r-1} = \underline{\bm{I}}_B$. Thus, similar to the first randomized chunk in the setting of $c=1$, treatment 1 (and the other treatments) randomly chooses one of the available units. 

Finally, we consider a general chunk. Without loss of generality, suppose treatment 1 gets to choose first in this chunk. We can write the corresponding transformed design matrix $\underline{\vardbtilde{\bm{X}}}_{r-1}$ as 
\begin{align}
    \underline{\vardbtilde{\bm{X}}}_{r-1} = \begin{psmallmatrix}
    &\underline{\bm{I}}_B \\
    &\underline{\bm{I}}_B \\
    &\vdots \\
    &\underline{\bm{I}}_B \\
    \underline{\bm{I}}_{b} & & \underline{\bm{0}}_{b\times (B-b)}
    \end{psmallmatrix}.
\end{align}
Here, without loss of generality, we have assumed that treatment 1 has chosen $c_0+b$ times from the first $b$ blocks and $c_0$ times from the remaining blocks, where $c_0<c$. This implies that treatment 1 aims to maximize.
\begin{align}
\vardbtilde{\bm{X}}^\top_i (\underline{\vardbtilde{\bm{X}}}^\top_{r-1} \underline{\vardbtilde{\bm{X}}}_{r-1})^{-1}  \vardbtilde{\bm{X}}_i & = \vardbtilde{\bm{X}}^\top_i \Big\{ c_0 \underline{\bm{I}}_{B} + \begin{psmallmatrix}
\underline{\bm{I}}_{b} \\
\underline{\bm{0}}^\top_{(B-b) \times b}
\end{psmallmatrix} \begin{psmallmatrix}
\underline{\bm{I}}_{b} & \underline{\bm{0}}_{b \times (B-b)}
\end{psmallmatrix} \Big\}^{-1}     \vardbtilde{\bm{X}}_i,    
\end{align}
which has the same form as the objective function in Equation \ref{eqA:3.3_0} in the $c=1$ setting. Thus, following similar arguments as in the $c=1$ setting, we conclude that at this stage, treatment 1 selects a unit randomly from blocks $b+1,...,B$, which conforms to the assignment mechanism of an RBD. Also, at subsequent stages of the randomized chunk, the choosing group follows the same selection strategy for their own group. This completes the proof of the theorem.

\noindent (b) With two groups of equal sizes, the SOM consists of successive random permutations of the `chunk' $(1,2)$. By Theorem \ref{thm:mahal}, for the first pair of stages of selection, the objective function (to maximize) is given by  
\begin{align}
(\bm{X}_i - \bar{\bm{X}}_{\text{full}})^\top (\underline{\bm{S}}_{\text{full}})^{-1} (\bm{X}_i - \bar{\bm{X}}_{\text{full}}).    \label{eqA:3.3_2}   
\end{align}
Under the assumption of identical twins and continuous data generating distributions, with probability 1, there are exactly two units (one being a twin of the other), whose common covariate value $\bm{X}^{(1)}$ (say) maximizes the objective function in Equation \ref{eqA:3.3_2}. Therefore, the choosing group at the first stage selects one of these two identical twins randomly, and in the next stage, the other treatment selects the remaining twin. This respects the assignment mechanism of a matched-pair design.

Consider the next pairs of stages. The objective function of the choosing treatment group is given by:
\begin{align}
    &(\bm{X}_i - \frac{1}{1+\epsilon}\bm{X}^{(1)})^\top \Big\{\bm{X}^{(1)}\bm{X}^{(1)\top} + \frac{\epsilon}{N}\underline{\bm{X}}^\top_{\text{full}}\underline{\bm{X}}_{\text{full}} - (1+\epsilon) \bm{X}^{(1)}\bm{X}^{(1)\top}   \Big\}^{-1}  (\bm{X}_i - \frac{1}{1+\epsilon}\bm{X}^{(1)})  \label{eqA:3.3_4}
\end{align}
Similar to the previous case, here also we have (with probability 1) exactly two units, one being a twin of the other, whose common covariate value $\bm{X}^{(2)}$ maximizes the objective function in Equation \ref{eqA:3.3_4}. Thus, the choosing group at the first stage of this pair selects one of these two twins randomly, and in the next stage, the other treatment chooses the remaining twin. 
Proceeding in this manner, it follows that, at the end of the selection process, each treatment group ends up selecting one twin randomly from $\frac{N}{2}$ identical twins, which is equivalent to a matched-pair design. This completes the proof.

\end{proof}

\subsubsection*{Proof of Proposition \ref{prop:aopt}}
\begin{proof}
The A-optimal selection function aims to minimize
\begin{align}
    &\text{trace }\Big\{\underline{\bm{T}}(\underline{\tilde{\bm{X}}}^\top_{r,i} \underline{\tilde{\bm{X}}}_{r,i})^{-1} \Big\}\\
    &= \text{trace}\Big[\underline{\bm{T}} \{\underline{\tilde{\bm{X}}}^\top_{r-1} \underline{\tilde{\bm{X}}}_{r-1} + \begin{psmallmatrix}
    1\\
    \bm{X}_i
    \end{psmallmatrix} (1, \bm{X}^\top_i)\}^{-1} \Big] \nonumber\\
    & = \text{trace} \Big\{\underline{\bm{T}} (\underline{\tilde{\bm{X}}}^\top_{r-1} \underline{\tilde{\bm{X}}}_{r-1})^{-1} - \underline{\bm{T}} \frac{(\underline{\tilde{\bm{X}}}^\top_{r-1} \underline{\tilde{\bm{X}}}_{r-1})^{-1} \begin{psmallmatrix}
    1\\
    \bm{X}_i
    \end{psmallmatrix} (1, \bm{X}^\top_i)    (\underline{\tilde{\bm{X}}}^\top_{r-1} \underline{\tilde{\bm{X}}}_{r-1})^{-1}}{1+ (1, \bm{X}^\top_i) (\underline{\tilde{\bm{X}}}^\top_{r-1} \underline{\tilde{\bm{X}}}_{r-1})^{-1} \begin{psmallmatrix}
    1\\
    \bm{X}_i
    \end{psmallmatrix} }    \Big\}  \nonumber\\
    & = \text{trace}\{\underline{\bm{T}} (\underline{\tilde{\bm{X}}}^\top_{r-1} \underline{\tilde{\bm{X}}}_{r-1})^{-1}\} - \text{trace}\Big\{\underline{\bm{T}} \frac{(\underline{\tilde{\bm{X}}}^\top_{r-1} \underline{\tilde{\bm{X}}}_{r-1})^{-1} \begin{psmallmatrix}
    1\\
    \bm{X}_i
    \end{psmallmatrix} (1, \bm{X}^\top_i)    (\underline{\tilde{\bm{X}}}^\top_{r-1} \underline{\tilde{\bm{X}}}_{r-1})^{-1}}{1+ (1, \bm{X}^\top_i) (\underline{\tilde{\bm{X}}}^\top_{r-1} \underline{\tilde{\bm{X}}}_{r-1})^{-1} \begin{psmallmatrix}
    1\\
    \bm{X}_i
    \end{psmallmatrix} }    \Big\} \nonumber\\
    & = \text{trace}\{\underline{\bm{T}} (\underline{\tilde{\bm{X}}}^\top_{r-1} \underline{\tilde{\bm{X}}}_{r-1})^{-1}\} - \text{trace}\Big\{  \frac{(1, \bm{X}^\top_i) (\underline{\tilde{\bm{X}}}^\top_{r-1} \underline{\tilde{\bm{X}}}_{r-1})^{-1} \underline{\bm{T}}(\underline{\tilde{\bm{X}}}^\top_{r-1} \underline{\tilde{\bm{X}}}_{r-1})^{-1} \begin{psmallmatrix}
    1\\
    \bm{X}_i
    \end{psmallmatrix}}{1+ (1, \bm{X}^\top_i) (\underline{\tilde{\bm{X}}}^\top_{r-1} \underline{\tilde{\bm{X}}}_{r-1})^{-1} \begin{psmallmatrix}
    1\\
    \bm{X}_i
    \end{psmallmatrix} } \Big\}    \label{eqA:3.4_1}
\end{align}
\end{proof}
Here the second equality holds due to the Sherman-Morrison-Woodbury formula, the third and fourth equality hold due to the linearity and cyclicality of $\text{trace}(\cdot)$, respectively. Equation \ref{eqA:3.4_1} shows that it is equivalent to maximize $\text{trace}\Big\{  \frac{(1, \bm{X}^\top_i) (\underline{\tilde{\bm{X}}}^\top_{r-1} \underline{\tilde{\bm{X}}}_{r-1})^{-1} \underline{\bm{T}}(\underline{\tilde{\bm{X}}}^\top_{r-1} \underline{\tilde{\bm{X}}}_{r-1})^{-1} \begin{psmallmatrix}
    1\\
    \bm{X}_i
    \end{psmallmatrix}}{1+ (1, \bm{X}^\top_i) (\underline{\tilde{\bm{X}}}^\top_{r-1} \underline{\tilde{\bm{X}}}_{r-1})^{-1} \begin{psmallmatrix}
    1\\
    \bm{X}_i
    \end{psmallmatrix} } \Big\}$. This completes the proof.

\begin{proposition}\normalfont
For $n_1 = n_2 = ... = n_G$, the difference-in-means statistic between treatment $g'$ and $g''$ is unbiased for $\text{SATE}_{g',g''}$. 
\label{fsm_prop:unbiased}
\end{proposition}
\begin{proof}
With equal-sized groups, by symmetry, every unit has an equal chance of belonging to one of the $G$ treatment groups. That is, $P(Z_i = g) = \frac{1}{G}$ for all $g \in \{1,2,...,G\}$. Therefore,
\begin{align}
    \mathbb{E}\Big\{\frac{1}{n_g}\sum_{i:Z_i = g}Y^{\text{obs}}_i|\bm{Y}(g)\Big\} &=     \mathbb{E}\Big\{\frac{G}{N}\sum_{i = 1}^{N} \mathbbm{1}(Z_i = g)Y_i(g)|\bm{Y}(g)\Big\} \nonumber \\
    & = \frac{G}{N}\sum_{i = 1}^{N} P(Z_i = g)Y_i(g) \nonumber \\
    & = \frac{1}{N}\sum_{i = 1}^{N} Y_i(g). \label{eq:expectation}
\end{align}
The proposition follows from Equation \ref{eq:expectation} applied to treatment groups $g'$ and $g''$.
\end{proof}

\subsection{Properties of D-optimal selection function}
\label{sec_properties_dopt}
\subsection{Affine invariance and covariate balance}
\begin{theorem}\normalfont
\begin{enumerate}[label=(\alph*)]
    \item The FSM with the D-optimal selection function is invariant under affine transformations of the covariate vector. 

    \item For continuous, symmetrically distributed covariates and two groups of equal size, the FSM with the D-optimal selection function almost surely produces exact mean-balance on all even transformations of the centered covariate vector.
\end{enumerate}
\label{thm:dopt_properties}
\end{theorem}

\subsubsection*{Proof of Theorem \ref{thm:dopt_properties}}
\begin{proof}
(a) We consider the case where $\tilde{n}_{r-1} \geq 1$ and $\underline{\tilde{\bm{X}}}^\top_{r-1} \underline{\tilde{\bm{X}}}_{r-1}$ is invertible. The proofs for the other two cases are similar. By Theorem \ref{thm:mahal}, in this case, the chosen unit $i'$ satisfies,
\begin{equation}
    i' \in \argmax\limits_{i \in \mathcal{R}_{r-1}} (\bm{X}_i - \bar{\bm{X}}_{r-1})^\top (\underline{\bm{S}}_{r-1})^{-1} (\bm{X}_i - \bar{\bm{X}}_{r-1}).
\end{equation}
Consider an affine transformation of the covariate $\bm{X}$ given by $\bm{U} = \underline{\bm{A}}\bm{X} + \bm{b}$, where $\underline{\bm{A}}$ is a $k \times k$ invertible matrix and $\bm{b}$ is a vector of dimension k. Let the corresponding values of $\bar{\bm{X}}_{r-1}$ and $\underline{\bm{S}}_{r-1}$ be $\bar{\bm{U}}_{r-1}$ and $\underline{\bm{S}}_{U,r-1}$, respectively.
We observe that,
\begin{align}
    (\bm{U}_i - \bar{\bm{U}}_{r-1})^\top (\underline{\bm{S}}_{U,r-1})^{-1} (\bm{U}_i - \bar{\bm{U}}_{r-1}) & = \{\underline{\bm{A}}(\bm{X}_i - \bar{\bm{X}}_{r-1})\}^\top (\underline{\bm{A}}\underline{\bm{S}}_{r-1} \underline{\bm{A}}^\top)^{-1} \underline{\bm{A}}(\bm{X}_i - \bar{\bm{X}}_{r-1}) \nonumber \\
    & = (\bm{X}_i - \bar{\bm{X}}_{r-1})^\top (\underline{\bm{S}}_{r-1})^{-1} (\bm{X}_i - \bar{\bm{X}}_{r-1}).
\end{align}
This shows that the D-optimal selection function remains unchanged under affine transformations and hence, FSM with the D-optimal selection function is affine invariant.

\noindent (b) The in-sample symmetry of the data essentially implies that if $\bm{X}$ belongs to the sample, then $-\bm{X}$ also belongs to the sample. Moreover, by the assumption of a continuous data generating distribution, with probability 1, the covariate values are different up to reflection. Now, consider an even transformation $g(\cdot)$, i.e., $g(-\bm{X}) = g(\bm{X})$.
With two groups of equal sizes, the SOM consists of successive random permutations of the `chunk' $(1,2)$. By Theorem \ref{thm:mahal}, for the first pair of stages of selection, the objective function (to maximize) is given by  
\begin{align}
(\bm{X}_i - \bar{\bm{X}}_{\text{full}})^\top (\underline{\bm{S}}_{\text{full}})^{-1} (\bm{X}_i - \bar{\bm{X}}_{\text{full}}) & =  \bm{X}^\top_i (\underline{\bm{S}}_{\text{full}})^{-1} \bm{X}_i. \label{eqA:3.2_2}   
\end{align}
It follows that, if a unit in the sample with covariate $\bm{X}^{(1)}$ maximizes the objective function in Equation \ref{eqA:3.2_2}, then so does the unit with covariate $-\bm{X}^{(1)}$. Moreover, due to the continuous data generating distribution, with probability 1, these are the only two units that maximize this objective function. Therefore, if treatment 1 selects the unit with covariate $\bm{X}^{(1)}$, treatment 2 selects the unit with covariate $-\bm{X}^{(1)}$, and vice-versa. This preserves exact balance on $g(\bm{X})$.

Now, consider the next pair of stages. Without loss of generality, suppose treatment 1 had chosen a unit with covariate $\bm{X}^{(1)}$ and treatment 2 had chosen a unit with covariate $-\bm{X}^{(1)}$ in their respective previous choices. Also, without loss of generality, assume that in this pair of stages, treatment 1 gets to choose first. By Theorem \ref{thm:mahal}, treatment 1 aims to maximize, 
\begin{align}
    &(\bm{X}_i - \bar{\bm{X}}^*_{r-1})^\top (\underline{\bm{S}}^*_{r-1})^{-1} (\bm{X}_i - \bar{\bm{X}}_{r-1}^*) \nonumber\\
    & = \{\bm{X}_i - \frac{1}{1+\epsilon}\bm{X}^{(1)}\}^\top \Big\{\bm{X}^{(1)}\bm{X}^{(1)\top} + \frac{\epsilon}{N}\underline{\bm{X}}^\top_{\text{full}}\underline{\bm{X}}_{\text{full}} - (1+\epsilon) \bm{X}^{(1)}\bm{X}^{(1)\top}   \Big\}^{-1}  \{\bm{X}_i - \frac{1}{1+\epsilon}\bm{X}^{(1)}\}. \label{eqA:3.2_3}    
\end{align}
Also, during treatment 2's turn in this pair of stages, it tries to maximize
\begin{align}
    &(\bm{X}_i + \frac{1}{1+\epsilon}\bm{X}^{(1)})^\top \Big\{\bm{X}^{(1)}\bm{X}^{(1)\top} + \frac{\epsilon}{N}\underline{\bm{X}}^\top_{\text{full}}\underline{\bm{X}}_{\text{full}} - (1+\epsilon) \bm{X}^{(1)}\bm{X}^{(1)\top}   \Big\}^{-1}  (\bm{X}_i + \frac{1}{1+\epsilon}\bm{X}^{(1)}) \nonumber \\
    & = \{(-\bm{X}_i) - \frac{1}{1+\epsilon}\bm{X}^{(1)}\}^\top \Big\{\bm{X}^{(1)}\bm{X}^{(1)\top} + \frac{\epsilon}{N}\underline{\bm{X}}^\top_{\text{full}}\underline{\bm{X}}_{\text{full}} - (1+\epsilon) \bm{X}^{(1)}\bm{X}^{(1)\top}   \Big\}^{-1}  \{(-\bm{X}_i) - \frac{1}{1+\epsilon}\bm{X}^{(1)}\}. \label{eqA:3.2_4}
\end{align}
Equations \ref{eqA:3.2_3} and \ref{eqA:3.2_4} imply that if treatment 1 chooses a unit with covariate value $\bm{X}^{(2)}$, then with probability 1, treatment 2 chooses the unit with covariate value $-\bm{X}^{(2)}$, and vice versa. This shows that, at the end of the second pair of stages in the SOM, exact mean balance on $g(\bm{X})$ is preserved. Proceeding in this manner it follows that, at the end of the selection process, with probability 1, both the treatment groups will have exact balance on $g(\bm{X})$. This completes the proof. 

\end{proof}

It follows from Theorem \ref{thm:dopt_properties}(a) that, for any SOM, the choices made by each treatment group remain unchanged even if the covariate vectors are transformed via an affine transformation (e.g., changing the units of measurement of the covariates).
Therefore, the FSM with the D-optimal selection function self-standardizes the covariates. 
In addition, if the covariate vector is symmetrically distributed in the sample, then by Theorem \ref{thm:dopt_properties}(b), the FSM exactly balances even transformations such as the second, fourth order moments, and the pairwise products of the centerd covariates. An implication of Theorem \ref{thm:dopt_properties}(b) is that, for covariates drawn from symmetric continuous distributions (such as the Normal, t, and Laplace distributions), the FSM tends to balance all these transformations due to the approximate symmetry of the covariates in the sample. The choice of the D-optimal selection function is thus robust in the sense that it allows the FSM to balance a family of transformations of the covariate vector by design, without explicitly including them in the assumed linear model nor requiring the specification of tuning parameters.

\subsubsection{Connection to A-optimality}
\label{sec_connection}

The original FSM used a criterion based on A-optimality as the selection function (see \citealt{morris1979finite}). In this section, we compare the A-and D-optimal selection functions. The A-optimal selection function requires prespecifying a \textit{policy matrix} $\underline{\bm{P}}_{p \times (k+1)}$ and a corresponding vector of \textit{policy weights} $\bm{w}_{p\times 1}$. Here, $\underline{\bm{P}}$ transforms the original vector of regression coefficients to a vector of $p$ linear combinations that are of policy interest, and $\bm{w}$ assigns weights to each combination according to their importance. 
If treatment 1 gets to choose at the $r$th stage, then this criterion selects the unit that minimizes the resulting $\text{trace}\Big\{\underline{\bm{T}} (\underline{\tilde{\bm{X}}}^\top_{r,i} \underline{\tilde{\bm{X}}}_{r,i})^{-1}\Big\}$, where $\underline{\bm{T}} = \underline{\bm{P}}^\top \text{diag}(\bm{w}) \underline{\bm{P}}$.
Proposition \ref{prop:aopt} shows an equivalent characterization of the A-optimal selection function.
\begin{proposition} \normalfont
Let treatment 1 be the choosing group at the $r$th stage. Assume that $\tilde{n}_{r-1} \geq 1$ and $\underline{\tilde{\bm{X}}}^\top_{r-1} \underline{\tilde{\bm{X}}}_{r-1}$ is invertible. The A-optimal selection function chooses unit $i'$ with covariate vector $\bm{X}_{i'} \in \mathbb{R}^k$, where 
 $     i' \in \argmax\limits_{i \in \mathcal{R}_{r-1}}\frac{(1, \bm{X}^\top_i) (\underline{\tilde{\bm{X}}}^\top_{r-1} \underline{\tilde{\bm{X}}}_{r-1})^{-1} \underline{\bm{T}} (\underline{\tilde{\bm{X}}}^\top_{r-1} \underline{\tilde{\bm{X}}}_{r-1})^{-1}  \begin{psmallmatrix}
       1\\
       \bm{X}_i
       \end{psmallmatrix}}{1+(1, \bm{X}^\top_i) (\underline{\tilde{\bm{X}}}^\top_{r-1} \underline{\tilde{\bm{X}}}_{r-1})^{-1}  \begin{psmallmatrix}
       1\\
       \bm{X}_i
       \end{psmallmatrix}}.$

\label{prop:aopt}
\end{proposition}
The A-optimality criterion provides a family of selection functions depending on $\underline{\bm{P}}$ and $\bm{w}$. For some choices of $\underline{\bm{P}}$ and $\bm{w}$, the selection function is not affine invariant , e.g., $\underline{\bm{P}} = \bm{I}$ and $\bm{w} = (1,1,...,1)^\top$, while for other choices it is, e.g., $\underline{\bm{P}} = \underline{\tilde{\bm{X}}}_{\text{full}}$ and $\bm{w} = (1,1,...,1)^\top$. In particular, the A-optimal selection function with $\underline{\bm{P}} = \underline{\tilde{\bm{X}}}_{\text{full}}$ and $\bm{w} = (1,1,...,1)^\top$ is closely related to the D-optimal selection function. To see this, consider a case where in the selection process, the design matrices in each treatment group scale similarly relative to the design matrix in the full sample, i.e., $(\underline{\tilde{\bm{X}}}^\top_{r-1} \underline{\tilde{\bm{X}}}_{r-1})^{-1} = c_r (\underline{\tilde{\bm{X}}}^\top_{\text{full}} \underline{\tilde{\bm{X}}}_{\text{full}})^{-1}$ for some constant $c_r>0$. In this case, the A-optimal selection function chooses unit $i'$ such that $i' \in \argmax\limits_{i \in \mathcal{R}_{r-1}}(1, \bm{X}^\top_i) (\underline{\tilde{\bm{X}}}^\top_{r-1} \underline{\tilde{\bm{X}}}_{r-1})^{-1} \begin{psmallmatrix}
       1\\
       \bm{X}_i
       \end{psmallmatrix}  \iff i' \in \argmax\limits_{i \in \mathcal{R}_{r-1}} (\bm{X}_i - \bar{\bm{X}}_{r-1})^\top (\underline{\bm{S}}_{r-1})^{-1} (\bm{X}_i - \bar{\bm{X}}_{r-1})$, which is equivalent to the D-optimal selection function. Hence, in this case, the FSM under the D-optimal and A-optimal selection functions make similar choices of units.

\subsection{Algorithms for constructing an SOM}

\subsubsection{The SCOMARS algorithm}
\label{secA_scomars}
Consider a setting with $G=2$ treatment groups of arbitrary sizes $n_1$ and $n_2$. Let $W_r$ be the binary indicator for selection of group 1 stage $r$, $r \in \{1, 2, …, N\}$, with $p_r := P(W_r=1)$ being the marginal probability of selection at stage $r$. Write $S_r := \sum_{j=1}^r W_j$  and $F_r := \mathbb{E}(S_r) =  \sum_{j=1}^{r}p_j$. A treatment assignment is sequentially controlled if $|S_r - F_r|< 1$ for all $r \in \{1, 2, …, N\}$.

The SCOMARS algorithm proceeds as follows:
\begin{itemize}
\item Stage 1, $P(W_1 = 1) = p_1$.
\item Stage $r \geq 2$, $P(W_{r}=1|S_{r-1} = s_{r-1}) = P \Big\{ U \leq \frac{p_{r} - \max(0,s_{r-1} - F_{r-1})}{1 - |s_{r-1} - F_{r-1}|} \Big\}$, where $U \sim \text{Unif}(0,1)$.
\end{itemize}
This algorithm satisfies the \textrm{sequentially controlled} condition, $|S_r - F_r|< 1$ for all $r \in \{1, 2, …, N\}$ (\citealt{morris1983sequentially}). It is Markovian because the probability of selection at stage $r$ depends solely on stage $r-1$.

\subsubsection{SOM for multi-group experiments}
We first define the randomized chunk algorithm for generating an SOM for multi-group experiments with equal group sizes. 
\begin{definition}[Randomized chunk algorithm] \normalfont
Suppose $n_1 = n_2 = ... = n_G$. The randomized chunk algorithm generates an SOM by generating and stacking $\frac{N}{G}$ independent random permutations of the `chunk' $(1,2,...,G)$. 
\label{def_randchunk}
\end{definition}
For example, with $N = 12$, $g = 3$, $n_1 = n_2 = n_3 = 4$, one instance of an SOM generated using randomized chunk is $(\underbrace{2,1,3},\underbrace{1,2,3},\underbrace{2,1,3},\underbrace{2,3,1})^\top$.

The following proposition shows that the randomized chunk algorithm is sequentially controlled.
\begin{proposition}
\label{prop:randchunk}
\normalfont
For $G \geq 2$ and $n_1 = n_2 = ... = n_G$, the randomized chunk algorithm satisfies $|S_{ig} - F_{ig}| \leq \frac{G-1}{G}<1$ for all $g \in \{1,2,...,G\}$. 
\label{prop_randchunk}
\end{proposition}
\begin{proof}
Let $S_{ig}$ and $F_{ig}$ be the same as defined in Section \ref{sec_multi} ($i \in \{1,2,...,N\}$, $g \in \{1,2,...,G\}$). For equal sized treatment groups, $F_{ig} = \frac{i}{G}$. Now, without loss of generality, it suffices to show that $|S_{i1} - F_{i1}|\leq \frac{G-1}{G}$ for all $i \in \{1,2,...,N\}$. Consider the first chunk in the SOM, which is a random permutation of $(1,2,...,G)$. If treatment 1 appears in position $i^* \in \{1,2,...,G\}$ the permutation ($j \in \{1,2,...,G\}$), then
\begin{equation}
  |S_{i1} - F_{i1}| =
    \begin{cases}
      \frac{i}{G} & \text{if $i\in \{1,...,i^*-1\}$}\\
      1-\frac{i}{G} & \text{if $i\in \{i^*,...,G\}$.}
    \end{cases}     
\end{equation}
In each case, $|S_{i1} - F_{i1}|\leq \frac{G-1}{G}$ for all $i \in \{1,2,...,G\}$. Moreover, since $|S_{G1} - F_{G1}|=0$, the SOM restarts itself after the first chunk. Hence, we can conclude that $|S_{i1} - F_{i1}| \leq \frac{G-1}{G}$ for all $i \in \{1,2,...,G\}$. This completes the proof. 

\end{proof}
Below we describe two algorithms to generate an SOM for multi-group experiments and show that they are sequentially controlled. The key idea in these algorithms is the formation of `supergroups', i.e., combination of one or more treatment groups. For example, with $g = 3$, $n_1 = 10, n_2 = 20, n_3 = 30$, one can consider two supergroups, namely $\{1,2\}$ of size $10 + 20 = 30$ and $\{3\}$ of size $30$.

\begin{theorem}\normalfont
\label{thm:multiscomars1}
For $1\leq G_1 \leq G-1$, let $n_1=n_2 =...=n_{G_1} \neq n^{(1)}$, and $n_{G_1+1}=n_{G_1+2} =...=n_{G} = n^{(2)}$, where $n^{(1)} \neq n^{(2)}$. Consider the following three-stage algorithm.
\begin{enumerate}
    \item Run SCOMARS with supergroups $\{1,...,G_1\}$ and $\{G_1+1,...,G\}$ to generate an SOM at the supergroup level.
    \item Consider the locations of the SOM in step 1 where supergroup $\{1,...,G_1\}$ chooses. Then, use randomized chunk to obtain the selection orders at the levels of the original groups in those locations.
    \item Consider the locations of the SOM in step 1 where supergroup $\{G_1 + 1,...,G\}$ chooses. Then, use randomized chunk to obtain the selection orders at the levels of the original groups in those locations.
\end{enumerate}
The above SOM generating algorithm is sequentially controlled.
\end{theorem}
We first prove a special case of Theorem \ref{thm:multiscomars1}, given below in Lemma \ref{lemma:multiscomars1}
\begin{lemma} \normalfont
The algorithm in Theorem \ref{thm:multiscomars1} is sequntially controlled for the special case of $G_1 = 1$.
\label{lemma:multiscomars1}
\end{lemma}
\begin{proof}
The first step of the algorithm in Theorem \ref{thm:multiscomars1} runs SCOMARS with treatment group 1 and the supergroup $\{2,3,...,G\}$. Thus, the first step itself determines the locations of the SOM where treatment 1 gets to choose. Since SCOMARS is sequentially controlled, we immediately have $|S_{i1} - F_{i1}|<1$ for all $i \in \{1,2,...,N\}$. 

It remains to show that for $g \in \{2,3,...,G\}$, $|S_{ig} - F_{ig}|<1$ for all $i \in \{1,2,...,N\}$. By symmetry, it suffices to show this for $g =2$. Now, the randomized chunk algorithm on the supergroup $\{2,3,...,G\}$ determines the locations of the SOM where treatment 2 gets to choose. We will prove the result by first mapping this SOM to an SOM where treatment 1 is absent, and then by using the sequential controlled property of randomized chunk. 

Let us first denote $1 \leq r_1<r_2<...<r_{n_1 - 1}< r_{n_1} \leq N$ as the stages or locations of the SOM where treatment 1 gets to choose. We consider the following cases,

\noindent (i) Case-1: $i \in \{1,2,...,r_1 - 1 \}$. In this case, by stage $i$, treatment 1 has not made any choices. Now,
\begin{align}
    |S_{i2} - F_{i2}| & = |S_{i2} - \frac{i n^{(2)}}{N}| \nonumber\\ 
    & \leq |S_{i2} - \frac{i}{G-1}| + |\frac{i}{G-1} - \frac{i n^{(2)}}{N}| \nonumber\\ 
    & \leq \frac{G-2}{G-1} + i \frac{n_1}{N(G-1)} \nonumber\\
    & < \frac{G-2}{G-1} + \frac{1}{G-1} = 1. \label{eqA:lemmaA2_1} 
\end{align}
Here the first inequality holds due to triangle inequality. To see that second inequality, consider a new experiment with treatment groups $\{2,...,G\}$ of size $n^{(2)}$ each and an SOM generated by randomized chunk as in the second step of the algorithm in Theorem \ref{thm:multiscomars1}. Let $\tilde{S}_{i2}$ be the number of selections made by treatment $2$ up to stage $i$ in this new experiment and $\tilde{F}_{i2} = \frac{i}{G-1}$ be its expectation. By Proposition \ref{prop_randchunk}, $|\tilde{S}_{i2} - \tilde{F}_{i2}| \leq \frac{G-2}{G-1}$. Now, $|S_{i1} - \frac{i}{G-1}| = |\tilde{S}_{i1} - \frac{i}{G-1}|$, which gives us the second inequality. Finally, the last inequality holds since $\frac{in_1}{N} = F_{i1} <1$.

\noindent (ii) Case-2: $i \in \{r_t, r_t+1,....,r_{t+1} - 1\}$ for some $t \in \{1,2,...,n_1-1\}$. In this case, by stage $i$, treatment 1 has made exactly $t$ choices. Now,
\begin{align}
    |S_{i2} - F_{i2}| & = |S_{i2} - \frac{i n^{(2)}}{N}| \nonumber\\ 
    & \leq |S_{i2} - \frac{i-t}{G-1}| + |\frac{i-t}{G-1} - \frac{i n^{(2)}}{N}| \nonumber\\ 
    & \leq \frac{G-2}{G-1} + \frac{1}{G-1}|t - \frac{in_1}{N}| \nonumber\\
    & < \frac{G-2}{G-1} + \frac{1}{G-1} = 1. \label{eqA:lemmaA2_2} 
\end{align}
Here, the first inequality is due to triangle inequality. To see the second inequality, we again consider the new experiment
described in Case-1. Notice that, $|S_{i2} - \frac{i-t}{G-1}|  = |\tilde{S}_{(i-t)2} - \tilde{F}_{(i-t)2}| \leq \frac{G-2}{G-1}$, where the last inequality holds by Proposition \ref{prop_randchunk}. Finally, the final inequality in Equation \ref{eqA:lemmaA2_2} holds since $|t - \frac{in_1}{N}| = |S_{i1} - F_{i1}|<1$. This completes the proof of the lemma.
\end{proof}
We now prove Theorem \ref{thm:multiscomars1}.
\begin{proof}
We first show that, for $g \in \{1,2,...,G_1\}$, 
\begin{equation}
    |S_{ig} - F_{ig}|<1 \quad \forall i \in \{1,2,...,N\}.
    \label{eqA:thmA2_1}
\end{equation}
To show this, we consider steps 1 and 2 of the algorithm as these two steps are sufficient to determine the location of treatments $1,...,G_1$ in the SOM. We note that, steps 1 and 2 generate an SOM for an experiment with $G_1+1$ treatment groups, namely supergroup $\{G_1+1,...,G\}$ (of size $(G-G_1)n^{(2)}$) and groups $1,2,...,G_1$ (each of size $n^{(1)}$). Thus, by Lemma \ref{lemma:multiscomars1}, it follows that Equation \ref{eqA:thmA2_1} holds for $g \in \{1,2,...,G_1\}$.  

To show that Equation \ref{eqA:thmA2_1} holds for $g \in \{G_1+1,...,G\}$, we first notice that steps 2 and 3 of the algorithm are completely independent and hence can be performed in any order. Therefore, by changing the order of steps 2 and 3 and applying the same argument as before, we get that Equation \ref{eqA:thmA2_1} holds for $g \in \{G_1+1,...,G\}$. This completes the proof of the theorem.
\end{proof}

\begin{theorem}\normalfont
\label{thm:multiscomars2}
Let $G_1,...,G_m$ be such that $1 \leq G_j \leq G-1$ for all $j \in \{1,2,...,m\}$ and $G_1+ G_2+...+G_m = G$. Moreover, for $j \in \{1,2,...,m\}$, let $n^{(j)}$ be the group size of $G_j$ many treatment groups, with $n^{(1)}G_1 = n^{(2)}G_2 = ...=n^{(m)}G_m $. Denote the collection of $G_j$ treatment groups with group sizes $n^{(j)}$ as supergroup $\mathcal{G}_j$. Consider the following multi-stage algorithm.
\begin{enumerate}
    \item Run randomized chunk on supergroups $\mathcal{G}_1,\mathcal{G}_2,...,\mathcal{G}_m$ to generate an SOM at the supergroup level.
    \item For $j \in \{1,2,...,m\}$, consider the locations of the SOM in step 1 where supergroup $\mathcal{G}_j$ chooses. Then, use randomized chunk to obtain the selection orders at the levels of the original groups in those locations.
\end{enumerate}
The above SOM generating algorithm is sequentially controlled. 
\end{theorem}
To prove this theorem, we first use the following Lemma.
\begin{lemma} \normalfont
Let $n_1 = n_2 = ... = n_G = n$. Consider the following SOM generating algorithm.
\begin{enumerate}
    \item Consider the supergroups $\{1\}$ (of size $n$) and $\{2,3,...,G\}$ (of size $(G-1)n$). Generate an SOM at the superpopulation level using SCOMARS.
    \item Consider the locations of the SOM in step 1 where supergroup $\{2,3,...,G\}$ chooses. Then, use randomized chunk to obtain the selection orders at the levels of the original groups in those locations.
\end{enumerate}
This algorithm is equivalent to the randomized chunk algorithm.
\label{lemma:randchunk}
\end{lemma}
Below we prove this lemma.
\begin{proof}
To show that the algorithm is equivalent to randomized chunk, we have to show that it generates a random permutation of $(1,2,...,G)$ for the first $G$ stages, a fresh random permutation of $(1,2,...,G)$ for the next $G$ stages, and so on. Since the locations of groups $\{2,...,G\}$ are chosen using randomized chunk, it thus suffices to show that, treatment 1 gets to choose once (in a random location) in the first $G$ stages, once in the next $G$ stages, and so on.

We use the notation as in Section \ref{secA_scomars}. Now, suppose among the first $G$ stages, treatment 1 gets to choose at stage $r^*$ first. Notice that $r^*$ cannot be greater than $G$ as
\begin{align}
        P(W_{G}=1|S_{G-1} = 0) = P \Big\{ U \leq \frac{\frac{1}{G} - \max(0,0 - F_{G-1})}{1 - |0 - F_{G-1}|} \Big\} = P\Big\{U \leq \frac{1}{G-(G-1)}\Big\} = 1.
        \label{eqA:lemmaA3_1}
\end{align}
Now, for $r \in \{1,2,...,r^*-1\}$ we have,
\begin{align}
        P(W_{r}=1|S_{r-1} = 0) = P \Big\{ U \leq \frac{\frac{1}{G} - \max(0,0 - F_{r-1})}{1 - |0 - F_{r-1}|} \Big\} = P\Big\{U \leq \frac{1}{G-(r-1)}\Big\} = \frac{1}{G-(r-1)}.
        \label{eqA:lemmaA3_2}
\end{align}
For $r^*+1 \leq r \leq G$,
\begin{align}
        P(W_{r}=1|S_{r-1} = 1) = P \Big\{ U \leq \frac{p_{r} - \max(0,1 - F_{r-1})}{1 - |1 - F_{r-1}|} \Big\} = P\Big\{U \leq \frac{\frac{1}{G} - 1 + \frac{r-1}{G}}{\frac{r-1}{G}}\Big\} = 0.
        \label{eqA:lemmaA3_3}
\end{align}
Finally, 
\begin{align}
    P(W_{G+1}=1|S_{G} = 1) = P \Big\{ U \leq \frac{p_{G+1} - \max(0,1 - F_{G})}{1 - |1 - F_{G}|} \Big\} = P\Big(U \leq \frac{1}{G}\Big) = \frac{1}{G}.
    \label{eqA:lemmaA3_4}
\end{align}
Therefore, by Equation \ref{eqA:lemmaA3_3}, if treatment 1 selects at the $r^*$th stage, it never selects again  $2,3,...,G$. Also,  by Equation \ref{eqA:lemmaA3_2}, before the $r^*$th stage, the conditional probabilities of treatment 1 selecting are same as what it would have been under random permutation of the group labels. Finally, by Equation \ref{eqA:lemmaA3_4}, the process restarts itself at the $(G+1)$th stage, which is equivalent to starting a fresh new random permutation of the group labels. This completes the proof of the lemma.
\end{proof}
We now prove Theorem \ref{thm:multiscomars2}.
\begin{proof}
By the symmetry of the problem, it suffices to show that $|S_{i1} - F_{i1}|<1$ for all $i \in \{1,2,...,N\}$. Without loss of generality, we assume that $\mathcal{G}_1 = \{1,2,...,G_1\}$, which implies that treatment 1 belongs to supergroup $\mathcal{G}_1$. Now, it suffices to focus on the following to steps of the algorithm:
\begin{enumerate}
\item Run randomized chunk on supergroups $\mathcal{G}_1,\mathcal{G}_2,...,\mathcal{G}_m$ to generate an SOM at the supergroup level.
\item Consider the locations of the SOM in step 1 where supergroup $\mathcal{G}_1$ chooses. Then, use randomized chunk to obtain the selection orders at the levels of the original groups in those locations.
\end{enumerate}
This is because, these two steps completely determine the locations of treatment 1 in the SOM. By Lemma \ref{lemma:randchunk}, these two steps can be equivalently performed as follows.
\begin{enumerate}
\item Consider the supergroups $\mathcal{G}_1$ (of size $n^{(1)}G_1$) and $\{\mathcal{G}_2,...,\mathcal{G}_m\}$ (of size $(m-1)n^{(1)}G_1$). Generate an SOM at this supergroup level using SCOMARS.
\item Consider the locations of the SOM in step 1 where supergroup $\{\mathcal{G}_2,...,\mathcal{G}_m\}$ chooses. Then, use randomized chunk to obtain the selection orders at the levels of $\mathcal{G}_j$ in those locations.

\item Consider the locations of the SOM in step 1 where supergroup $\mathcal{G}_1$ chooses. Then, use randomized chunk to obtain the selection orders at the levels of the original groups in those locations.
\end{enumerate}
We note that this above algorithm is exactly equivalent to the SOM generating algorithm in Theorem \ref{thm:multiscomars1} for an experiment with $G_1+m-1$ treatment groups, namely, $1,2,...,G_1,\mathcal{G}_2,\mathcal{G}_3,...,\mathcal{G}_m$. Thus, by Theorem \ref{thm:multiscomars1}, we have $|S_{i1} - F_{i1}|<1$ for all $i \in \{1,2,...,N\}$.
\end{proof}
\subsection{FSM for stratified experiments}
In this section, we discuss two potential approaches to use an FSM for stratified experiments. We consider stratified experiments where the treatment group sizes within each stratum are set by the investigator beforehand. To accommodate the FSM to such experiments, we again need to carefully construct an SOM. In particular, we append the SOM with an additional column of stratum labels, indicating which stratum the treatment group selects from at each stage of the selection process. This column of stratum labels is specified in such a way that the resulting SOM satisfies the group size requirements within each stratum.

Conceptually, the most straightforward approach is to generate a separate SOM for each stratum. This is equivalent to setting the column of stratum labels as $(\underbrace{1,...,1}_{m_1},\underbrace{2,...,2}_{m_2},...,\underbrace{S,....,S}_{m_S})^\top$, where $S$ is the number of strata and $m_s$ is the size of $s$th stratum, $s \in \{1,2,...,S\}$.
This approach is easy to implement and can be useful if, e.g., data on each stratum is available at different stages of the experiment, akin to a sequential experiment. However, in this approach, the treatment groups only get to explore the covariate space of a single stratum for a number of successive stages of selection and hence may not make the most efficient choices. 
We address this issue with an alternative approach. 
For ease of exposition, we consider two strata: 1 and 2. Let $n_{1g}$ and $n_{2g}$ be the (fixed) sizes of treatment group $g \in \{1,2,...,G\}$ in strata 1 and 2, respectively, where $n_{1g} + n_{2g} = n_g$. 
In this approach, we first generate a usual SOM with group sizes $n_1,...,n_G$. For $g \in \{1,2,...,G\}$, we then select the order of the strata that treatment $g$ chooses from by running a SCOMARS algorithm with group sizes $n_{1g}$ and $n_{2g}$. 
By allowing the treatment groups to select units from different strata in a balanced manner, this approach mimics the unstratified FSM where the covariate space of the entire sample is explored for choosing units. 
Also, by design, this approach satisfies the size requirement of each treatment group within each stratum.

\subsection{FSM for sequential experiments}
In this section, we describe our approach to using the FSM for sequential experiments.
Suppose treatment 1 gets to choose at the first stage of selection for the new batch. Let $\underline{\tilde{\bm{X}}}_{\text{old}}$ be the design matrix based on units already assigned to treatment 1. Also, for each unit $i$ in the new batch, let $\underline{\tilde{\bm{X}}}_{\text{new},i} := \begin{psmallmatrix}
\underline{\tilde{\bm{X}}}_{\text{old}}\\
(1, \bm{X}^\top_i) 
\end{psmallmatrix}$ be the resulting design matrix in treatment group 1 if unit $i$ is selected. Treatment 1 selects the unit that maximizes $\det(\underline{\tilde{\bm{X}}}^\top_{\text{new},i}\underline{\tilde{\bm{X}}}_{\text{new},i})$.
In other words, we use the design matrix based on all the units already assigned to the choosing treatment group to evaluate the D-optimal selection function for each unit in the new batch, and select the unit that maximizes the selection function. By carrying over the existing design matrix to the new batch, this approach tends to correct for any existing covariate imbalances.

\subsection{A simulation study}
\label{sec_simulation}
\subsubsection{Setup}
We now compare the performance of the FSM to complete randomization and rerandomization in a simulation study. 
Here, $N=120$, $G=2$, $n_1 = n_2 = 60$, and $k=6$.
The covariates are generated following the design of  \cite{hainmueller2012balancing}:
\begin{equation}
\begin{psmallmatrix}
X_{1}\\
X_{2}\\
X_3
\end{psmallmatrix} \sim
\mathcal{N}_3\left\{\begin{psmallmatrix}
0\\
0\\
0
\end{psmallmatrix},\begin{psmallmatrix}
2 & 1 & -1\\
1 & 1 & -0.5\\
-1 & -0.5 & 1
\end{psmallmatrix}\right\},\hspace{0.1cm}
 X_4 \sim \text{Unif}(-3,3), \hspace{0.1cm} X_5 \sim \chi^2_1, \hspace{0.1cm} X_6 \sim \text{Bernoulli}(0.5). \label{dgp}
\end{equation}
\noindent In this design, $X_4$, $X_5$, and $X_6$ are mutually independent and separately independent of $(X_1,X_2,X_3)^\top$. 
We draw a sample of 120 units once from the data generating mechanism in (\ref{dgp}). 
Conditional on this sample, we compare four different assignment methods, namely a completely randomized design (CRD), rerandomization with 0.01 acceptance rate (RR 0.01), rerandomization with 0.001 acceptance rate (RR 0.001), and the FSM. 
Both RR 0.01 and RR 0.001 use as rerandomization criteria the Mahalanobis distance between the two treatment groups on the original covariates. 
The FSM uses a linear potential outcome model on the original covariates and the D-optimal selection function.
For each design we draw 800 independent assignments. The assignments under the FSM are generated using the open source R package \texttt{FSM} available on CRAN. 
The total runtime of the FSM for the 800 simulated experiments was about one and a half minutes on a Windows 64-bit computer with an Intel(R) Core i7 processor.
See \cite{chattopadhyay2021randomized} for detailed step-by-step instructions and vignettes on the use of FSM package.  

\subsubsection{Balance}

We evaluate balance on the main and transformed covariates.
Figures \ref{fig:simu_asmd}(a) and \ref{fig:simu_asmd}(b) show density plots of the Absolute Standardized Mean Differences (ASMD; \citealt{rosenbaum1985constructing}, \citealt{stuart2010matching}) of the six main covariates and their second-order transformations (including squares and pairwise products), respectively. 
A smaller ASMD for a covariate indicates better mean-balance on that covariate between the two treatment groups. 
Figure \ref{fig:simu_asmd}(a) indicates that both rerandomization methods improve balance on the means of the original covariates over CRD. 
As expected, the ASMD distribution under RR 0.001 is more concentrated than that of RR 0.01, with 32\% smaller mean ASMD than RR 0.01. 
Both the FSM and RR 0.001 have similar distributions of the ASMD with FSM having moderately (9\%) smaller mean ASMD. See Table \ref{tab_app:simu_asmd_org} for a comparison of the average ASMD of each covariate. 

\begin{figure}[!ht]
\begin{subfigure}{.33\textwidth}
  \includegraphics[scale =0.4]{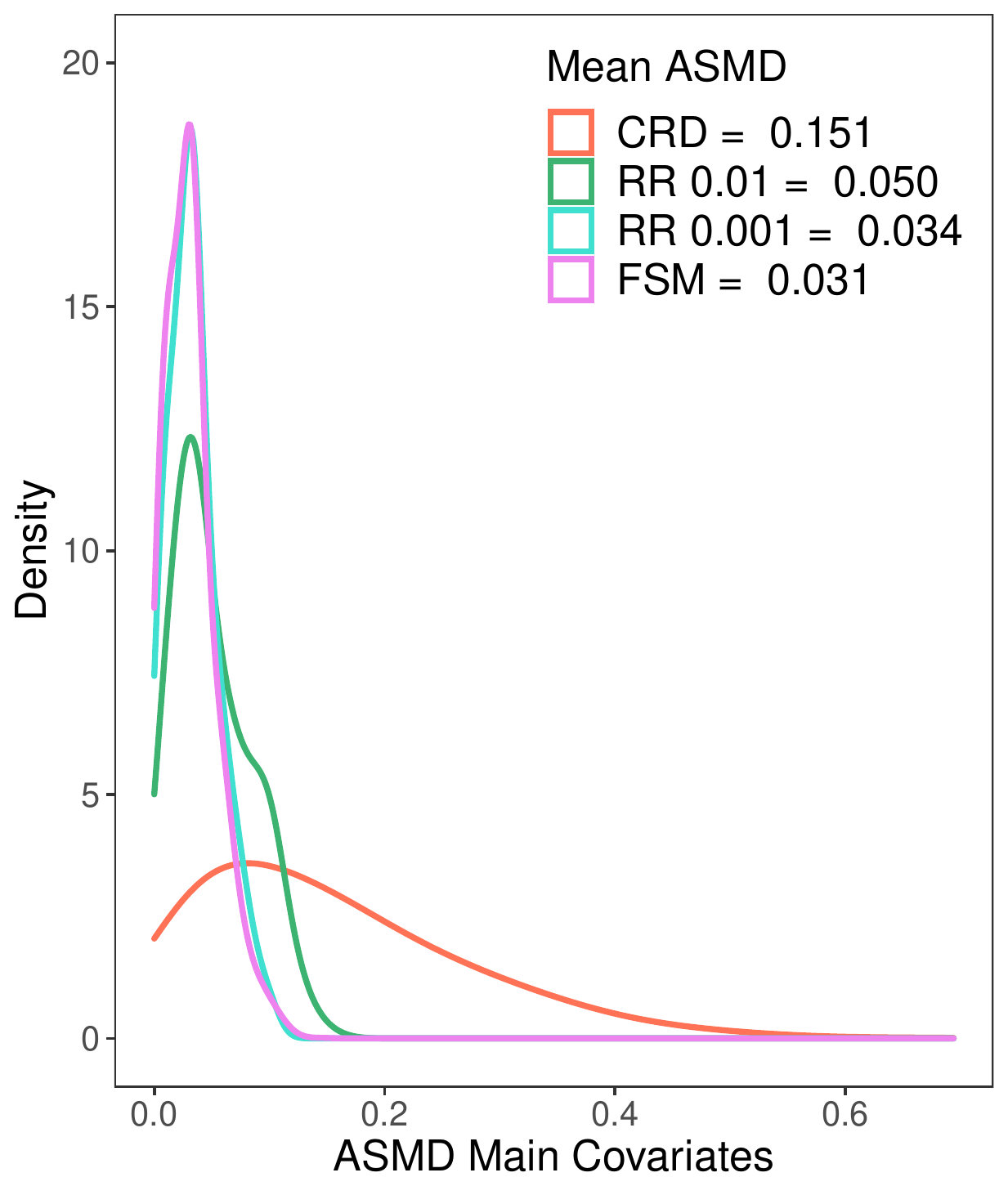}
  \caption{\footnotesize Main covariates}
\end{subfigure}%
\begin{subfigure}{.33\textwidth}
  \includegraphics[scale =0.4]{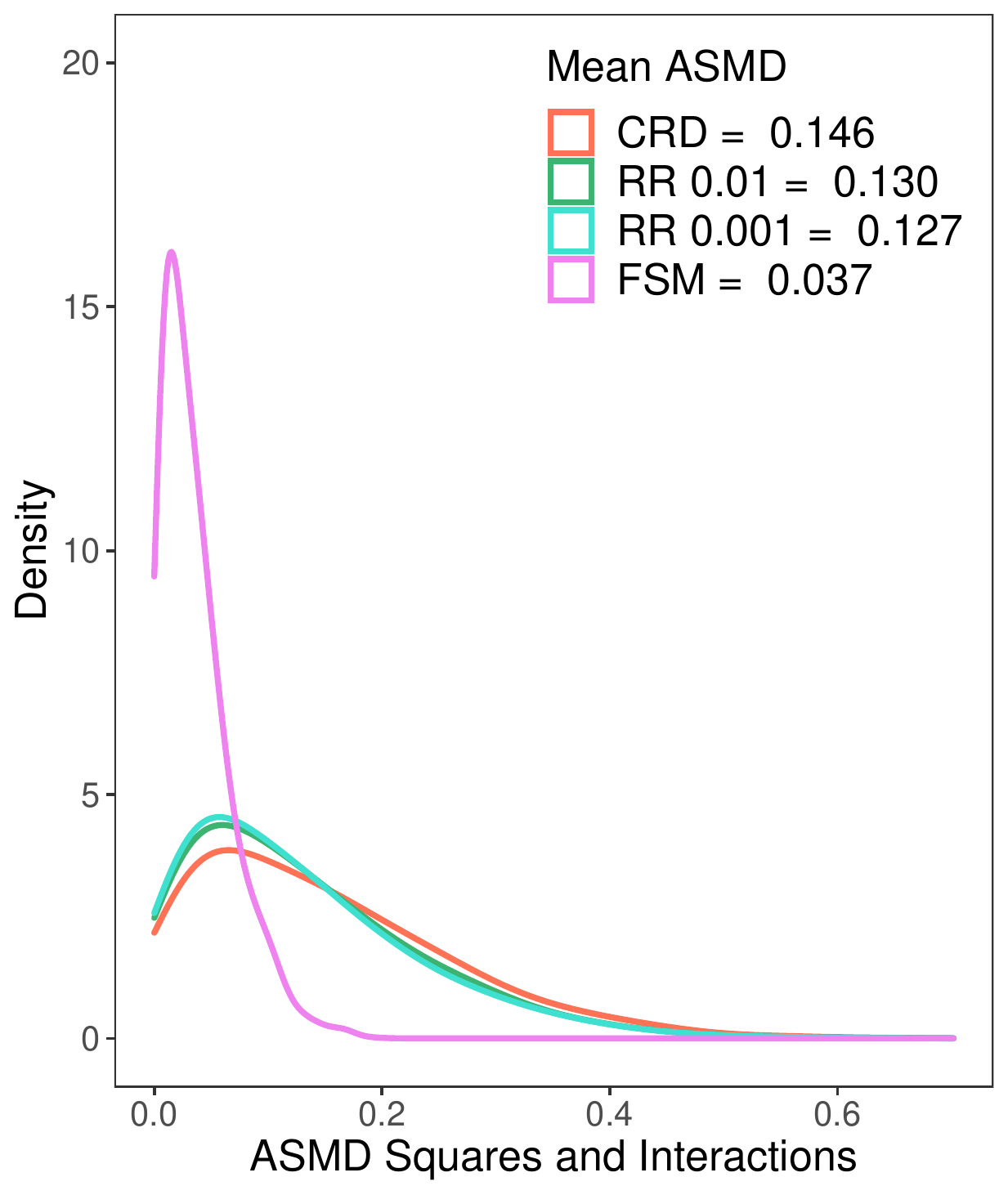}
  \caption{\footnotesize Squares and pairwise products}
\end{subfigure}
\begin{subfigure}{.33\textwidth}
  \includegraphics[scale =0.4]{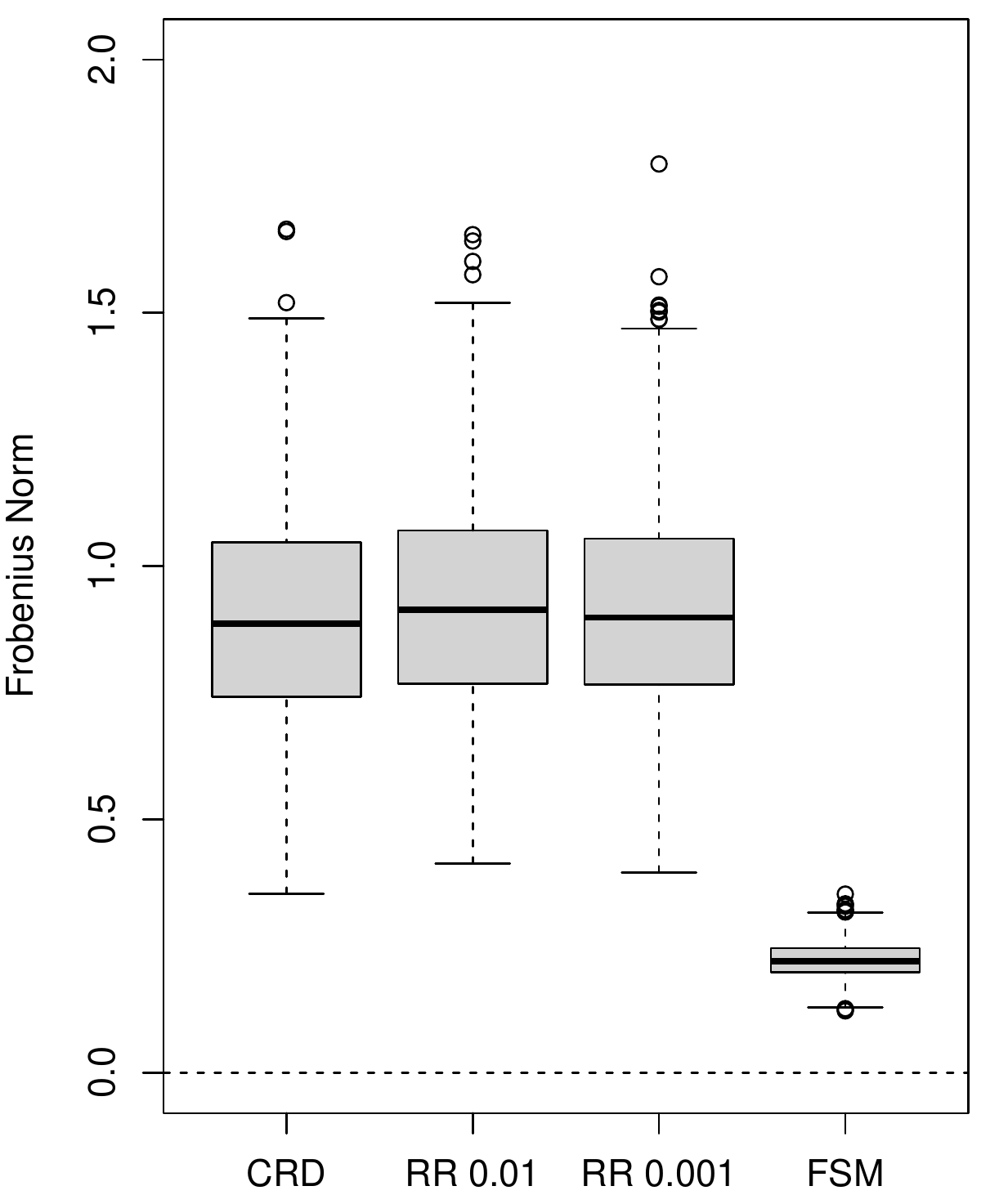}
  \caption{\footnotesize Frobenius norm}
\end{subfigure}
\caption{Panels (a) and (b) show distributions of absolute standardized mean differences (ASMD) of the main covariates and all their second-order transformations across 800 randomizations. For each plot, the legend presents the average ASMD across simulations for the four methods. Panel (c) shows distributions of discrepancies between the correlation matrices of the covariates in the treatment and the control group (as measured by the Frobenius norm, $||\underline{\bm{R}}_1 - \underline{\bm{R}}_2||_F$).
On average the FSM achieves better covariate balance.
In terms of the main covariates, the FSM marginally outperforms RR 0.001. 
In terms of the second-order transformations and correlation matrices, the FSM substantially outperforms RR 0.001.}
\label{fig:simu_asmd}
\end{figure}


Figure \ref{fig:simu_asmd}(b) shows that the imbalances of covariate transformations are substantially smaller with the FSM than with CRD, RR 0.01, and RR 0.001. 
In fact, the FSM achieves a 70\% reduction in the mean ASMD with respect to RR 0.001.
Thus, although the FSM and RR 0.001 exhibit comparable balance in terms of the main covariates, the FSM balances these transformations of the covariates much better than RR 0.001. 
This highlights the improved robustness of the FSM against model misspecification.
Moreover, reducing the tuning parameter of rerandomization from 0.01 to 0.001 yields only 2\% improvement in the mean ASMD.\footnote{In fact, for some covariate transformations, reducing this tuning parameter exacerbates imbalance (see Table \ref{tab_app:simu_asmd_sqint}).}
In Figure \ref{fig:simu_asmd}(b), both RR 0.01 and RR 0.001 often produce ASMD larger than 0.1, and in some cases, larger than 0.5, indicative of substantial imbalances on these covariate transformations.

For each method, we also compare balance in the overall correlation structure of the covariates. Figure \ref{fig:simu_asmd}(c) shows the boxplots of the distributions of $||\underline{\bm{R}}_1 - \underline{\bm{R}}_2||_F$. 
The FSM outperforms the other three designs with at least 75\% smaller average $||\underline{\bm{R}}_1 - \underline{\bm{R}}_2||_F$.
In particular, among the 800 randomizations, the highest value of $||\underline{\bm{R}}_1 - \underline{\bm{R}}_2||_F$ under FSM is smaller than the corresponding lowest value under the other three designs, indicating that in terms of the correlation structure (and hence the interactions) of the covariates, the least balanced realization of the 800 FSMs exhibits better balance than the best balanced realization of the 800 complete randomizations and rerandomizations.

\subsubsection{Efficiency}

We now compare the efficiency of the  methods under both model- and randomization-based approaches to inference. 
Under the model-based approach, we consider a potential outcome model where $\mathbb{E}\{Y_i(g)|\bm{X}_i\}$ is linear in $\bm{X}_i$ (Model A1) and another model where $\mathbb{E}\{Y_i(g)|\bm{X}_i\}$ is linear in $\bm{X}_i$ and all its second-order transformations (Model A2).
For each potential outcome model, we fit the corresponding observed outcome model by OLS and estimate $\text{PATE}_{2,1}$ using the regression imputation method described in Section \ref{sec_inference}. 
Tables \ref{tab:simu_var_model}(a) and \ref{tab:simu_var_model}(b) show the average and maximum  model-based standard error (SE) of the regression imputation estimator relative to the FSM across 800 randomizations under the two models.

\begin{singlespacing}
\begin{table}[H]
   \caption{Average and maximum model-based standard errors relative to the FSM across randomizations. Under Model A1 (linear model on the main covariates), the FSM and RR exhibit similar performance, improving over CRD. Under Model A2 (linear model on the main covariates and their second-order transformations), the FSM is considerably more efficient than both CRD and RR.}
   \begin{subtable}{.5\linewidth}
   \centering
   \caption{\footnotesize Model A1}
   \scalebox{0.75}{
            \begin{tabular}{p{2.5cm}cccc}
    \toprule 
    \multirow{2}{5cm}{} & \multicolumn{4}{c}{Designs}\\
   \cline{2-5}
    & CRD & RR 0.01 & RR 0.001 & FSM\\
    \toprule
Average SE & 1.03 & 1.00 & 1.00 & 1.00 \\
Maximum SE & 1.13 & 1.00 & 1.00 & 1.00 \\
\bottomrule
  \end{tabular}
}
    \end{subtable}%
    \begin{subtable}{.5\linewidth}
    \centering
   \caption{\footnotesize Model A2}
   \scalebox{0.75}{
        \begin{tabular}{p{2.5cm}cccc}
    \toprule 
    \multirow{2}{5cm}{} & \multicolumn{4}{c}{Designs}\\
   \cline{2-5}
    & CRD & RR 0.01 & RR 0.001 & FSM\\
    \toprule
Average SE & 1.39 & 1.27 & 1.26 & 1.00 \\
Maximum SE & 3.61 & 1.97 & 1.80 & 1.00 \\
\bottomrule 
  \end{tabular}
  }
  \end{subtable}
  \label{tab:simu_var_model}
    \end{table}
\end{singlespacing}

Under Model A1, since both rerandomization and the FSM are able to adequately balance the means of the original covariates, 
they lead to lower SE (hence, higher efficiency) than CRD. 
Across randomizations, the worst case SE under RR 0.01, RR 0.001, and the FSM are 13\% smaller than under CRD. 
Under Model A1, the FSM has similar model-based SE as the two rerandomization methods. 
However, under Model A2, the FSM uniformly outperforms the other three designs, with a 26\% reduction in average SE and an 80\% reduction in maximum SE than RR 0.001. 
This improvement in efficiency can be attributed to the balance achieved by the FSM on the main covariates and their squares and pairwise products. 
In sum, when the model assumed at the design stage is correct and is used at the analysis stage, the FSM is as efficient as the two rerandomizations for estimating the treatment effect. 
However, when the model assumed at the design stage is misspecified and later corrected by augmenting transformations of the covariates (e.g., squares and pairwise products), the FSM is considerably more efficient and robust than the other designs. 

Under the randomization-based approach, we compare the standard errors of the difference-in-means statistic under each design. Following \cite{hainmueller2012balancing}, the potential outcomes are generated using the models: $Y(1) = X_1 + X_2 + X_3 - X_4 +X_5 + X_6 + \eta$,  $Y(2) = Y(1)$ (Model B1) and $Y(1) = (X_1 + X_2 + X_5)^2 + \eta$, $Y(2) = Y(1)$ (Model B2), where $\eta \sim \mathcal{N}(0,1)$. 
Both generative models satisfy the sharp-null hypothesis of zero treatment effect for every unit and hence, $\text{SATE}_{2,1} = 0$. Conditional on these potential outcomes, $\text{SATE}_{2,1}$ is estimated under each design using the standard difference-in-means estimator. The corresponding randomization-based SE of this estimator is obtained by generating 800 randomizations of the design and computing the standard deviation of the difference-in-means estimator across these 800 randomizations.   
Table \ref{tab:simu_var_rand} shows the randomization-based SE of the difference-in-means statistic for $\text{SATE}_{2,1}$ under each model. 

\begin{singlespacing}
\begin{table}[H]
 \caption{Randomization-based standard errors relative to the FSM. The standard error for the FSM is 0.2 under Model B1 (linear model on the main covariates) and 0.43 under Model B2 (linear model on the main covariates and their second-order transformations). Especially under Model B2, the FSM is considerably more efficient than both CRD and RR.}
   \begin{subtable}{.5\linewidth}
   \centering
   \caption{\footnotesize Model B1}
   \scalebox{0.75}{
            \begin{tabular}{p{1.5cm}cccc}
    \toprule 
    \multirow{2}{5cm}{} & \multicolumn{4}{c}{Designs}\\
   \cline{2-5}
    & CRD & RR 0.01 & RR 0.001 & FSM\\
    \toprule
SE & 2.72 & 1.26 & 1.08 & 1 \\
\bottomrule
  \end{tabular}
}
    \end{subtable}%
    \begin{subtable}{.5\linewidth}
    \centering
   \caption{\footnotesize Model B2}
   \scalebox{0.75}{
        \begin{tabular}{p{1.5cm}cccc}
    \toprule 
    \multirow{2}{5cm}{} & \multicolumn{4}{c}{Designs}\\
   \cline{2-5}
    & CRD & RR 0.01 & RR 0.001 & FSM\\
    \toprule
SE & 5.69 & 4.56 & 4.47 & 1 \\
\bottomrule 
  \end{tabular}
  }
  \end{subtable}
\label{tab:simu_var_rand}
    \end{table}
\end{singlespacing}

Under Model B1, the potential outcomes depend linearly on the covariates and therefore balancing the means of the covariates improves efficiency. 
This is reflected in Table \ref{tab:simu_var_rand} as the FSM has the smallest SE, closely followed by RR 0.001. 
Under Model B2, the potential outcomes depend linearly on the squares and pairwise products of the covariates. 
By better balancing these transformations, the FSM yields a considerably smaller SE than the other designs. 
In particular, under Model B2, the SE under the FSM is 67\% smaller than the SE under RR 0.001.
Therefore, as in the model-based approach, in the randomization-based approach the FSM exhibits comparable efficiency to rerandomization under correct-specification of the outcome model and considerable robustness under model misspecification.

\subsubsection{Additional tables and figures from the simulation study}

\begin{singlespacing}
\begin{table}[H]
   \caption{\footnotesize Averages of the ASMD of the original covariates across 800 randomizations.}
\scalebox{0.85}{
     \centering
            \begin{tabular}{p{3cm}cccc}
    \toprule 
    \multirow{2}{5cm}{Covariates} & \multicolumn{4}{c}{Designs}\\
   \cline{2-5}
    & CRD & RR 0.01 & RR 0.001 & FSM\\
    \toprule
$X_1$ & 0.162 & 0.051 & 0.035 & 0.029 \\ 
$X_2$ & 0.156 & 0.048 & 0.033 & 0.025 \\ 
$X_3$ & 0.158 & 0.049 & 0.033 & 0.042 \\ 
$X_4$ & 0.150 & 0.049 & 0.034 & 0.029 \\ 
$X_5$ & 0.140 & 0.052 & 0.034 & 0.029 \\ 
$X_6$ & 0.141 & 0.052 & 0.036 & 0.035 \\ 
\hline 
Mean & 0.151 & 0.050 & 0.034 & 0.032\\ 
\bottomrule
  \end{tabular}
}
\label{tab_app:simu_asmd_org}
\end{table}
\end{singlespacing}

\begin{singlespacing}
\begin{table}[H]
\scalebox{0.85}{
         \begin{tabular}{p{3cm}cccc}
    \toprule 
    \multirow{2}{5cm}{Covariate \\ transformations} & \multicolumn{4}{c}{Designs}\\
   \cline{2-5}
    & CRD & RR 0.01 & RR 0.001 & FSM\\
    \toprule
$X_1 X_2$ & 0.144 & 0.153 & 0.148 & 0.041 \\ 
$X_1 X_3$ & 0.144 & 0.140 & 0.137 & 0.041 \\ 
$X_1 X_4$ & 0.141 & 0.148 & 0.147 & 0.023 \\ 
$X_1 X_5$ & 0.150 & 0.135 & 0.134 & 0.035 \\ 
$X_1 X_6$ & 0.152 & 0.109 & 0.101 & 0.030 \\ 
$X_2 X_3$ & 0.147 & 0.147 & 0.146 & 0.051 \\ 
$X_2 X_4$ & 0.140 & 0.155 & 0.150 & 0.027 \\ 
$X_2 X_5$ & 0.147 & 0.143 & 0.136 & 0.030 \\ 
$X_2 X_6$ & 0.152 & 0.115 & 0.104 & 0.026 \\ 
$X_3 X_4$ & 0.141 & 0.143 & 0.152 & 0.032 \\ 
$X_3 X_5$ & 0.149 & 0.140 & 0.139 & 0.096 \\ 
$X_3 X_6$ & 0.148 & 0.099 & 0.091 & 0.035 \\ 
$X_4 X_5$ & 0.148 & 0.132 & 0.130 & 0.037 \\ 
$X_4 X_6$ & 0.152 & 0.100 & 0.095 & 0.027 \\ 
$X_5 X_6$ & 0.146 & 0.095 & 0.094 & 0.024 \\ 
$X_1^2$ & 0.140 & 0.145 & 0.143 & 0.031 \\ 
$X_2^2$ & 0.151 & 0.155 & 0.150 & 0.038 \\ 
$X_3^2$ & 0.144 & 0.136 & 0.132 & 0.041 \\ 
$X_4^2$ & 0.143 & 0.145 & 0.147 & 0.053 \\ 
$X_5^2$ & 0.142 & 0.073 & 0.067 & 0.013 \\
\hline
Mean & 0.146 & 0.130 & 0.127 & 0.037 \\
\hline
$X_5^{1.5}$ & 0.141 & 0.060 & 0.048 & 0.018 \\ 
$X_2^3$ & 0.155 & 0.090 & 0.081 & 0.071 \\ 
$X_4^4$ & 0.140 & 0.143 & 0.147 & 0.072 \\ 
$\frac{1}{4+X_3}$ & 0.157 & 0.073 & 0.064 & 0.050 \\ 
\hline
Mean & 0.148 & 0.092 & 0.085 & 0.053 \\
\bottomrule 
  \end{tabular}
}
\caption{\footnotesize Averages of the ASMD of squares, pairwise products, and other transformations of the covariates across 800 randomizations.} 
\label{tab_app:simu_asmd_sqint}      
\end{table}
\end{singlespacing}


\begin{figure}[H]
\centering
\includegraphics[scale =0.5]{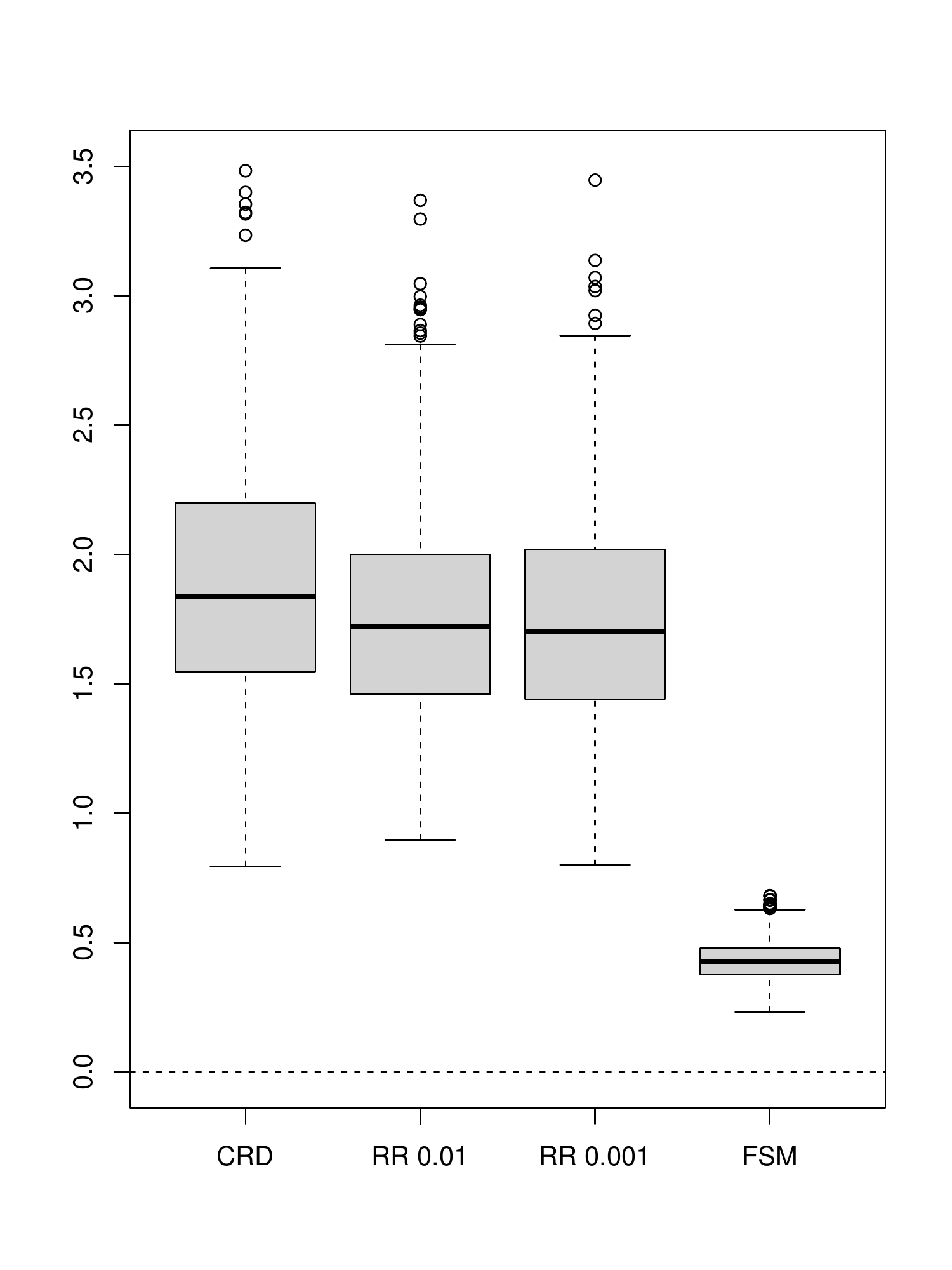}
\caption{\footnotesize Boxplot of the distribution of $||\underline{\bm{S}}_1 - \underline{\bm{S}}_2||_F$ across 800 randomizations, where $\underline{\bm{S}}_g$ is the sample covariance matrix of the covariates in treatment group $g \in \{1,2\}$.}
\label{fig:simu_frob}
\end{figure}

\subsubsection{Additional tables and figures from the Health Insurance Experiment}
\label{appsec_hie}

\begin{singlespacing}
\begin{table}[H]
\centering
\scalebox{0.85}{
\begin{tabular}{p{6cm}cccc}
\toprule 
\multirow{2}{5cm}{Covariates} & \multicolumn{4}{c}{Designs}\\
 \cline{2-5}
& CRD & RR Wilks & RR Mahalanobis & FSM\\
\toprule
$X_{1}:$ Log family size & 0.052 & 0.039 & 0.038 & 0.012 \\ 
$X_{2}:$ Log family income & 0.052 & 0.040 & 0.043 & 0.010 \\ 
$X_{3}:$ Max hourly wage & 0.051 & 0.042 & 0.047 & 0.017 \\ 
$X_{4}:$ Adult med visits & 0.049 & 0.043 & 0.041 & 0.014 \\ 
$X_{5}:$ Kid med visits & 0.048 & 0.039 & 0.040 & 0.010 \\
  \hline
$X_{6}:$ Female & 0.047 & 0.039 & 0.040 & 0.010 \\ 
 $X_{7}:$ Age 0 to 5 & 0.053 & 0.038 & 0.039 & 0.010 \\ 
 $X_{8}:$ Age 6 to 17 & 0.051 & 0.041 & 0.039 & 0.011 \\ 
 $X_{9}:$ Age 18 to 44 & 0.053 & 0.038 & 0.040 & 0.010 \\ 
 $X_{10}:$ Male HS Grad & 0.051 & 0.038 & 0.041 & 0.006 \\ 
 $X_{11}:$ Male more than HS & 0.048 & 0.037 & 0.041 & 0.006 \\ 
 $X_{12}:$ Insured & 0.049 & 0.040 & 0.038 & 0.010 \\ 
 $X_{13}:$ Excellent health & 0.052 & 0.040 & 0.037 & 0.009 \\ 
 $X_{14}:$ Good health & 0.053 & 0.038 & 0.037 & 0.010 \\ 
  \hline
 $X_{15}:$ Family income mis & 0.052 & 0.038 & 0.041 & 0.011 \\ 
 $X_{16}:$ Max hourly wage mis & 0.051 & 0.038 & 0.041 & 0.013 \\ 
 $X_{17}:$ Adult med visits mis & 0.054 & 0.040 & 0.040 & 0.011 \\ 
 $X_{18}:$ Kid med visits mis & 0.057 & 0.041 & 0.039 & 0.011 \\ 
 $X_{19}:$ Education male mis & 0.048 & 0.038 & 0.041 & 0.008 \\ 
 $X_{20}:$ Insured mis & 0.048 & 0.039 & 0.038 & 0.011 \\ \hline
Mean & 0.051 & 0.039 & 0.040 & 0.011\\
   \bottomrule
\end{tabular}
}
\caption{\footnotesize Average ASMD of the main covariates between treatment groups $1$ and $2$ across 400 randomizations.}
\label{tab:hie1}
\end{table}
\end{singlespacing}

\begin{singlespacing}
\begin{table}[H]
\centering
\scalebox{0.85}{
\begin{tabular}{p{4cm}cccc}
\toprule 
\multirow{2}{5cm}{Treatment group} & \multicolumn{4}{c}{Designs}\\
 \cline{2-5}
& CRD & RR Wilks & RR Mahalanobis & FSM\\
\toprule
$1,2$ & 0.051 & 0.039 & 0.040 & 0.011 \\ 
$1,3$ & 0.055 & 0.041 & 0.043 & 0.011 \\ 
$1,4$ & 0.049 & 0.038 & 0.039 & 0.010 \\ 
$2,3$ & 0.056 & 0.043 & 0.045 & 0.012 \\
$2,4$ & 0.053 & 0.040 & 0.041 & 0.010 \\
$3,4$ & 0.056 & 0.042 & 0.044 & 0.012 \\
  \hline
  Mean & 0.053 & 0.040 & 0.042 & 0.011 \\
   \bottomrule
\end{tabular}
}
\caption{\footnotesize Averages of the ASMD between each pair of treatment groups across the original covariates and across 400 randomizations.}
\label{tab:hie_avg_org}
\end{table}
\end{singlespacing}

\begin{singlespacing}
\begin{table}[H]
\centering
\scalebox{0.85}{
\begin{tabular}{p{4cm}cccc}
\toprule 
\multirow{2}{5cm}{Covariates} & \multicolumn{4}{c}{Designs}\\
 \cline{2-5}
& CRD & RR Wilks & RR Mahalanobis & FSM\\
\toprule
$X_1X_2$ & 0.053 & 0.039 & 0.041 & 0.020 \\ 
$X_1X_3$ & 0.053 & 0.047 & 0.046 & 0.027 \\ 
$X_1X_4$ & 0.054 & 0.045 & 0.045 & 0.020 \\ 
$X_1X_5$ & 0.049 & 0.040 & 0.041 & 0.013 \\ 
$X_2X_3$ & 0.054 & 0.049 & 0.053 & 0.038 \\ 
$X_2X_4$ & 0.050 & 0.045 & 0.048 & 0.017 \\ 
$X_2X_5$ & 0.052 & 0.039 & 0.039 & 0.015 \\ 
$X_3X_4$ & 0.054 & 0.043 & 0.045 & 0.022 \\ 
$X_3X_5$ & 0.050 & 0.042 & 0.046 & 0.022 \\ 
$X_4X_5$ & 0.054 & 0.044 & 0.045 & 0.015 \\ 
$X^2_1$ & 0.053 & 0.041 & 0.040 & 0.026 \\ 
$X^2_2$ & 0.051 & 0.041 & 0.042 & 0.015 \\ 
$X^2_3$ & 0.057 & 0.055 & 0.058 & 0.026 \\ 
$X^2_4$ & 0.053 & 0.053 & 0.053 & 0.012 \\ 
$X^2_5$ & 0.051 & 0.043 & 0.044 & 0.004 \\ 
  \hline
  Mean  & 0.053 & 0.044 & 0.046 & 0.019\\
\bottomrule
\end{tabular}
}
\caption{\footnotesize Averages of the ASMD of the squares and pairwise products of the (demeaned) covariates $X_1$,..., $X_5$ between treatment groups $1$ and $2$ across 400 randomizations.}
\label{tab:hie2}
\end{table}
\end{singlespacing}

\begin{singlespacing}
\begin{table}[H]
\centering
\scalebox{0.85}{
\begin{tabular}{p{4cm}cccc}
\toprule 
\multirow{2}{5cm}{Treatment group} & \multicolumn{4}{c}{Designs}\\
 \cline{2-5}
& CRD & RR 0.01 & RR 0.001 & FSM\\
\toprule
$1,2$ & 0.053 & 0.044 & 0.046 & 0.019 \\ 
$1,3$ & 0.056 & 0.046 & 0.048 & 0.020 \\ 
$1,4$ & 0.051 & 0.043 & 0.044 & 0.017 \\ 
$2,3$ & 0.058 & 0.049 & 0.049 & 0.021 \\
$2,4$ & 0.054 & 0.046 & 0.046 & 0.018 \\
$3,4$ & 0.058 & 0.048 & 0.049 & 0.023 \\
  \hline
  Mean & 0.055 & 0.046 & 0.047 & 0.020 \\
   \bottomrule
\end{tabular}
}
\caption{\footnotesize Averages of the ASMD between each pair of treatment groups across the squares and pairwise products of the (demeaned) covariates $X_1$,..., $X_5$ and across 400 randomizations.}
\label{tab:hie_avg_sqint}
\end{table}
\end{singlespacing}

\begin{figure}[H]
    \centering
    \includegraphics[scale = 0.7]{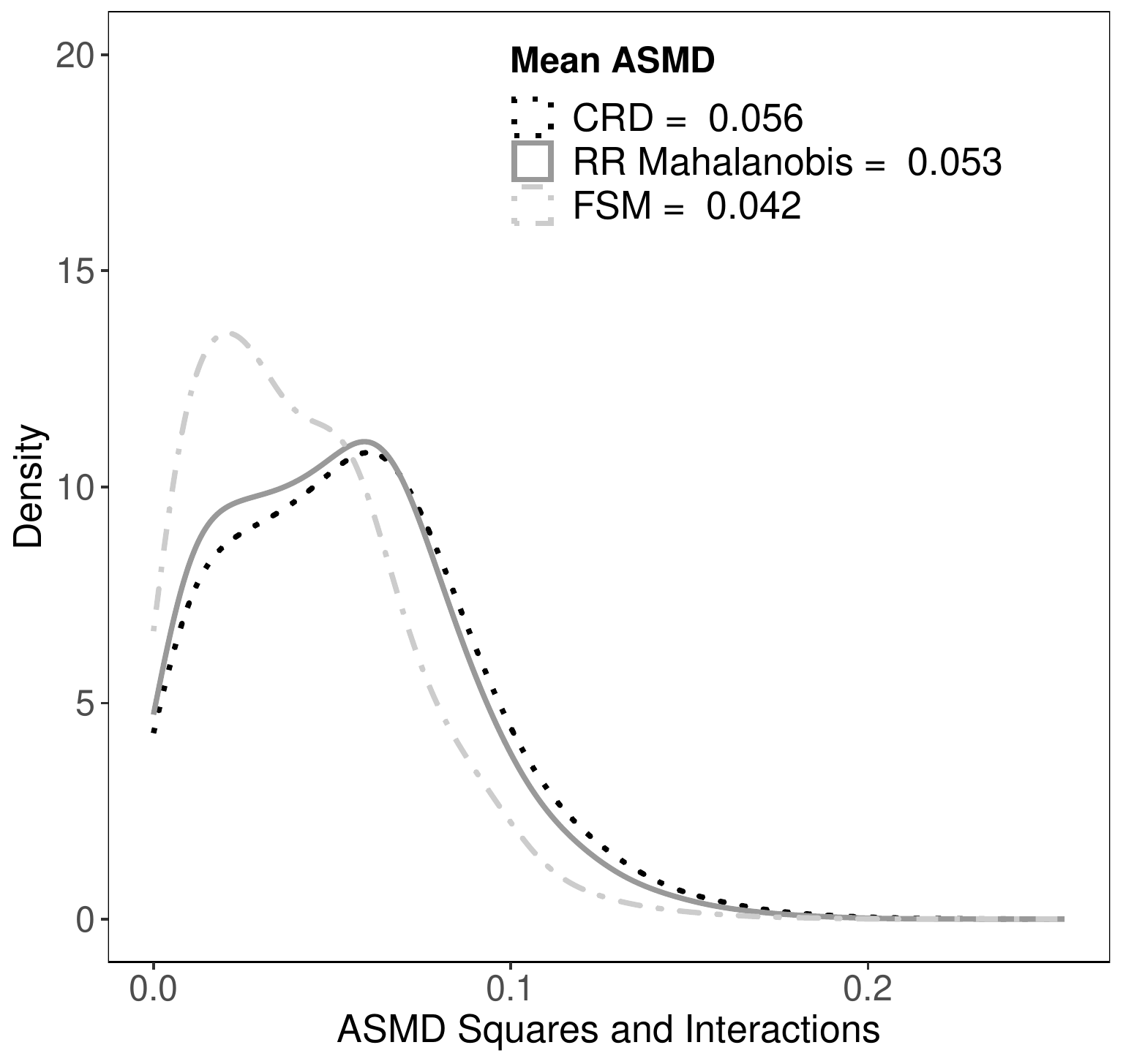}
    \caption{Distributions of ASMD of all cubes and three-way interactions of the non-binary covariates across randomizations.}
    \label{fig:hie_threeway}
\end{figure}

\begin{figure}[!ht]
\begin{subfigure}{.52\textwidth}
  \includegraphics[scale =0.7]{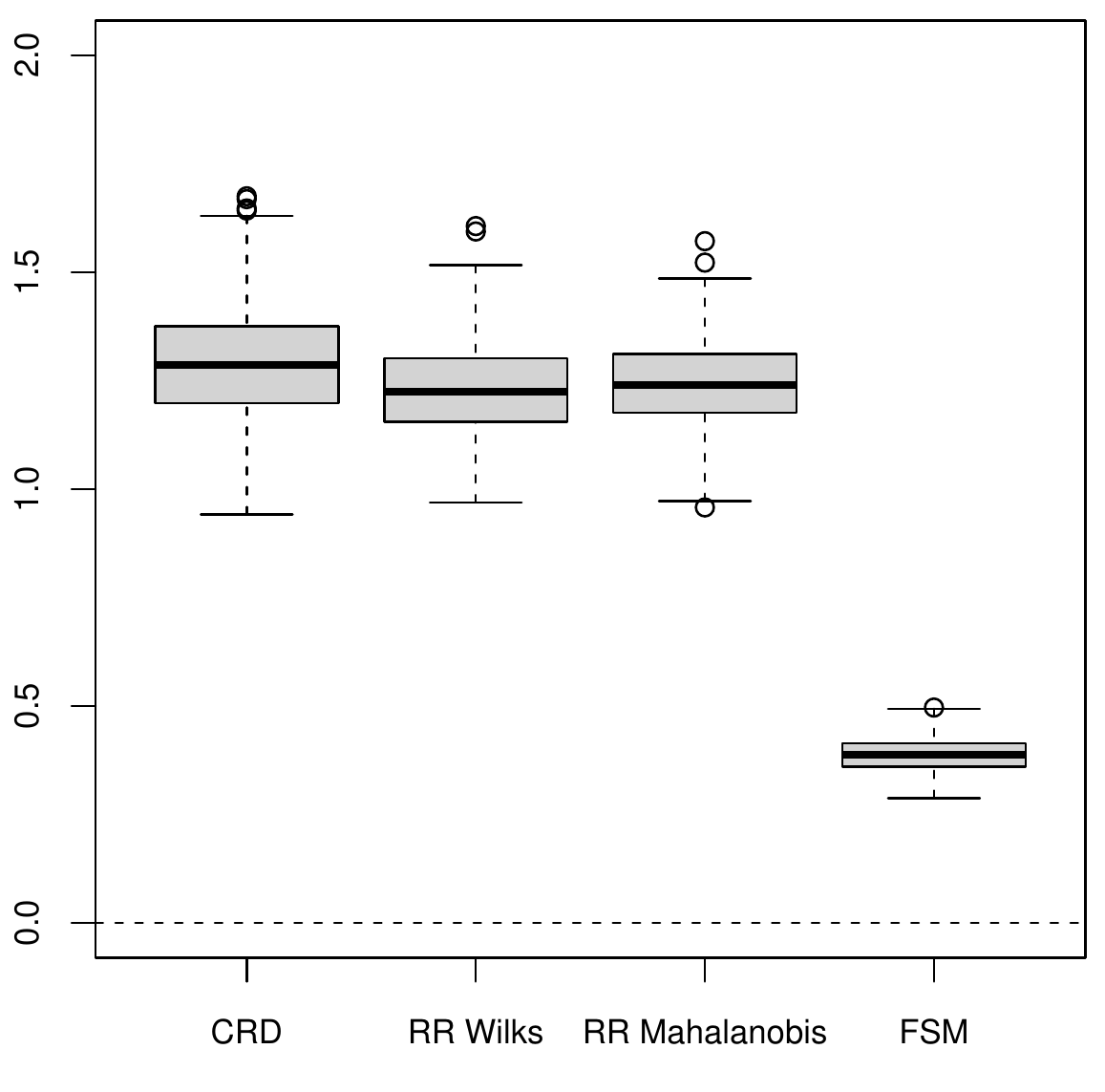}
  \caption{\footnotesize $||\underline{\bm{R}}_1 - \underline{\bm{R}}_2||_F$}
  \label{fig:simu_asmd_org}
\end{subfigure}%
\begin{subfigure}{.52\textwidth}
  \includegraphics[scale =0.7]{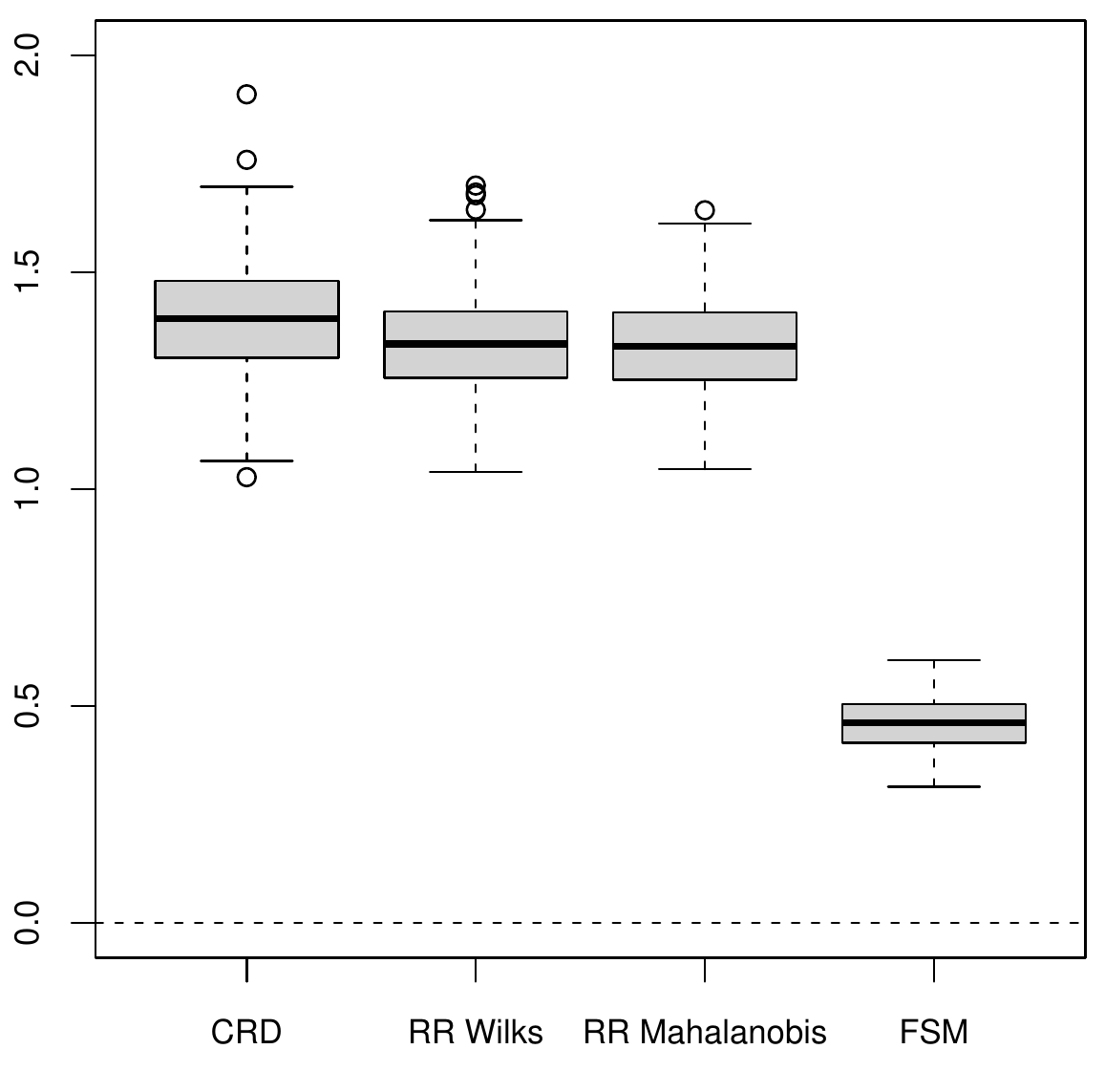}
  \caption{\footnotesize $||\underline{\bm{R}}_3 - \underline{\bm{R}}_4||_F$}
  \label{fig:simu_asmd_sqint}
\end{subfigure}
\caption{\footnotesize Distributions of discrepancies of the correlation matrices of the covariates in the treatment groups of the HIE data across randomizations.
The discrepancies are measured by $||\underline{\bm{R}}_{g} - \underline{\bm{R}}_{g'}||_F$, where $\underline{\bm{R}}_g$ is the sample correlation matrix of the covariates in treatment group $g$ and $||\cdot||_F$ is the Frobenius norm.
The FSM systematically produces lower discrepancies than the other methods, exhibiting substantially improved balance on the correlations of the covariates.
}
\label{figfrob}
\end{figure}

\begin{figure}[H]
\begin{subfigure}{.52\textwidth}
  \includegraphics[scale =0.7]{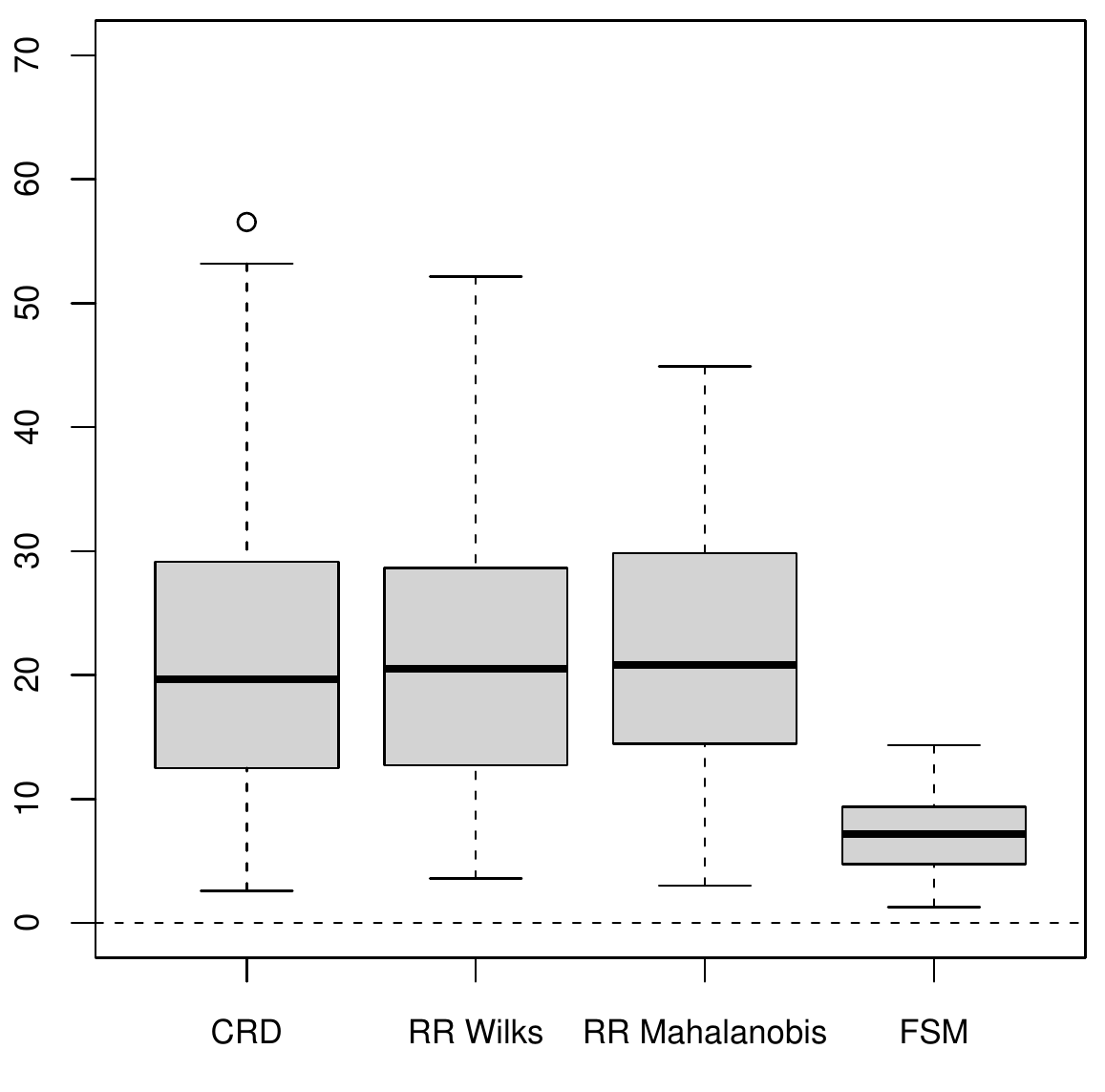}
  \caption{\footnotesize $||\underline{\bm{S}}_1 - \underline{\bm{S}}_2||_F$}
  \label{fig:simu_asmd_org}
\end{subfigure}%
\begin{subfigure}{.52\textwidth}
  \includegraphics[scale =0.7 ]{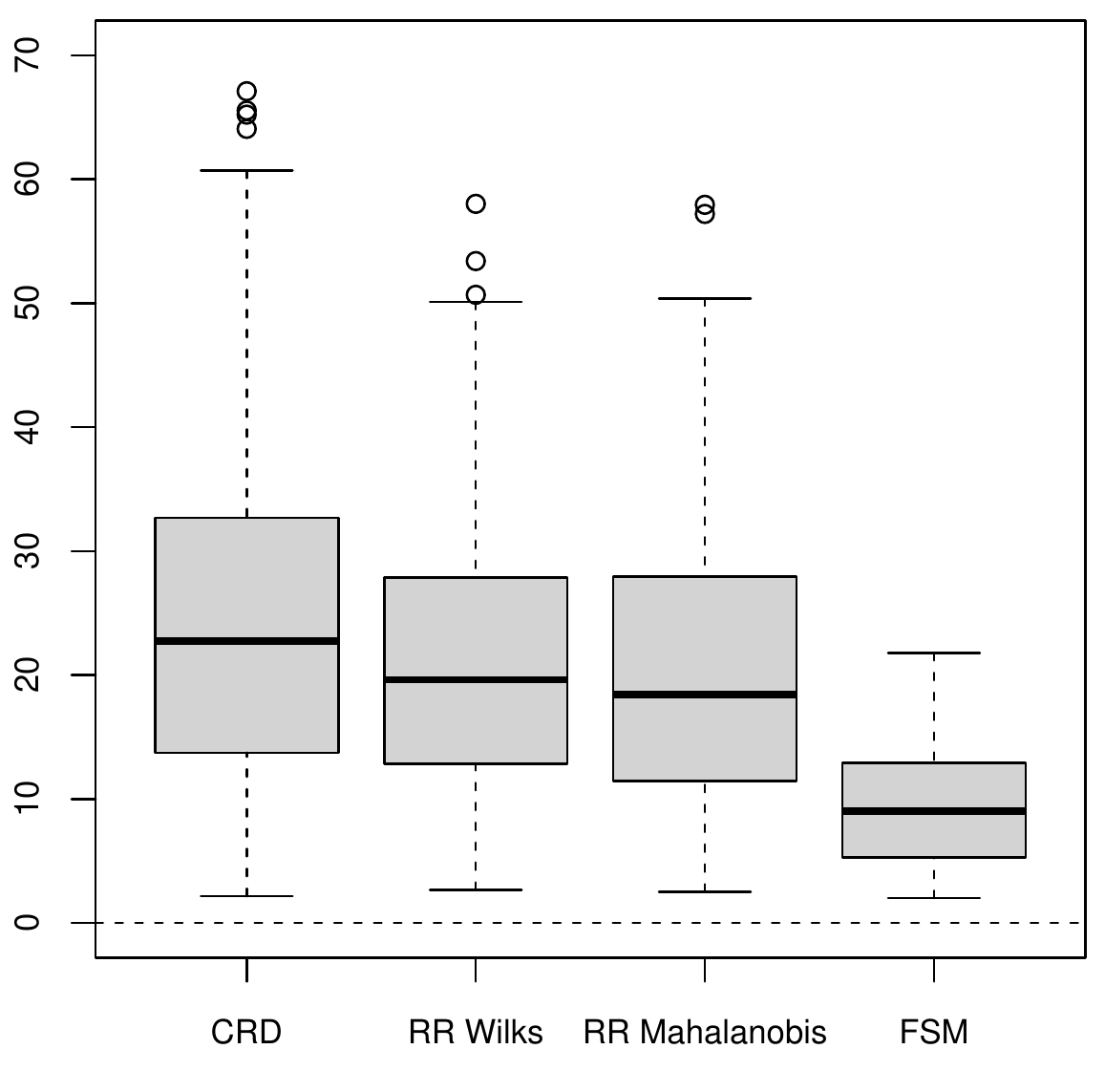}
  \caption{\footnotesize $||\underline{\bm{S}}_3 - \underline{\bm{S}}_4||_F$}
  \label{fig:simu_asmd_sqint}
\end{subfigure}
\caption{\footnotesize Boxplot of the distribution of $||\underline{\bm{S}}_{g} - \underline{\bm{S}}_{g'}||_F$ across 400 randomizations, where $\underline{\bm{S}}_g$ is the sample covariance matrix of the covariates in treatment group $g \in \{1,2,3,4\}$.}
\label{fig:test}
\end{figure}

We now compare the efficiency of the designs in the randomization-based approach with four additional potential outcome models given below.
\begin{itemize}
    \item Model B3: $Y(3) = 5 - 3X_1 + X_2 + X_3 - 0.2X_4 + 0.8X_5 + \eta$, $Y(3) = Y(2)$.
    \item Model B4: $Y(3) = 5 - 2X^2_1 + 0.5X^2_3 + 0.5X^2_5 + 5X_1X_2 - 0.8X_3X_5 + \eta$, $Y(3) = Y(2)$.
    \item Model B5: $Y(3) = 10 + 8X_1X_2 + 3X_2X_5 - 0.5X_3X_5 + \eta$, $Y(3) = Y(2)$.
    \item Model B6: $Y(3) = 0.8X_1X_2 - 3X^2_3 + \frac{1}{1+X_4} - 4X^3_1 + \eta$
\end{itemize}
For each model, the error term $\eta \sim \mathcal{N}(0,1.5^2)$.
Under each design, $\text{SATE}_{3,2}$ is estimated using the standard difference-in-means estimator and the corresponding randomization-based SE is obtained by generating 800 randomizations and computing the standard deviation of the estimator across these 800 randomizations.  
The results are summarized in Table \ref{tab:hie_var_rand2}.

\begin{singlespacing}
\begin{table}[H]
   \caption{Randomization-based standard errors relative to the FSM under Models B3, B4, B5, B6}
 \begin{tabular}{p{3cm}cccc}
    \toprule 
    \multirow{2}{5cm}{} & \multicolumn{4}{c}{Designs}\\
   \cline{2-5}
    & CRD & RR Wilks & RR Mahalanobis & FSM\\
    \toprule
Model B3 & 2.12 & 1.56 & 1.71 & 1 \\
Model B4 & 2.14 & 1.75 & 1.81 & 1 \\
Model B5 & 2.92 & 2.33 & 2.36 & 1 \\
Model B6 & 1.51 & 1.37 & 1.33 & 1 \\
\bottomrule
  \end{tabular}
\label{tab:hie_var_rand2}
\end{table}
\end{singlespacing}

We finish this section by evaluating and comparing the covariate balance on the main covariates and the second order transformations, for CRD, RR, and the FSM, where RR uses the Mahalanobis distance based on the main covariates only and accepts $0.1\%$ of the assignments.

\begin{figure}[H]
\begin{subfigure}{.33\textwidth}
  \includegraphics[scale =0.43]{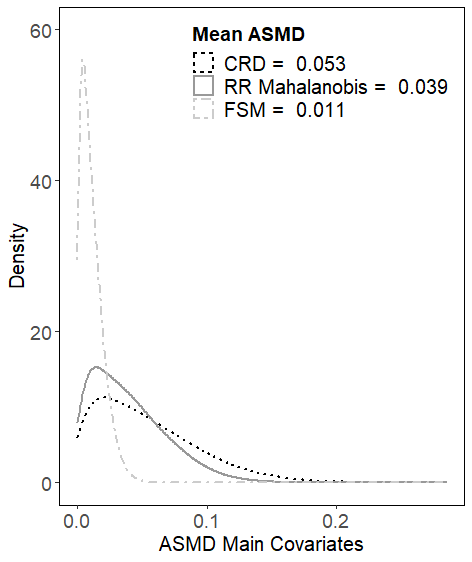}
  \caption{\footnotesize Main covariates}
\end{subfigure}%
\begin{subfigure}{.33\textwidth}
  \includegraphics[scale =0.43]{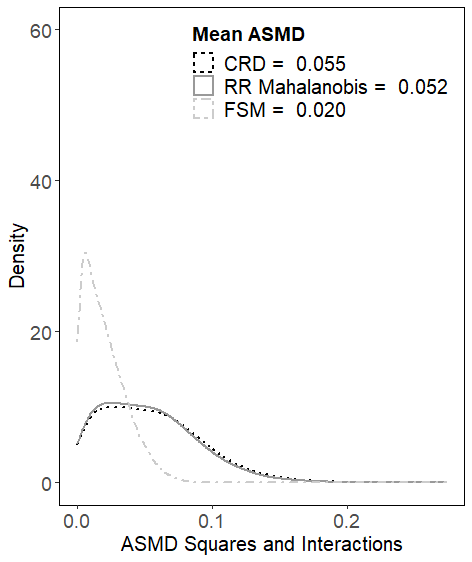}
  \caption{\footnotesize Squares and pairwise products}
\end{subfigure}
\begin{subfigure}{.33\textwidth}
  \includegraphics[scale =0.43]{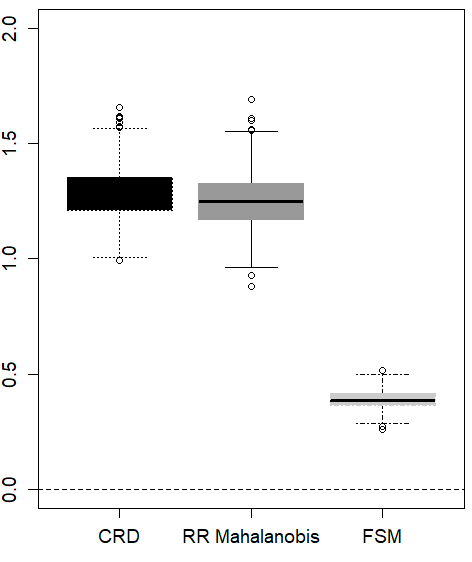}
  \caption{\footnotesize Frobenius norm}
\end{subfigure}
\caption{Distributions of absolute standardized mean differences (ASMD) of the main covariates (panel (a)) and the squares and pairwise products of the scaled covariates (panel (b)) across randomizations. For each plot, the legend presents the average ASMD across simulations for the four methods. Panel (c) shows distributions of discrepancies between the correlation matrices of the covariates in treatment groups 1 and 2 (as measured by the Frobenius norm, $||\underline{\bm{R}}_1 - \underline{\bm{R}}_2||_F$).
In terms of the main covariates, second-order transformations, and correlation matrices, the FSM substantially outperforms CRD and RR.}
\label{fig:test1_old}
\end{figure}

\subsubsection{Additional figures from the case studies}
\label{sec_additional2}

\begin{figure}
\centering
\includegraphics[scale =0.85]{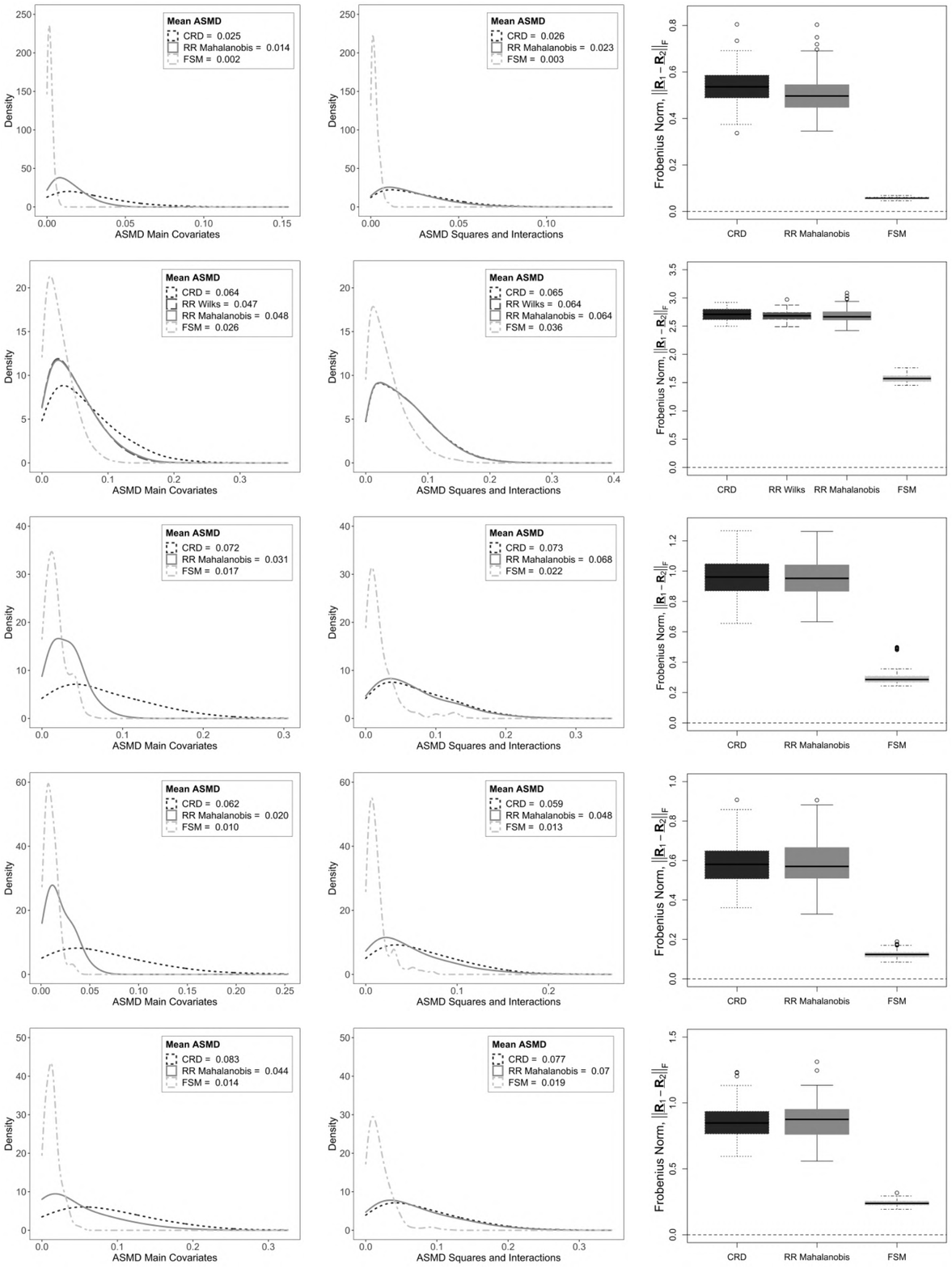}
\caption{Distributions of the absolute standardized mean differences of the main covariates and their squares and interactions, and the Frobenius norms of $\underline{\bm{R}}_1 - \underline{\bm{R}}_2$ under complete randomization, rerandomization, and the FSM, for the five studies: (1) Angrist, (2) Blattman, (3) Durocher, (4) Finkelstein, (5) Lalonde.}
\label{fig:five_studies}
\end{figure}

\begin{figure}
\centering
\includegraphics[scale =0.85]{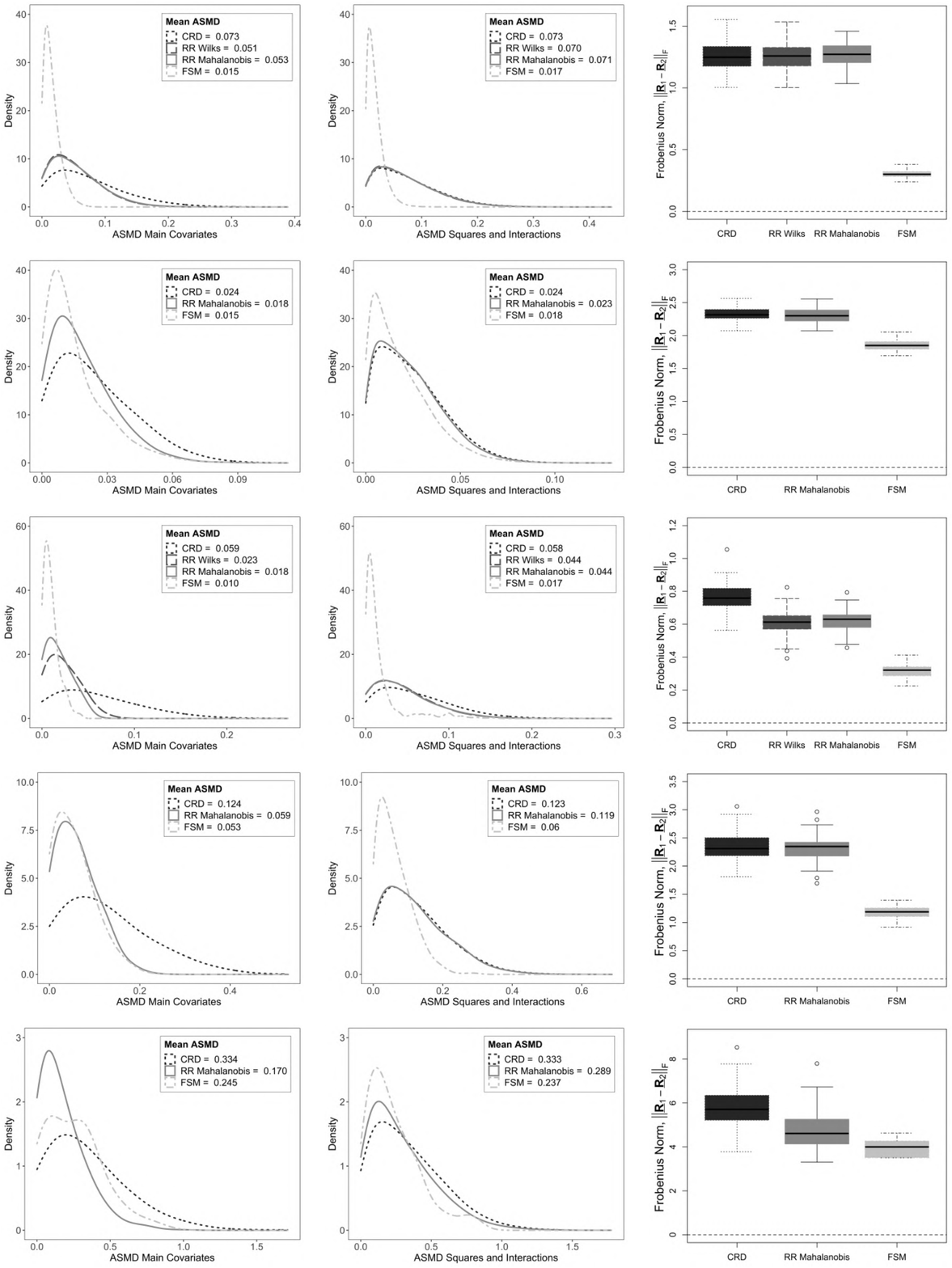}
\caption{Distributions of the absolute standardized mean differences of the main covariates and their squares and interactions, and the Frobenius norms of $\underline{\bm{R}}_1 - \underline{\bm{R}}_2$ under complete randomization, rerandomization, and the FSM, for the five studies: (6) Ambler, (7) Crepon, (8) Dupas, (9) Karlan, (10) Wantchekon.}
\label{fig:five_studies2}
\end{figure}

%% file: main.bbl
\begin{thebibliography}{35}
\newcommand{\enquote}[1]{``#1''}
\expandafter\ifx\csname natexlab\endcsname\relax\def\natexlab#1{#1}\fi

\bibitem[{Ambler et~al.(2015)Ambler, Aycinena, and Yang}]{ambler2015channeling}
Ambler, K., Aycinena, D., and Yang, D. (2015), \enquote{Channeling remittances
  to education: A field experiment among migrants from El Salvador,}
  \textit{American Economic Journal: Applied Economics}, 7, 207--32.

\bibitem[{Angrist and Lavy(1999)}]{angrist1999using}
Angrist, J.~D. and Lavy, V. (1999), \enquote{Using Maimonides' rule to estimate
  the effect of class size on scholastic achievement,} \textit{The Quarterly
  Journal of Economics}, 114, 533--575.

\bibitem[{Aron-Dine et~al.(2013)Aron-Dine, Einav, and
  Finkelstein}]{aron2013rand}
Aron-Dine, A., Einav, L., and Finkelstein, A. (2013), \enquote{The RAND health
  insurance experiment, three decades later,} \textit{Journal of Economic
  Perspectives}, 27, 197--222.

\bibitem[{Banerjee et~al.(2017)Banerjee, Chassang, and
  Snowberg}]{banerjee2017decision}
Banerjee, A.~V., Chassang, S., and Snowberg, E. (2017), \enquote{Decision
  theoretic approaches to experiment design and external validity,} in
  \textit{Handbook of Economic Field Experiments}, Elsevier, vol.~1, pp.
  141--174.

\bibitem[{Bertsimas et~al.(2015)Bertsimas, Johnson, and
  Kallus}]{bertsimas2015power}
Bertsimas, D., Johnson, M., and Kallus, N. (2015), \enquote{The power of
  optimization over randomization in designing experiments involving small
  samples,} \textit{Operations Research}, 63, 868--876.

\bibitem[{Blattman and Dercon(2018)}]{blattman2018impacts}
Blattman, C. and Dercon, S. (2018), \enquote{The impacts of industrial and
  entrepreneurial work on income and health: Experimental evidence from
  Ethiopia,} \textit{American Economic Journal: Applied Economics}, 10, 1--38.

\bibitem[{Brook et~al.(2006)Brook, Keeler, Lohr, Newhouse, Ware, Rogers,
  Davies, Sherbourne, Goldberg, Camp, et~al.}]{brook2006health}
Brook, R.~H., Keeler, E.~B., Lohr, K.~N., Newhouse, J.~P., Ware, J.~E., Rogers,
  W.~H., Davies, A.~R., Sherbourne, C.~D., Goldberg, G.~A., Camp, P., et~al.
  (2006), \enquote{The health insurance experiment: a classic RAND study speaks
  to the current health care reform debate,} \textit{Santa Monica, CA: RAND
  Corporation}.

\bibitem[{Chattopadhyay et~al.(2021)Chattopadhyay, Morris, and
  Zubizarreta}]{chattopadhyay2021randomized}
Chattopadhyay, A., Morris, C.~N., and Zubizarreta, J.~R. (2021),
  \enquote{Randomized and Balanced Allocation of Units into Treatment Groups
  Using the Finite Selection Model for R,} \textit{arXiv preprint
  arXiv:2105.02393}.

\bibitem[{Cochran and Cox(1957)}]{cochran1957experimental}
Cochran, W. and Cox, G. (1957), \textit{Experimental Designs}, John Wiley \&
  Sons New York.

\bibitem[{Cox and Reid(2000)}]{cox2000theory}
Cox, D.~R. and Reid, N. (2000), \textit{The Theory of the Design of
  Experiments}, CRC Press.

\bibitem[{Cr{\'e}pon et~al.(2015)Cr{\'e}pon, Devoto, Duflo, and
  Parient{\'e}}]{crepon2015estimating}
Cr{\'e}pon, B., Devoto, F., Duflo, E., and Parient{\'e}, W. (2015),
  \enquote{Estimating the impact of microcredit on those who take it up:
  Evidence from a randomized experiment in Morocco,} \textit{American Economic
  Journal: Applied Economics}, 7, 123--50.

\bibitem[{Dupas et~al.(2016)Dupas, Hoffmann, Kremer, and
  Zwane}]{dupas2016targeting}
Dupas, P., Hoffmann, V., Kremer, M., and Zwane, A.~P. (2016),
  \enquote{Targeting health subsidies through a nonprice mechanism: A
  randomized controlled trial in Kenya,} \textit{Science}, 353, 889--895.

\bibitem[{Durocher et~al.(2019)Durocher, Dzuba, Carroli, Morales, Aguirre,
  Martin, Esquivel, Carroli, and Winikoff}]{durocher2019does}
Durocher, J., Dzuba, I.~G., Carroli, G., Morales, E.~M., Aguirre, J.~D.,
  Martin, R., Esquivel, J., Carroli, B., and Winikoff, B. (2019), \enquote{Does
  route matter? Impact of route of oxytocin administration on postpartum
  bleeding: A double-blind, randomized controlled trial,} \textit{PloS one},
  14, e0222981.

\bibitem[{Finkelstein et~al.(2020)Finkelstein, Zhou, Taubman, and
  Doyle}]{finkelstein2020health}
Finkelstein, A., Zhou, A., Taubman, S., and Doyle, J. (2020), \enquote{Health
  care hotspotting—a randomized, controlled trial,} \textit{New England
  Journal of Medicine}, 382, 152--162.

\bibitem[{Fisher(1925)}]{fisher1925statistical}
Fisher, R.~A. (1925), \enquote{Statistical methods for research workers, 13e,}
  \textit{London: Oliver and Loyd, Ltd}, 99--101.

\bibitem[{Fisher(1935)}]{fisher1935design}
--- (1935), \textit{The Design of Experiments}, London: Oliver \& Boyd.

\bibitem[{Fujiwara and Wantchekon(2013)}]{fujiwara2013can}
Fujiwara, T. and Wantchekon, L. (2013), \enquote{Can informed public
  deliberation overcome clientelism? Experimental evidence from Benin,}
  \textit{American Economic Journal: Applied Economics}, 5, 241--55.

\bibitem[{Greevy et~al.(2004)Greevy, Lu, Silber, and
  Rosenbaum}]{greevy2004optimal}
Greevy, R., Lu, B., Silber, J.~H., and Rosenbaum, P.~R. (2004),
  \enquote{Optimal multivariate matching before randomization,}
  \textit{Biostatistics}, 5, 263--275.

\bibitem[{Hainmueller(2012)}]{hainmueller2012balancing}
Hainmueller, J. (2012), \enquote{Entropy balancing for causal effects: a
  multivariate reweighting method to produce balanced samples in observational
  studies,} \textit{Political Analysis}, 20, 25--46.

\bibitem[{Imbens and Rubin(2015)}]{imbens2015causal}
Imbens, G.~W. and Rubin, D.~B. (2015), \textit{Causal Inference in Statistics,
  Social, and Biomedical Sciences}, Cambridge University Press.

\bibitem[{Karlan et~al.(2014)Karlan, Osei, Osei-Akoto, and
  Udry}]{karlan2014agricultural}
Karlan, D., Osei, R., Osei-Akoto, I., and Udry, C. (2014),
  \enquote{Agricultural decisions after relaxing credit and risk constraints,}
  \textit{The Quarterly Journal of Economics}, 129, 597--652.

\bibitem[{Krieger et~al.(2019)Krieger, Azriel, and
  Kapelner}]{krieger2019nearly}
Krieger, A.~M., Azriel, D., and Kapelner, A. (2019), \enquote{Nearly random
  designs with greatly improved balance,} \textit{Biometrika}, 106, 695--701.

\bibitem[{LaLonde(1986)}]{lalonde1986evaluating}
LaLonde, R.~J. (1986), \enquote{Evaluating the econometric evaluations of
  training programs with experimental data,} \textit{The American Economic
  Review}, 604--620.

\bibitem[{Lock(2011)}]{lock2011rerandomization}
Lock, K.~F. (2011), \textit{Rerandomization to Improve Covariate Balance in
  Randomized Experiments}, Harvard University.

\bibitem[{Morgan and Rubin(2012)}]{morgan2012rerandomization}
Morgan, K.~L. and Rubin, D.~B. (2012), \enquote{Rerandomization to improve
  covariate balance in experiments,} \textit{Annals of Statistics}, 40,
  1263--1282.

\bibitem[{Morris(1979)}]{morris1979finite}
Morris, C. (1979), \enquote{A finite selection model for experimental design of
  the health insurance study,} \textit{Journal of Econometrics}, 11, 43--61.

\bibitem[{Morris(1983)}]{morris1983sequentially}
--- (1983), \enquote{Sequentially controlled Markovian random sampling
  (SCOMARS),} \textit{Institute of Mathematical Statistics Bulletin}, 12, 237.

\bibitem[{Morris and Hill(2000)}]{morris1993the}
Morris, C. and Hill, J. (2000), \enquote{The health insurance experiment:
  design using the finite selection model,} \textit{Public Policy and
  Statistics: Case Studies from RAND}, Springer Science \& Business Media,
  29--53.

\bibitem[{Newhouse et~al.(1993)}]{newhouse1993free}
Newhouse, J.~P. et~al. (1993), \textit{Free for All?}, Harvard University
  Press.

\bibitem[{Neyman(1923, 1990)}]{neyman1923application}
Neyman, J. (1923, 1990), \enquote{On the application of probability theory to
  agricultural experiments,} \textit{Statistical Science}, 5, 463--480.

\bibitem[{Rosenbaum(2002)}]{rosenbaum2002observational}
Rosenbaum, P.~R. (2002), \textit{Observational Studies}, Springer.

\bibitem[{Rosenbaum(2010)}]{rosenbaum2010design2}
--- (2010), \enquote{Design sensitivity and efficiency in observational
  studies,} \textit{Journal of the American Statistical Association}, 105,
  692--702.

\bibitem[{Rosenbaum and Rubin(1985)}]{rosenbaum1985constructing}
Rosenbaum, P.~R. and Rubin, D.~B. (1985), \enquote{Constructing a control group
  using multivariate matched sampling methods that incorporate the propensity
  score,} \textit{The American Statistician}, 39, 33--38.

\bibitem[{Rubin(1974)}]{rubin1974estimating}
Rubin, D.~B. (1974), \enquote{Estimating causal effects of treatments in
  randomized and nonrandomized studies.} \textit{Journal of Educational
  Psychology}, 66, 688.

\bibitem[{Stuart(2010)}]{stuart2010matching}
Stuart, E.~A. (2010), \enquote{Matching methods for causal inference: a review
  and a look forward,} \textit{Statistical Science}, 25, 1--21.

\end{thebibliography}
